\documentclass[reqno]{amsart}

\usepackage{amsmath}
\usepackage{amssymb}
\usepackage{amsfonts}
\usepackage{amsthm}
\usepackage{mathtools}
\mathtoolsset{%
}
\usepackage{mathabx}
\newcommand{\debug}[1]{#1}

\usepackage[utf8]{inputenc}
\usepackage[T1]{fontenc}
\usepackage[sf,mono=false]{libertine}
\usepackage{dsfont}
\usepackage{multicol}
\usepackage[%
cal=cm,
]
{mathalfa}

\usepackage[dvipsnames,svgnames]{xcolor}
\colorlet{MyBlue}{DodgerBlue!60!Black}
\colorlet{MyGreen}{DarkGreen!85!Black}

\usepackage[font=small,labelfont=bf]{caption}
\usepackage{subfigure}
\usepackage{tikz}
\usetikzlibrary{calc,patterns}

\usepackage{acronym}
\usepackage{booktabs}       % professional-quality tables
\usepackage{latexsym}
\usepackage{paralist}
\usepackage{wasysym}
\usepackage{xspace}
\usepackage{comment}
\usepackage{longtable}
\usepackage[shortlabels]{enumitem}

\usepackage{multirow,array}
\usepackage{tabu}

\usepackage{relsize}

\usepackage{scalerel,stackengine}
\stackMath
\newcommand\reallywidehat[1]{%
\savestack{\tmpbox}{\stretchto{%
  \scaleto{%
    \scalerel*[\widthof{\ensuremath{#1}}]{\kern.1pt\mathchar"0362\kern.1pt}%
    {\rule{0ex}{\textheight}}%WIDTH-LIMITED CIRCUMFLEX
  }{\textheight}% 
}{4.4ex}}%
\stackon[-11.9pt]{#1}{\tmpbox}%
}
\usepackage[american]{babel}

\usepackage[authoryear,compress,comma]{natbib}
\usepackage{hyperref}
\hypersetup{
colorlinks=true,
linktocpage=true,
pdfstartview=FitH,
breaklinks=true,
pdfpagemode=UseNone,
pageanchor=true,
pdfpagemode=UseOutlines,
plainpages=false,
bookmarksnumbered,
bookmarksopen=false,
bookmarksopenlevel=1,
hypertexnames=true,
pdfhighlight=/O,
urlcolor=MyBlue!60!black,linkcolor=MyBlue!70!black,citecolor=DarkGreen!70!black, % <--- for screen
pdftitle={},
pdfauthor={},
pdfsubject={},
pdfkeywords={},
pdfcreator={pdfLaTeX},
pdfproducer={LaTeX with hyperref}
}

\numberwithin{equation}{section}  %numberwithin goes before cleverefs when using hyperref
\usepackage[sort&compress,capitalize,nameinlink]{cleveref}
\crefrangeformat{equation}{\upshape(#3#1#4)\textendash(#5#2#6)}

\newcommand{\reduc}[1]{\widehat#1}
\newcommand{\dreduc}[1]{\Hat{\Hat#1}}

\newcommand{\argdot}{\,\cdot\,}

\newcommand{\R}{\mathbb{R}}

\newcommand{\game}{\boldsymbol{\debug g}}
\newcommand{\gamealt}{\game^{\prime}}

\newcommand{\games}{\mathcal{\debug G}}
\newcommand{\subsgames}{\mathcal{\debug X}}
\newcommand{\act}{\debug s}
\newcommand{\actalt}{\debug t}
\newcommand{\actr}{\debug r}
\newcommand{\actprof}{\boldsymbol{\act}}
\newcommand{\actaltprof}{\boldsymbol{\actalt}}
\newcommand{\actrprof}{\boldsymbol{\actr}}
\newcommand{\actions}{\mathcal{\debug S}}
\newcommand{\actionsalt}{\actions^{\prime}}
\newcommand{\mixed}{\debug x}
\newcommand{\mixedalt}{\debug y}
\newcommand{\mixedprof}{\boldsymbol{\mixed}}

\newcommand{\nplayers}{\debug n}
\newcommand{\players}{\mathcal{\debug N}}
\newcommand{\play}{\debug i}
\newcommand{\playalt}{\debug j}
\newcommand{\payoff}{\debug g}
\newcommand{\equils}{\mathcal{\debug E}}
\newcommand{\potent}{\debug \Psi}

\DeclareMathOperator{\NS}{\mathsf{\debug{NS}}}
\DeclareMathOperator{\No}{\mathsf{\debug{No}}}
\DeclareMathOperator{\Pot}{\mathsf{\debug{P}}}
\DeclareMathOperator{\Har}{\mathsf{\debug{H}}}
\DeclareMathOperator{\NSG}{\mathsf{\debug{NSG}}}
\DeclareMathOperator{\NoG}{\mathsf{\debug{NoG}}}

\DeclareMathOperator{\NoGd}{\mathsf{\debug{-NoG}}}
\DeclareMathOperator{\PGd}{\mathsf{\debug{-PG}}}
\DeclareMathOperator{\HGd}{\mathsf{\debug{-HG}}}

\DeclareMathOperator{\NSd}{\mathsf{\debug{-NS}}}

\DeclareMathOperator{\Potd}{\mathsf{\debug{-P}}}
\DeclareMathOperator{\Hard}{\mathsf{\debug{-H}}}

\newcommand{\simplex}{\mathlarger{\debug \bigtriangleup}}
\newcommand{\measures}{\mathcal{\debug M}}
\newcommand{\bmeas}{\debug b}
\newcommand{\betameas}{\debug \beta}
\newcommand{\cmeas}{\debug {\debug c}}
\newcommand{\gammameas}{\debug \gamma}
\newcommand{\mumeas}{\debug \mu}
\newcommand{\numeas}{\debug \nu}
\newcommand{\thetameas}{\debug \theta}
\newcommand{\bprof}{\boldsymbol{\bmeas}}
\newcommand{\cprof}{\boldsymbol{\cmeas}}
\newcommand{\muprof}{\boldsymbol{\mumeas}}

\newcommand{\setA}{\debug A}
\newcommand{\mapT}{\debug T}

\newcommand{\maps}{\mathcal{\debug T}}
\newcommand{\scala}{\debug \eta}
\newcommand{\scalb}{\debug \theta}

\newcommand{\norml}[1]{\overline#1}
\newcommand{\permut}{\debug \sigma}
\newcommand{\nonstratf}{\debug \ell}

\newcommand{\Czero}{\debug{C_{0}}}
\newcommand{\Cone}{\debug{C_{1}}}
\newcommand{\func}{\debug h}
\newcommand{\funcalt}{\debug f}
\newcommand{\funcvec}{\boldsymbol{\func}}
\newcommand{\threedecmap}{\debug \zeta}
\newcommand{\threedecmapalt}{\threedecmap^{\prime}}
\newcommand{\threedecmaps}{\mathcal{\debug Z}}

\newcommand{\nsgameb}{\boldsymbol{\betameas}}
\newcommand{\nsgameg}{\boldsymbol{\gammameas}}
\newcommand{\varscal}{\debug \beta}

\newcommand{\opLam}{\debug \Lambda}
\newcommand{\opPi}{\debug \Pi}
\newcommand{\flow}{\debug X}
\newcommand{\flowalt}{\debug Y}
\newcommand{\compar}{\debug E}
\newcommand{\graph}{\debug \Gamma}
\newcommand{\symmfunc}{\debug W}
\newcommand{\gradop}{\debug \delta}
\newcommand{\basis}{\debug \varepsilon}
\newcommand{\embop}{\debug D}
\newcommand{\funcphi}{\debug \varphi}
\newcommand{\funcpsi}{\debug \psi}
\newcommand{\subactions}{\mathcal{\debug Y}}
\newcommand{\weight}{\debug w}
\newcommand{\weightprof}{\boldsymbol{\weight}}
\newcommand{\ared}{\debug \alpha}
\newcommand{\dggp}{\debug d}

\DeclareMathOperator{\dup}{\mathsf{\debug{dup}}}

\DeclareMathOperator{\red}{\mathsf{\debug{red}}}
\DeclareMathOperator{\aredred}{\ared-{\red}}

\DeclareMathOperator{\ex}{\mathbb{E}}

\DeclareMathOperator{\Ima}{\mathsf{\debug{Im}}}
\DeclareMathOperator{\Id}{\mathsf{\debug{Id}}}

\DeclareMathOperator{\Ker}{\mathsf{\debug{Ker}}}

\DeclareMathOperator{\prob}{\mathbb{P}}

\DeclareMathOperator{\supp}{supp}

\DeclarePairedDelimiter{\braces}{\{}{\}}
\DeclarePairedDelimiter{\bracks}{[}{]}
\DeclarePairedDelimiter{\parens}{(}{)}

\DeclarePairedDelimiter{\abs}{\lvert}{\rvert}
\DeclarePairedDelimiter{\norm}{\lVert}{\rVert}

\DeclarePairedDelimiterX{\braket}[2]{\langle}{\rangle}{#1\mathopen{}\hspace{1pt}\delimsize\vert\hspace{1pt}\mathopen{}#2}

\DeclarePairedDelimiterX{\inner}[2]{\langle}{\rangle}{#1,#2}
\DeclarePairedDelimiterX{\setdef}[2]{\{}{\}}{#1:#2}

\DeclarePairedDelimiterXPP{\probof}[1]{\prob}{(}{)}{}{%

#1}

\DeclarePairedDelimiterXPP{\exof}[1]{\ex}{[}{]}{}{%

#1}

\usepackage[textwidth=30mm]{todonotes}
\usepackage{soul}
\setstcolor{red}
\sethlcolor{SkyBlue}

\theoremstyle{plain}
\newtheorem{theorem}{Theorem}
\newtheorem{corollary}[theorem]{Corollary}
\newtheorem*{corollary*}{Corollary}
\newtheorem{lemma}[theorem]{Lemma}
\newtheorem{proposition}[theorem]{Proposition}

\theoremstyle{definition}
\newtheorem{definition}[theorem]{Definition}
\newtheorem*{definition*}{Definition}

\newtheorem*{hypothesis*}{Hypothesis}

\theoremstyle{remark}
\newtheorem{remark}[theorem]{Remark}
\newtheorem*{remark*}{Remark}
\newtheorem*{notation*}{Notational remark}
\newtheorem{example}[theorem]{Example}

\numberwithin{theorem}{section}
\newacro{GT}[GT]{game transformation}
\newacro{NC}[NC]{Nash-consistent}
\newacro{NP}[NP]{Nash-preserving}
\newacro{3DM}[$3$-DM]{$3$-decomposition map}

\newacro{NE}{Nash equilibrium}
\newacroplural{NE}[NE]{Nash equilibria}
\newacro{BNE}{Bayesian Nash equilibrium}
\newacroplural{BNE}[BNE]{Bayesian Nash equilibria}
\newacro{PNE}{pure Nash equilibrium}
\newacroplural{PNE}[PNE]{pure Nash equilibria}
\newacro{MNE}{mixed Nash equilibrium}
\newacroplural{MNE}[MNE]{mixed Nash equilibria}

\begin{document}

%----------------------------------------------------------------------
%%% TITLE & AUTHORS
%----------------------------------------------------------------------
\title
[Decomposition of games: some strategic considerations] 
{Decomposition of games: some strategic considerations}

\author
[J.~Abdou]
{Joseph Abdou}
\address{Centre d'\'Economie de la Sorbonne, Universit\'e Paris 1, Panth\'eon-Sorbonne, 106-112 boulevard de l'H\^opital, 75647 Paris Cedex 13, France; email: abdou@univ-paris1.fr.}

\author
[N.~Pnevmatikos]
{Nikolaos Pnevmatikos}
\address{Lemma, Universit\'e Paris 2, Panth\'eon-Assas, 4 Rue Desgoffe, 75006, Paris, France; email: nikolaos.pnevmatikos@u-paris2.fr.}

\author
[M.~Scarsini]
{Marco Scarsini}
\address{Dipartimento di Economia e Finanza, LUISS, Viale Romania 32, 00197 Rome, Italy, email: marco.scarsini@luiss.it.}

\author
[X.~Venel]
{Xavier Venel}
\address{Ecole d'\'Economie de Paris, Universit\'e Paris 1, Panth\'eon-Sorbonne, 106-112 boulevard de l'H\^opital, 75647 Paris Cedex 13, France; email: xavier.venel@univ-paris1.fr}

\thanks{
Nikos Pnevmatikos's research was supported by Labex MME-DII. Part of this research was carried out when he was visiting the Engineering Systems and Design pillar at Singapore University of Technology and Design.
Marco Scarsini is a member of GNAMPA-INdAM.
His work has been partly supported by the INdAM- GNAMPA Project 2019 ``Markov chains and games on networks,'' and by the Italian MIUR PRIN 2017 Project ALGADIMAR ``Algorithms, Games, and Digital Markets.''
Xavier Venel acknowledges the support of the Agence Nationale de la Recherche [ANR CIGNE ANR-15-CE38-0007-01].
This work was partially supported by GAMENET  COST Action CA 16228.
The authors also want to thank Ozan Candogan, Sung-Ha Hwang and Panayotis Mertikopoulos for valuable comments.}

\date{\today}

%----------------------------------------------------------------------
%%% KEYWORDS
%----------------------------------------------------------------------


\begin{abstract}
\citet{CanMenOzdPar:MOR2011} provide an orthogonal direct-sum decomposition of finite games into potential, harmonic and nonstrategic components. 
In this paper we study the issue of decomposing games that are strategically equivalent from a game-theoretical point of view, for instance games obtained via transformations such as duplications of strategies or positive affine mappings of of payoffs. 
We show the need to define classes of decompositions to achieve commutativity of game transformations and decompositions.
\end{abstract} 

\maketitle

%----------------------- SECTION ---------------------------------

\section{Introduction}

Potential games are an interesting class of games that admit pure Nash equilibria and behave well with respect to the most common learning procedures. 
Some games, although they are not potential games, are close---in a sense to be made precise---to a potential game. It is therefore interesting to examine whether their equilibria are close to the equilibria of the potential game, (see \citet{CanOzdPar:GEB2013}). 
With this in mind, in their seminal paper \citet{CanMenOzdPar:MOR2011} were able to show that the class of strategic-form games having a fixed set of players and a fixed set of strategies for each player is a linear space that can be decomposed into the orthogonal sum of three components, called  the \emph{potential}, \emph{harmonic} and \emph{nonstrategic} component. 
Games  in the harmonic component have a completely mixed equilibrium where all players mix uniformly over their strategies; games in the nonstrategic component are such that the payoff of each player is not affected by her own strategy, but only by other players' strategies.
To achieve this decomposition the authors associate to each game a graph where vertices are strategy profiles and edges connect profiles that differ only for the strategy of one player. The analysis is then carried out by studying flows on graphs and using the  Helmholtz decomposition theorem.

The decomposition of  \citet{CanMenOzdPar:MOR2011} refers to games having all the same set of players and the same set of strategies for each player. 
In their construction nothing connects the decomposition of a specific game $\game$ with the decomposition of another game $\widecheck{\game}$ that is obtained from $\game$ by adding a strategy to the set $\actions^{\play}$ of player $\play$'s feasible strategies.
One may argue that this is reasonable, since the two games live in linear spaces of different dimension and  the new game with an extra strategy may have equilibria that are very different from the ones in the original game, so, in general,  the two games may have very little in common. 
In some situations, though, the two games are indeed strongly related. 
For instance, consider the case where the payoffs corresponding to the new strategy  are just a replica of the payoffs of another strategy. 
In this case, from a strategic viewpoint, the two games $\widecheck{\game}$ and $\game$ are actually the same game and every equilibrium in $\widecheck{\game}$ can be mapped to an equilibrium in $\game$.
It would be reasonable to expect the decomposition of $\game$ and $\widecheck{\game}$ to be strongly related. 
Unfortunately this is not the case.
Consider for instance the matching-pennies game $\game$
\begin{center}
\medskip
{\tabulinesep=1.2mm
\begin{tabu}{ r|cc|cc| }
\multicolumn{1}{r}{}
 &  \multicolumn{2}{c}{$\act^{2}$}
 & \multicolumn{2}{c}{$\actalt^{2}$} \\
\cline{2-5}
$\act^{1}$ & $1$ & $-1$ & $-1$ & $1$ \\
\cline{2-5}
$\actalt^{1}$ & $-1$ & $1$ & $1$ & $-1$ \\
\cline{2-5}
\multicolumn{5}{c}{$\game$}
\end{tabu}}
\medskip
\end{center}
and the game $\widecheck{\game}$, obtained by replicating strategy $\actalt^{1}$ of the row-player
\begin{center}
\medskip
{\tabulinesep=1.2mm
\begin{tabu}{ r|cc|cc| }
\multicolumn{1}{r}{}
 &  \multicolumn{2}{c}{$\act^{2}$}
 & \multicolumn{2}{c}{$\actalt^{2}$} \\
\cline{2-5}
$\act^{1}$ & $1$ & $-1$ & $-1$ & $1$ \\
\cline{2-5}
$\actalt_{0}^{1}$ & $-1$ & $1$ & $1$ & $-1$ \\
\cline{2-5}
$\actalt_{1}^{1}$ & $-1$ & $1$ & $1$ & $-1$ \\
\cline{2-5}
\multicolumn{5}{c}{$\widecheck{\game}$}
\end{tabu}} .
\medskip
\end{center}

The matching-pennies game $\game$  admits a unique Nash equilibrium where each player randomizes uniformly between the two available strategies. 
This game is harmonic, so its decomposition has the potential and nonstrategic component identically equal to zero and the harmonic component $\game_{\Har}$ equal to $\game$.

The game $\widecheck{\game}$ admits a continuum of equilibria where the column player mixes uniformly and the row-player mixes $(1/2,\mixed/2,(1-\mixed)/2)$ with $\mixed\in[0,1]$. 
Notice that in each of these mixed equilibria the mixed strategy of the row-player assigns probability $1/2$ to $\actalt_{0}^{1}\cup \actalt_{1}^{1}$. The decomposition result of \citet[Theorem~4.1]{CanMenOzdPar:MOR2011} applied to $\widecheck{\game}$ yields

\begin{center}
\medskip
{\tabulinesep=1.2mm
\begin{tabu}{ r|cc|cc| }
\multicolumn{1}{r}{}
 &  \multicolumn{2}{c}{$\act^{2}$}
 & \multicolumn{2}{c}{$\actalt^{2}$} \\
\cline{2-5}
$\act^{1}$ & $4/15$ & $3/5$ & $-4/15$ & $-3/5$ \\
\cline{2-5}
$\actalt_{0}^{1}$ & $-2/15$ & $2/10$ & $2/15$ & $-2/10$ \\
\cline{2-5}
$\actalt_{1}^{1}$ & $-2/15$ & $2/10$ & $2/15$ & $-2/10$ \\
\cline{2-5}
\multicolumn{5}{c}{$(\widecheck{\game})_{\Pot}$}
\end{tabu}}
\qquad
{\tabulinesep=1.2mm\begin{tabu}{ r|cc|cc| }
\multicolumn{1}{r}{}
 &  \multicolumn{2}{c}{$\act^{2}$}
 & \multicolumn{2}{c}{$\actalt^{2}$} \\
\cline{2-5}
$\act^{1}$ & $16/15$ & $-8/5$ & $-16/15$ & $8/5$ \\
\cline{2-5}
$\actalt_{0}^{1}$ & $-8/15$ & $4/5$ & $8/15$ & $-4/5$ \\
\cline{2-5}
$\actalt_{1}^{1}$ & $-8/15$ & $4/5$ & $8/15$ & $-4/5$ \\
\cline{2-5}
\multicolumn{5}{c}{$(\widecheck{\game})_{\Har}$}
\end{tabu}}
\medskip
\end{center}

\begin{center}
{\tabulinesep=1.2mm
\begin{tabu}{ r|cc|cc| }
\multicolumn{1}{r}{}
&  \multicolumn{2}{c}{$\act^{2}$}
 & \multicolumn{2}{c}{$\actalt^{2}$} \\
\cline{2-5}
$\act^{1}$ & $-1/3$ & $0$ & $1/3$ & $0$ \\
\cline{2-5}
$\actalt_{0}^{1}$ & $-1/3$ & $0$ & $1/3$ & $0$ \\
\cline{2-5}
$\actalt_{1}^{1}$ & $-1/3$ & $0$ & $1/3$ & $0$ \\
\cline{2-5}
\multicolumn{5}{c}{$(\widecheck{\game})_{\NS}$}
\end{tabu}} ,
\medskip
\end{center}
where $(\widecheck{\game})_{\Pot}$, $(\widecheck{\game})_{\Har}$, and $(\widecheck{\game})_{\NS}$ are the potential, harmonic, and nonstrategic components of $\widecheck{\game}$, respectively.

We see that, although the two games $\game$ and $\widecheck{\game}$ are strategically equivalent in the sense described before, their decompositions are quite different.  In the sequel, when considering a game $\game$ with duplicate strategies, we will call the game where duplication of strategies has been eliminated the \emph{reduced version} of $\game$ and we will use the notation $\reduc{\game}$ for it.

A similar problem appears in games whose payoffs are suitable affine game transformation of some other game's payoffs.
For instance, given a game $\game$,  consider the game $\widetilde{\game}$ which is obtained by multiplying the payoffs of each player in $\game$ by the same positive constant. Multiplying the payoff of a player by a positive constant is innocuous with respect to strategic considerations. 
To illustrate this, if $\game$ is  the matching-pennies game,  consider the  game $\widetilde{\game}$ where the payoffs of the row-player in $\game$ have been multiplied by $2$.  

\begin{center}
\medskip
{\tabulinesep=1.2mm
\begin{tabu}{ r|cc|cc| }
\multicolumn{1}{r}{}
 &  \multicolumn{2}{c}{$\act^{2}$}
 & \multicolumn{2}{c}{$\actalt^{2}$} \\
\cline{2-5}
$\act^{1}$ & $2$ & $-1$ & $-2$ & $1$ \\
\cline{2-5}
$\actalt^{1}$ & $-2$ & $1$ & $2$ & $-1$ \\
\cline{2-5}
\multicolumn{5}{c}{$\widetilde{\game}$}
\end{tabu}} .
\medskip
\end{center}

This game admits a unique equilibrium where each player plays uniformly over her set of strategies. The decomposition result of \citet{CanMenOzdPar:MOR2011} applied in $\widetilde{\game}$ yields the following decomposition:

\begin{center}
\medskip
{\tabulinesep=1.2mm
\begin{tabu}{ r|cc|cc| }
\multicolumn{1}{r}{}
 &  \multicolumn{2}{c}{$\act^{2}$}
 & \multicolumn{2}{c}{$\actalt^{2}$} \\
\cline{2-5}
$\act^{1}$ & $1/2$ & $1/2$ & $-1/2$ & $-1/2$ \\
\cline{2-5}
$\actalt^{1}$ & $-1/2$ & $-1/2$ & $1/2$ & $1/2$ \\
\cline{2-5}
\multicolumn{5}{c}{$(\widetilde{\game})_{\Pot}$}
\end{tabu}
\qquad
\begin{tabu}{ r|cc|cc| }
\multicolumn{1}{r}{}
 &  \multicolumn{2}{c}{$\act^{2}$}
 & \multicolumn{2}{c}{$\actalt^{2}$} \\
\cline{2-5}
$\act^{1}$ & $3/2$ & $-3/2$ & $-3/2$ & $3/2$ \\
\cline{2-5}
$\actalt^{1}$ & $-3/2$ & $3/2$ & $3/2$ & $-3/2$ \\
\cline{2-5}
\multicolumn{5}{c}{$(\widetilde{\game})_{\Har}$}
\end{tabu}} .
\medskip
\end{center}
Although the games $\game$ and $\widetilde{\game}$ share the same Nash equilibrium set, they admit different decompositions. 
Notice that $\game$ and $(\widetilde{\game})_{\Har}$ are both harmonic games and admit the same unique mixed equilibrium, since they are related by an affine game transformation that does not affect the harmonic property: one is obtained from the other by multiplying all payoffs by the same positive constant.

In the sequel, we will be  interested in affine game transformations where payoffs are multiplied by a positive constant that depends not only on the players but also on their strategies. 
We will refer to this kind of game transformations as \emph{scalings}.  
We illustrate this with an example. 
Let $\game$ be the matching-pennies game and let $\overline{\game}$ be the game where the payoffs of the row-player in $\game$ have been multiplied by $2$ only in the first column. 

\begin{center}
\medskip
{\tabulinesep=1.2mm
\begin{tabu}{ r|cc|cc| }
\multicolumn{1}{r}{}
 &  \multicolumn{2}{c}{$\act^{2}$}
 & \multicolumn{2}{c}{$\actalt^{2}$} \\
\cline{2-5}
$\act^{1}$ & $2$ & $-1$ & $-1$ & $1$ \\
\cline{2-5}
$\actalt^{1}$ & $-2$ & $1$ & $1$ & $-1$ \\
\cline{2-5}
\multicolumn{5}{c}{$\overline{\game}$}
\end{tabu}} .
\medskip
\end{center}
The game $\overline{\game}$ is the image  of $\game$  by $\mapT$, the transformation that operates on two-person $2$-by-$2$ matrix games and multiplies the payoff of the row-player on the first column by $2$. This scaling is interesting since it transforms the Nash equilibria of any game in a way that is simple to describe. 

Any Nash equilibrium of any  $2$-by$2$  game $\gamealt$ is transformed into a Nash equilibrium of $\mapT(\gamealt)$ as follows: the strategy of the row player stays the same whereas the weight on the  strategy $\act^{2}$ of the column player is divided by two (and then the weights are normalized). 
For example, $\overline{\game}$ admits a unique equilibrium where the row-player plays uniformly over her set of strategies, whereas the column-player plays  $\act^{2}$ with probability $1/3$ and  $\actalt^{2}$ with probability $2/3$. 

The decomposition of \citet{CanMenOzdPar:MOR2011} applied to $\overline{\game}$ yields the following:

\begin{center}
\medskip
{\tabulinesep=1.2mm
\begin{tabu}{ r|cc|cc| }
\multicolumn{1}{r}{}
 &  \multicolumn{2}{c}{$\act^{2}$}
 & \multicolumn{2}{c}{$\actalt^{2}$} \\
\cline{2-5}
$\act^{1}$ & $3/4$ & $1/4$ & $1/4$ & $-1/4$ \\
\cline{2-5}
$\actalt^{1}$ & $-3/4$ & $-1/4$ & $-1/4$ & $1/4$ \\
\cline{2-5}
\multicolumn{5}{c}{$(\overline{\game})_{\Pot}$}
\end{tabu}
\qquad
\begin{tabu}{ r|cc|cc| }
\multicolumn{1}{r}{}
 &  \multicolumn{2}{c}{$\act^{2}$}
 & \multicolumn{2}{c}{$\actalt^{2}$} \\
\cline{2-5}
$\act^{1}$ & $5/4$ & $-5/4$ & $-5/4$ & $5/4$ \\
\cline{2-5}
$\actalt^{1}$ & $-5/4$ & $5/4$ & $5/4$ & $-5/4$ \\
\cline{2-5}
\multicolumn{5}{c}{$(\overline{\game})_{\Har}$}
\end{tabu}} .
\medskip
\end{center}

Hence, although the games $\game$ and $\overline{\game}$ are related by a simple transformation, their   decompositions are quite different.

The question that we want to address in this paper is the following: 
Consider a game $\game$ where the set of players is $\players$ and apply a transformation to it---such as duplication of a strategy or scaling of payoffs---and then decompose the transformed game;
now take the same game $\game$, first decompose it, and then apply a suitable transformation to the components of the decomposition.
Under which conditions do the two procedures produce the same result?
We saw that if we na\"{\i}vely use the same decomposition in the two cases, then we are doomed to fail.
Therefore, we need to consider a whole class of suitably parametrized decompositions and a class of game transformations that achieve the result.
When this occurs, we say that the class of transformations and the class of decompositions commute.
To achieve our goal, we adapt the approach of \citet{CanMenOzdPar:MOR2011} by introducing two parameters $\muprof$ and $\nsgameg$ in our decompositions.

The duplication example shows the need to introduce for every player $\play\in\players$ some weights on her own strategies. 
Indeed, after duplication, the weight on one strategy has to be split between the two strategies that duplicate it. 
Since the duplication of a strategy of some player affects only her own mixed strategy, for every player~$\play$, we introduce a measure $\mumeas^{\play}$ on her own strategies. 
We then obtain a vector $\muprof=\parens*{\mumeas^{\play}}_{\play\in\players}$.

The scaling example, on the other hand, shows that  changing the payoff of one player has an impact on the equilibrium strategy of the others players. 
Hence, for a player $\play$, we introduce a positive function $\gammameas^{\play}$, which we call \emph{co-measure} on the strategies of the other players.
Aggregating these functions into one parameter, we obtain the vector $\nsgameg=\parens{\gammameas^{\play}}_{\play\in\players}$.  

An equivalent way to describe what we do would be to define for every player $\play$, a  measure $\thetameas^\play$ over the set $\actions$ of strategy profiles such that for every action profile $\actprof\in\actions$, we have
\[
\thetameas^{\play}(\actprof)=\mumeas^{\play}(\act^{\play})\gammameas^{\play}(\actprof^{-\play}).
\]
Hence, each player has a measure on the  space $\actions$. 

We  then use the pair $(\muprof,\nsgameg)$ to parametrize a class of inner products that are in turn employed to define the decompositions and we generalize the decomposition result of \citet{CanMenOzdPar:MOR2011} using several metrics on the space of games induced by the pair $(\muprof,\nsgameg)$.

Then, given a decomposition of the game $\game$ based on the pair $(\muprof,\nsgameg)$, we show that there exists some new pair $(\widecheck{\muprof},\widecheck{\nsgameg})$ where $\widecheck{\muprof}$ depends only on $\muprof$ and  $\widecheck{\nsgameg}$  depends only on $\nsgameg$, such that replications of strategies and  decompositions commute.
The approach for scalings and other transformations is similar.

%----------------------- SUBSECTION ---------------------------------

\subsection{Related literature}\label{suse:related}

In the context of noncooperative game theory, several approaches have been proposed to decompose a game into simpler games that admit more flexible and attractive equilibrium analysis. 
\citet{San:GEB2010} proposes a method to decompose $\nplayers$-player normal-form games into $2^{\nplayers}$ simultaneously-played component games. As a by-product, this decomposition provides a characterization of normal-form potential games. 
\citet{KalKal:QJE2013} introduce a novel solution concept founded on a decomposition of two-player games into \emph{zero-sum} and \emph{common-interest} games. 
This decomposition result is based on the fact that all matrices can be decomposed into the sum of symmetric and antisymmetric matrices. 
\citet{SzaBodSam:PRE2017}, in order to study evolutionary dynamic games, refine this dyadic decomposition by further decomposing the antisymmetric component. 
In a different direction, \citet{JesSaa:UCI2013} present a \emph{strategic}-\emph{behavioral} decomposition of games with two strategies per player and highlight that certain solution concepts are determined by a game's strategic part or influenced by the behavioral portion. 
Their analysis has been expanded in \citet{JesSaa:Springer2019}.
More recently, \citet{HwaRey:arxiv2016} study the space of games as a Hilbert space and prove several decomposition theorems for arbitrary games identifying components such as potential games and games that are strategically equivalent---in the sense of sharing the same Nash equilibria---to zero-sum games. 
They further extend their results to games with uncountable strategy sets. 
Using this decomposition, they also provide an alternative proof for the well-known characterization of potential games presented by \citet{MonSha:GEB1996}. 

The paper of \citet{CanMenOzdPar:MOR2011} on the decomposition of finite games into potential, harmonic, and nonstrategic components is a milestone in the field. 
To achieve their decomposition, they represent an arbitrary game with an undirected graph where nodes stand for  strategy profiles and edges connect nodes that differ in the strategy of only one player. 
In our paper, we follow the same graph representation for a given game. 
\citet{Liu:ORL2018} uses a different graph representation for a finite game where nodes stand for players and edges connect players whose change of strategies influences the other player's payoff. 
Given a finite game, the author investigates necessary and sufficient conditions for the existence of a pure Nash equilibrium in terms of the structure of its associated directed graph. 
The decomposition in \citet{CanMenOzdPar:MOR2011} is based on the Helmholtz decomposition theorem---a fundamental tool in vector calculus---which states that any vector field can be decomposed into a \emph{divergence-free} and a \emph{curl-free} components.\footnote{A generalization of the Helmholtz theorem is known in the literature as the Hodge decomposition theorem, which is defined for differentiable forms on Riemannian manifolds.}  
Due to the ubiquitous nature of vector fields, this theorem has been applied by various research communities to a wide range of issues. 
In the context of discrete vector fields, \citet{JanLimYaoYe:MP2011} provide an implementation of the Helmholtz decomposition  in statistical ranking. 
\citet{SteTet:arXiv2017} apply the Helmholtz decomposition to cooperative games and obtain a novel characterization of the Shapley value in terms of the decomposition's components. 
Various papers related to \citet{CanMenOzdPar:MOR2011} have  appeared in the literature. 
For instance, \citet{LiuQiChe:CCC2015,LiLiuHeCheQiHon:IEEE2016} focus on the detailed description of the decomposition subspaces by providing some geometric and algebraic expressions and present an explicit formula for the decomposition. 
More recently, \citet{Zha:CCC2017} provides explicit polynomial expressions for the orthogonal projections onto the subspaces of potential and harmonic games, respectively.
\citet{LiCheHe:arXiv2019} study the compatibility of decompositions that are based on different inner products.

In the terminology of \citet{GovWil:PNAS2005}, two strategies of one player are equivalent if they yield every player the same expected payoff for each profile of other players' strategies. 
A pure strategy of player $\play$ is \emph{redundant} if player $\play$ has another equivalent strategy. 
In our paper, we study the behavior of the proposed decomposition with respect to redundant strategies and to suitable game transformations of payoff vectors that do not alter the strategic structure of the game. The issue of redundant strategies has been dealt with by \citet{GovWil:GEB1997,GovWil:PNAS2005,GovWil:Palgrave2008,GovWil:E2009} in the framework of equilibrium refinement. 
In particular, the authors show that the degree of a Nash component is invariant under addition or deletion of redundant strategies. 
As shown, for instance, by \citet{OsbRub:AER1998}, some solution concepts are not invariant with respect to addition of redundant strategies. 
In the framework of decomposition of games, \citet{KalKal:QJE2013} show that their decomposition is invariant to redundant strategies. \citet{CheLiuZhaQi:IEEE2016} provide a generalization of the decomposition of \citet{CanMenOzdPar:MOR2011} in terms of \emph{weighted potential} and \emph{weighted harmonic games}. 
This work is close to ours but still quite different. 
Precisely, \citet{CanMenOzdPar:MOR2011} assume that the weight of each player is equal to the number of her strategies while \citet{CheLiuZhaQi:IEEE2016} relax this hypothesis by considering any possible weight. 
Their approach is consistent with simple scalings like multiplication of the payoffs of some player by a constant that depends only on the other players, but not with more general scalings. 
Moreover, weighted harmonic games still admit the uniformly strategy profile as equilibrium and therefore their class is not robust to elimination of duplicate strategies. 
Constant-scalings were also implicitly used in \citet{CanOzdPar:CDC2010} when computing the closest weighted potential game.

%----------------------- SUBSECTION ---------------------------------

\subsection{Structure of the paper} 
In \cref{se:decomposition}, we introduce our decomposition results for games. 
In \cref{se:compati}, we provide a general framework of interesting game transformations and study the commutation of our decompositions and some specific classes of game transformations. 
All proofs can be found in  \cref{app:proofs}.

%----------------------- SECTION ---------------------------------

\section{Games transformations and decompositions}
\label{se:decomposition}

\subsection{Notation and definitions.}
\label{suse:notation}

We introduce some notation that will be used throughout the paper. 
Given a finite set $\setA$,  its cardinality is denoted by $\abs{\setA}$. 
Moreover, $\simplex(\setA)$ is the set of probability distributions over $\setA$, seen as a subset of $\R^{\abs{\setA}}$, and $\measures(\setA)$ the set of measures over $\setA$. 
A measure is said to be \emph{strictly positive} if every nonempty set has strictly positive measure. 
The set of strictly positive measures over $\setA$ is denoted by $\measures_{+}(\setA)$. 
The identity function is denoted by $\Id$.
In the whole paper, if $\numeas\in\measures(\setA)$ and $y\in\setA$, then, for the sake of simplicity, we write $\numeas(y)$, instead of $\numeas(\braces{y})$.  

Let $\nplayers\geq 2$. 
A \emph{finite game} $\parens{\players,\parens{\actions^{\play},\payoff^{\play}}_{\play\in\players}}$ consists of a finite set of players, denoted by $\players=\braces{1,\dots,\nplayers}$, and, for each player $\play\in\players$, a finite set of strategies $\actions^{\play}$ such that $\abs{\actions^{\play}}\geq 2$ and a payoff function $\payoff^{\play}:\actions \to \R$, where $\actions\coloneqq\bigtimes_{\play\in\players}\actions^{\play}$ is the space of strategy profiles $\actprof$. 
A game is called \emph{constant} if for every $\play \in \players$, $\game^{\play}$ is constant over $\actions$.
The symbol $\actions^{-\play}$ denotes the set of strategy subprofiles that excludes player $\play$.

For fixed $\players$ and $\actions$, the game $\parens{\players,\parens{\actions^{\play},\payoff^{\play}}_{\play\in\players}}$ is uniquely defined by its vector of payoff functions $\game$, which, with a slight abuse of language, we will call a \emph{game}.

The space $\games_{\actions}$ of games with player set $\players$ and strategy-profile set $\actions$ can be identified with $\R^{\nplayers\abs{\actions}}$.
As a consequence, $\dim\games_{\actions}=\nplayers\prod_{\play\in\players}\abs{\actions^{\play}}$.

For $\play\in\players$, let $\numeas^{\play}$ be a strictly positive measure  on $\actions^{\play}$. 
The product measure $\numeas$ on $\actions$ is defined as
\begin{equation}
\label{eq:prod-nu}
\numeas(\actprof) \coloneqq \prod_{\play\in\players} \numeas^{\play}(\act^{\play}).
\end{equation}
For every $\play \in \players$,  $\numeas^{-\play}$ is a strictly positive product measure on $\actions^{-\play}$ defined as
\begin{equation}
\label{eq:prod-nu-i}
\numeas^{-\play}(\actprof^{-\play}) \coloneqq \prod_{\playalt\neq\play} \numeas^{\playalt}(\act^{\playalt}).
\end{equation}
Given any finite measure $\numeas^{\play}$ on $\actions^{\play}$, its normalization is defined as 
\begin{equation}
\label{eq:norm-nu}
\norml{\numeas}^{\play}(\act^{\play}) \coloneqq \frac{\numeas^{\play}(\act^{\play})}{\sum_{\actalt^{\play}\in\actions^{\play}}\numeas^{\play}(\actalt^{\play})}.
\end{equation}

\begin{definition}
\label{de:co-measure}
For $\play\in\players$, a strictly positive finite measure $\betameas^{\play}\in\measures_{+}(\actions^{-\play})$ is called a \emph{co-measure} of player $\play$.
The vector $\nsgameb=\parens*{\betameas_{1},\dots,\betameas_{\players}}$ is called a \emph{co-measure vector}.
\end{definition}

Later in this section we will generalize the decomposition proposed by \citet{CanMenOzdPar:MOR2011}.
To this end, we introduce the  classes of games that will appear in our decomposition. 
Since in this section the set of strategy profiles is  fixed, the set of games is denoted just by $\games$.

\begin{definition}
For every player $\play\in\players$, let  $\mumeas^{\play}$ be a strictly positive finite measure on $\actions^{\play}$, let $\gammameas^{\play}$ be a  co-measure on $\actions^{\play}$,  and let $\muprof=\parens*{\mumeas^{\play}}_{\play\in\players}$ and $\nsgameg=\parens*{\gammameas^{\play}}_{\play\in\players}$.
\label{de:classes-of-games}
\begin{enumerate}[(a)]
\item
\label{it:classes-nonstrategic}
A game $\game\in\games$ is \emph{nonstrategic} if, for each $\play\in\players$, there exists a function $\nonstratf \colon \actions^{-\play}\to\R$ such that
\begin{equation}
\label{eq:nonstrategic}
\payoff^{\play}(\act^{\play},\actprof^{-\play})=\nonstratf^{\play}(\actprof^{-\play}).
\end{equation}
The class of nonstrategic games is denoted by $\NSG$.

\item
\label{it:classes-mu-normalized}
A game $\game\in\games$ is \emph{$\muprof$-normalized} if, for each $\play\in\players$ and for each $\actprof^{-\play}\in\actions^{-\play}$, we have
\begin{equation}
\label{eq:mu-normalized}
\sum_{\act^{\play}\in\actions^{\play}}\mumeas^{\play}(\act^{\play})\payoff(\act^{\play},\actprof^{-\play})=0.
\end{equation}
The class of $\muprof$-normalized games is denoted by $\muprof\NoGd$.

\item
\label{it:classes-eta-potential}
A game $\game\in\games$ is \emph{$\nsgameg$-potential} if there exist $\potent \colon \actions\to\R$ such that, for all $\play\in\players$, for all $\act^{\play},\actalt^{\play}\in\actions^{\play}$, and for all $\actprof^{-\play}\in\actions^{-\play}$, we have
\begin{equation}
\label{eq:eta-potential}
\gammameas^{\play}(\actprof^{-\play})\parens*{\payoff^{\play}(\actalt^{\play},\actprof^{-\play})-\payoff^{\play}(\act^{\play},\actprof^{-\play})}
=\potent(\actalt^{\play},\actprof^{-\play})-\potent(\act^{\play},\actprof^{-\play}).
\end{equation}
The function $\potent$ is referred to as a \emph{potential function} of the game.
The class of $\nsgameg$-potential games is denoted by $\nsgameg\PGd$.

\item
\label{it:classes-mu-eta-harmonic}
A game $\game\in\games$ is \emph{$(\muprof,\nsgameg)$-harmonic} if, for all $\actprof\in\actions$, we have
\begin{equation}
\label{eq:mu-eta-harmonic}
\sum_{\play\in\players} \sum_{\actalt^{\play}\in\actions^{\play}} \mumeas^{\play}(\actalt^{\play})\gammameas^{\play}(\actprof^{-\play})\parens*{\payoff^{\play}(\act^{\play},\actprof^{-\play})-\payoff^{\play}(\actalt^{\play},\actprof^{-\play})}
=0
\end{equation}
The class of $(\muprof,\nsgameg)$-harmonic games is denoted by $(\muprof,\nsgameg)\HGd$.
\end{enumerate}
\end{definition}

%----------------------- SUBSECTION ---------------------------------

\subsection{Potential and harmonic games}
\label{suse:potential-harmonic}
In the terminology of \citet{MonSha:GEB1996}, any weighted potential game is $\nsgameg$-potential and any $\nsgameg$-potential game is ordinal potential. 
An immediate consequence of this is the following proposition.

\begin{proposition}
\label{pr:potential-pure}
A game $\game\in\nsgameg\PGd$  admits a pure equilibrium.
\end{proposition}

Our definition of  $(\muprof,\nsgameg)$-harmonic game generalizes the definition of harmonic game given in \citet{CanMenOzdPar:MOR2011}.
Basically, a harmonic game satisfies \cref{eq:mu-eta-harmonic} when $(\muprof,\nsgameg)$ is such that $\mumeas^{\play}(\act^{\play})=1$ for all $\play \in \players$ and all $\act^{\play} \in \actions^{\play}$  and $\gammameas^{\play}(\actprof^{-\play})=1$ for all $\actprof^{-\play}\in \actions^{-\play}$.

\begin{definition}
\label{de:scaling}
Let  $\nsgameb=\parens*{\betameas^{\play}}_{\play\in\play}$ be a  co-measure vector. The \emph{general scaling of parameter $\nsgameb$} is defined by
\begin{equation}\label{eq:gamma-scaling}
\forall \game \in \games, \ \forall \play\in\players,\ \forall \actprof\in\actions,\ (\nsgameb \cdot\game)^{\play}(\actprof)=\varscal^{\play}(\actprof^{-\play})\payoff^{\play}(\actprof).
\end{equation}
We call $(\nsgameb\cdot\game)$ the \emph{$\nsgameb$-scaled game}.
\end{definition}

\begin{theorem} 
\label{th:mu-eta-mixed-equilibrium}
Let $\game$ be a $(\muprof,\nsgameg)$-harmonic game. Then, the completely mixed strategy profile $\norml{\mumeas^{\play}}$ is a Nash equilibrium of the scaled game $(\nsgameg \cdot \game)$, i.e., for all $\play\in\players$ and for all $\actr^{\play},\actalt^{\play}\in\actions^{\play}$, we have
\begin{equation}
\label{eq:mu-compl-mixed}
\sum_{\actprof^{-\play}\in\actions^{-\play}} \norml{\mumeas}^{-\play}(\actprof^{-\play})\gammameas^{\play}(\actprof^{-\play})\payoff^{\play}(\actr^{\play},\actprof^{-\play})   
=
\sum_{\actprof^{-\play}\in\actions^{-\play}} \norml{\mumeas}^{-\play}(\actprof^{-\play})\gammameas^{\play}(\actprof^{-\play})\payoff^{\play}(\actalt^{\play},\actprof^{-\play}).
\end{equation}
\end{theorem}

When all $\mumeas^{\play}$ are the uniform measure and all $\gammameas^{\play}$ are the uniform co-measure, we obtain the result of \citet{CanMenOzdPar:MOR2011}.
Notice that in \citet{CanMenOzdPar:MOR2011} any harmonic game admits the uniform profile as a Nash equilibrium. This is not the case anymore in our context where there is no immediate information on the Nash-equilibrium set of the game $\game$, but only on the Nash-equilibrium set of the scaled game $(\nsgameg\cdot\game)$. In particular, it is possible that two $(\muprof,\nsgameg)$-harmonic games admit disjoint Nash-equilibrium sets as illustrated by the next example.

\begin{example}
\label{ex:incompatible-1}
Consider the following two $3$-player games where player $1$ chooses the row, player $2$ chooses the column, and player $3$ chooses the matrix. 
For $\play\in\braces{1,2,3}$, player $\play$ has two strategies $\act^{\play}$ and $\actalt^{\play}$. 
In the first game $\game$, player $3$ is dummy and players $1$ and  $2$ play a generalized matching pennies game, where player $2$ wants to match and player $1$ wants to mismatch:
\begin{center}
\medskip
{\tabulinesep=1.2mm
\begin{tabu}{ r|ccc|ccc| }
\multicolumn{1}{r}{}
 &  \multicolumn{3}{c}{$\act^{2}$}
 & \multicolumn{3}{c}{$\actalt^{2}$} \\
\cline{2-7}
$\act^{1}$ & $-1$ & $1$ & $0$ & $2$ & $-1$ & $0$ \\
\cline{2-7}
$\actalt^{1}$ & $1$ & $-1$ & $0$ & $-2$ & $1$ & $0$ \\
\cline{2-7}
\multicolumn{7}{c}{$\act^{3}$}
\end{tabu}
\qquad
\begin{tabu}{ r|ccc|ccc| }
\multicolumn{1}{r}{}
 &  \multicolumn{3}{c}{$\act^{2}$}
 & \multicolumn{3}{c}{$\actalt^{2}$} \\
\cline{2-7}
$\act^{1}$ & $-1$ & $1$ & $0$ & $2$ & $-1$ & $0$ \\
\cline{2-7}
$\actalt^{1}$ & $1$ & $-1$ & $0$ & $-2$ & $1$ & $0$ \\
\cline{2-7}
\multicolumn{7}{c}{$\actalt^{3}$}
\end{tabu}}
\medskip
\end{center}
To simplify notations, we describe a strategy of a player by the probability that he is playing $\actalt$.  
The set of equilibria of this game is $\braces{(1/2,1/3,\mixed), \ \mixed\in [0,1]}$.

In the second game $\gamealt$, player $1$ is a dummy and players $2$ and  $3$ play a generalized matching pennies game, where player $2$ wants to match and player $3$ wants to mismatch:
\begin{center}
\medskip
{\tabulinesep=1.2mm
\begin{tabu}{ r|ccc|ccc| }
\multicolumn{1}{r}{}
 &  \multicolumn{3}{c}{$\act^{2}$}
 & \multicolumn{3}{c}{$\actalt^{2}$} \\
\cline{2-7}
$\act^{1}$ & $0$ & $1$ & $-1$ & $0$ & $-1$ & $3$ \\
\cline{2-7}
$\actalt^{1}$ & $0$ & $1$ & $-1$ & $0$ & $-1$ & $3$ \\
\cline{2-7}
\multicolumn{7}{c}{$\act^{3}$}
\end{tabu}
\qquad
\begin{tabu}{ r|ccc|ccc| }
\multicolumn{1}{r}{}
 &  \multicolumn{3}{c}{$\act^{2}$}
 & \multicolumn{3}{c}{$\actalt^{2}$} \\
\cline{2-7}
$\act^{1}$ & $0$ & $-1$ & $1$ & $0$ & $1$ & $-3$ \\
\cline{2-7}
$\actalt^{1}$ & $0$ & $-1$ & $1$ & $0$ & $1$ & $-3$ \\
\cline{2-7}
\multicolumn{7}{c}{$\actalt^{3}$}
\end{tabu}} .
\medskip
\end{center}
The set of equilibria of this game is $\braces{(\mixed,1/4,1/2), \ \mixed\in [0,1]}$. 
Hence the two sets of Nash-equilibria  are disjoint.

 Moreover, for all $\play \in \players$, let $\mumeas^{\play}$ be the uniform distribution over $\braces*{\act^{\play},\actalt^{\play}}$ and let $\nsgameg$ be defined by
\begin{itemize}
\item $\gammameas^{2}(\argdot,\argdot)=1$,
\item $\gammameas^{1}$ only depends on player $2$: $\gammameas^{1}(\act^{2},\act^{3})=\gammameas^{1}(\act^{2},\actalt^{3})=1$ and $\gammameas^{1}(\actalt^{2},\act^{3})=\gammameas^{1}(\actalt^{2},\actalt^{3})=1/2$,
\item $\gammameas^{3}$ only depends on player $2$: $\gammameas^{3}(\act^{1},\act^{2})=\gammameas^{3}(\actalt^{1},\act^{2})=1$ and $\gammameas^{3}(\act^{1},\actalt^{2})=\gammameas^{3}(\actalt^{1},\actalt^{2})=1/3$. 
\end{itemize}
Then  $\game$ and $\gamealt$ are both $(\muprof,\nsgameg)$-harmonic.
\end{example}

We now introduce a class of co-measures having interesting properties when dealing with $(\muprof,\nsgameg)$-harmonic games.

\begin{definition}
\label{de:product-co-measure}
A  co-measure vector $\nsgameg$ is called a \emph{product  co-measure vector} if there exists  $\cprof:=\parens*{\cmeas^{\playalt}}_{\playalt\in\players}$, with  $\cmeas^{\playalt}\in\measures_{+}(\actions^{\playalt})$ such that, for all $\play\in\players$, we have
\begin{equation}
\label{eq:prod-gamma-i}
\gammameas^{\play}(\actprof^{-\play}) \coloneqq \prod_{\playalt\neq\play} {\cmeas}^{\playalt}(\act^{\playalt})={\cmeas}^{-\play}(\actprof^{-\play}).
\end{equation}
In this case, we say that $\nsgameg$ is generated by $\cprof$.
\end{definition}

Define now, for any $\play\in\players$, 
\begin{equation}
\label{eq:mu-eta-i}
(\mumeas \cmeas)^{\play}(\act^{\play}) \coloneqq \mumeas^{\play}(\act^{\play}){\cmeas}^{\play}(\act^{\play})
\end{equation}
and let $\norml{{\mumeas \cmeas}}^{\play}$ be its normalized version, as in \cref{eq:norm-nu}. 

\begin{corollary} 
\label{cor:mu-eta-mixed-equilibrium-product}
Let $\game$ be a $(\muprof,\nsgameg)$-harmonic game such that $\nsgameg$ is a  co-measure vector generated by $\cprof$. 
Then the completely mixed strategy profile $\norml{{\mumeas \cmeas}}^{\play}$ is a Nash equilibrium of the game $\game$, i.e., for all $\play\in\players$ and for all $\actr^{\play},\actalt^{\play}\in\actions^{\play}$, we have
\begin{equation}
\label{eq:mu-eta-compl-mixed}
\sum_{\actprof^{-\play}\in\actions^{-\play}} \prod_{\playalt\neq\play}\norml{{\mumeas \cmeas}}^{\playalt}(\act^{\playalt})\payoff^{\play}(\actr^{\play},\actprof^{-\play})   
=
\sum_{\actprof^{-\play}\in\actions^{-\play}} \prod_{\playalt\neq\play}\norml{{\mumeas \cmeas}}^{\playalt}(\act^{\playalt})\payoff^{\play}(\actalt^{\play},\actprof^{-\play}).
\end{equation}
It follows that all $(\muprof,\nsgameg)$-harmonic games share a common Nash equilibrium.
\end{corollary}

Again, the framework of \citet{CanMenOzdPar:MOR2011} is obtained when, for each $\play\in\players$, $\mumeas^{\play}$  is the uniform measure and  $\gammameas^{\play}$ is the uniform co-measure.

\begin{remark}
As mentioned before, constant co-measures are a particular case of product co-measures.
As a consequence, applying \cref{cor:mu-eta-mixed-equilibrium-product}, we can show that, in this specific case, we have $\norml{{\mumeas \cmeas}}^{\play}=\norml{\mumeas}^{\play}$. 
Hence, the completely mixed strategy profile $\norml{\mumeas}^{\play}$ is a Nash equilibrium of $\game$.
\end{remark}

%----------------------- SUBSECTION ---------------------------------

\subsection{A first decomposition result}
\label{suse:decomp1}

Our first decomposition states that any game can be written as the sum of a $\muprof$-normalized game and a nonstrategic game. 
The proof follows \citet{CanMenOzdPar:MOR2011} apart from the fact that we choose $\muprof$-normalized games instead of the particular case of uniform $\mumeas^{\play}$. 
In the context of games with continuous strategy sets, a similar decomposition result was proved by \citet{HwaRey:arxiv2016}, who view the set of all games as a Hilbert space.

To present our decomposition results, we first define the space $\Czero\coloneqq\braces{\func: \actions \to \R}$ endowed with the following inner product:
\begin{equation} 
\label{eq:inn-prod-0}
\forall \func,\funcalt\in\Czero,\quad \inner{\func}{\funcalt}_{0}
=\sum_{\actprof\in\actions} \mumeas(\actprof) \func(\actprof)\funcalt(\actprof).
\end{equation}

Notice that $\payoff^{\play}\in\Czero$, for all $\play\in\players$; therefore $\games\cong\Czero^{\nplayers}$. 
Given a function $\func\in\Czero$ and a co-measure  $\nsgameg$,  $\gammameas^{\play}\func$ denotes the function in $\Czero$ defined as
\begin{equation}\label{eq:eta-minus-i}
(\gammameas^{\play}\func)(\actprof)=\gammameas^{\play}(\actprof^{-\play}) \func(\actprof).
\end{equation}
In particular, given a game $\game$, we have
\[
(\nsgameg\cdot\game)^{\play}=(\gammameas^{\play} \game^{\play}).
\]

Our next step is to endow the space of games with a suitable inner product. For any $\game_{1},\game_{2}\in \games$, using \cref{eq:inn-prod-0}, we define
\begin{equation} 
\label{eq:g1g2mu}
\inner{\game_{1}}{\game_{2}}_{\muprof,\nsgameg} 
= \sum_{\play\in\players} \mumeas^{\play}(\actions^{\play}) \inner{\gammameas^{\play}\payoff_{1}^{\play}}{\gammameas^{\play}\payoff^{\play}_{2}}_{0}
\end{equation}
and  $\oplus_{\muprof,\nsgameg}$ will denote the direct orthogonal sum with respect to the above inner product.

\begin{proposition} 
\label{pr:direct-sum}
The space of games $\games$ is the direct orthogonal-sum of the $\muprof$-normalized and nonstrategic subspaces, i.e.,
\begin{equation} 
\label{eq:direct-sum}
\games = \muprof\NoGd \oplus_{\muprof,\nsgameg} \NSG.
\end{equation}
\end{proposition}

The decomposition of \citet{CanMenOzdPar:MOR2011} can be obtained by choosing for $\mumeas^{\play}$  the uniform measure and for $\gammameas^{\play}$ the uniform co-measure.

%----------------------- SUBSECTION ---------------------------------

\subsection{$(\muprof,\nsgameg)$-decomposition}
\label{suse:mu-eta-decomposition}

In this section we show that any game can be decomposed into the direct sum of three component games: a $\nsgameg$-potential $\muprof$-normalized game, a $(\muprof,\nsgameg)$-harmonic $\muprof$-normalized game, and a nonstrategic game. 
Moreover, this decomposition is orthogonal for the inner product in \cref{eq:g1g2mu}.

\begin{theorem} 
\label{th:ortho-sum}
The space of games is the direct orthogonal-sum of the $\muprof$-normalized $\nsgameg$-potential, $\muprof$-normalized $(\muprof,\nsgameg)$-harmonic and nonstrategic subspaces, i.e.,
\begin{equation}
\label{eq:G=sum}
\games=\parens*{\muprof\NoGd \cap \nsgameg\PGd} \oplus_{\muprof,\nsgameg} \parens*{\muprof\NoGd \cap (\muprof,\nsgameg)\HGd} \oplus_{\muprof,\nsgameg} \NSG.
\end{equation}
\end{theorem}

As before, choosing  $\mumeas^{\play}$ to be the uniform measure and  $\gammameas^{\play}$ to be the uniform co-measure  gives the decomposition of \citet{CanMenOzdPar:MOR2011}.

\cref{th:ortho-sum} guarantees that, given a game $\game$, and a pair $(\muprof,\nsgameg)$, we can decompose this game into three components:
a nonstrategic game, a $\nsgameg$-potential $\muprof$-normalized game, and a $(\muprof,\nsgameg)$-harmonic $\muprof$-normalized game. 
The nonstrategic game will be denoted by $\game_{(\muprof,\nsgameg)\NSd}$, the  $\nsgameg$-potential $\muprof$-normalized game by $\game_{(\muprof,\nsgameg)\Potd}$,  and  the $(\muprof,\nsgameg)$-harmonic $\muprof$-normalized game by $\game_{(\muprof,\nsgameg)\Hard}$. Hence, we obtain
\begin{equation}
\label{eq:g=sum}
\game=\game_{(\muprof,\nsgameg)\NSd}+\game_{(\muprof,\nsgameg)\Potd}+\game_{(\muprof,\nsgameg)\Hard}.
\end{equation}
For the sake of simplicity, when there is no risk of confusion, we omit the indication of $(\muprof,\nsgameg)$. 

We emphasize the fact that the set of nonstrategic games does not depend on $(\muprof,\nsgameg)$, but the nonstrategic component in the $(\muprof,\nsgameg)$-decomposition does depend on $(\muprof,\nsgameg)$.
Similarly, the set of $\nsgameg$-potential games does not depend on $\muprof$, but the $\nsgameg$-potential component  in the $(\muprof,\nsgameg)$-decomposition does depend on $(\muprof,\nsgameg)$.

The key point in the proof of \cref{th:ortho-sum} is to associate a given game to a flow  on a suitable graph, as it is done in \citet{CanMenOzdPar:MOR2011}. 
We then characterize $\nsgameg$-potential games, nonstrategic games, and $(\muprof,\nsgameg)$-harmonic games in terms of their induced flows. 
Any flow generated by a game can be decomposed into two particular flows using linear algebra tools. 
From this decomposition result, we can obtain a decomposition in terms of games. 
The definition of the flow and the decomposition result are built on the gradient operator and the inner product, respectively. Hence, the decomposition is implicitly related to some notion of metric induced in the space of games. 
\citet{CanMenOzdPar:MOR2011} use the same metric to define both the flow generated by a game and the corresponding decomposition. 
We generalize their approach  in the following ways:  first, we deal with families of metrics instead of a unique one and, second, we consider different metrics for the definition of the flow generated by a game and the decomposition. 

To prove our decomposition result, we use the Moore-Penrose pseudo-inverse of the gradient operator \citep[see, e.g.,][]{BenGre:Springer2003}.
Since the components of the decomposition are orthogonal, the pseudo-inverse operator allows us to determine the closest $\nsgameg$-potential game to an arbitrary game with respect to the induced distance in the space of games. 
\citet{CanMenOzdPar:MOR2011} provide in fact two approaches to obtain their decomposition result. 
The first approach relies on the Helmholtz decomposition tool, which becomes a degenerate case when applied to flows induced by games.\footnote{According to the graph representation of games, \emph{curl flows} in games are generated only by consecutive deviations of the same player and hence, they are mapped to 0. As a consequence, the divergence-free component of the Helmholtz decomposition tool is reduced to \emph{harmonic flows}.} 
The second one relies on the Moore-Penrose pseudo-inverse operator. 
We choose the second approach since we find it better suited to study the relation between duplicate strategies and decomposition.

Our  result states that, given a pair $(\muprof,\nsgameg)$, the decomposition into $\muprof$-normalized $\nsgameg$-potential, $\muprof$-normalized $(\muprof,\nsgameg)$-harmonic, and nonstrategic games is unique. 
The map that associates to a given game its components will be referred to as the \emph{$(\muprof,\nsgameg)$-decomposition map}.

The  reason for using two parameters $(\muprof,\nsgameg)$, rather than one will be clear once we introduce game transformations in \cref{se:compati}, in particular duplications of strategies and scalings, which are radically  different game transformations. 
One way to see this difference is that duplicating a strategy of some player affects only her own equilibrium strategy, whereas scaling of some player's payoffs  changes the equilibrium strategies of all the other players, but not her own. 
As a natural consequence, for every player $\play$ we have a parameter $\mumeas^{\play}$, which is a measure on $\actions^{\play}$, and  a second parameter $\gammameas^{\play}$, which is a measure on $\actions^{-\play}$. 
As we will see, the family of decompositions that we obtain also commutes with other game transformations, such as permutations.

Nevertheless, there is some redundancy by having two parameters. This gives rise to the following  question: what is the set of parameters $(\muprof,\nsgameg)$ that induce the same $(\muprof,\nsgameg)$-decomposition map?

\begin{proposition}\label{pro:equivalence-class}
Let $\muprof,\mathring{\muprof}$ be two strictly positive  measure vectors and let $\nsgameg,\mathring{\nsgameg}$ be two   co-measure vectors. Then, the $(\muprof,\nsgameg)$-decomposition map is identical to the $(\mathring{\muprof},\mathring{\nsgameg})$-decomposition map if and only if  there exist $\scala,\scalb>0$ such that, for all $\play\in\players$
\[
\mathring{\mumeas}^{\play}=\scala \mumeas^{\play}\quad\text{and}\quad \mathring{\gammameas}^{\play}=\scalb \gammameas^{\play}.
\]
\end{proposition}

%----------------------- SUBSECTION ---------------------------------

\subsection{Closest potential game and approximate equilibria.}
\label{suse:closest-potential}

Our decomposition is also related to \citet{CanOzdPar:CDC2010}, who develop a theory on how to associate to an arbitrary game its closest weighted potential game and then deduce $\varepsilon$-approximate equilibria.
The set of weighted potential games is not convex, but given a vector $\weightprof=\parens*{\weight^{\play}}$ of weights, the set of $\weightprof$-weighted potential games is convex. 
The authors first compute the closest $\weightprof$-weighted potential game with respect to a metric that depends on $\weightprof$.
Then, to find the closest weighted potential game, they  optimize over $\weightprof$.    
The choice of the metric ensures that this second optimization is a convex problem.

In fact, if we take for every $\play \in \players$ and every $\act^{\play}\in \actions^{\play}$, $\mumeas^{\play}(\act^{\play})=1$ and for all $\actprof^{- \play} \in \actions^{- \play}$, $\gammameas^{\play}(\actprof^{-\play})=\weight^{\play}$, then the metric induced by   $\inner{\argdot}{\argdot}_{\muprof,\nsgameg}$ coincides with the metric used in \citet{CanOzdPar:CDC2010} associated to the weights  $(\weight^{\play})_{\play\in \players}$. 
As a consequence, the closest $\weightprof$-potential game in the sense of \citet{CanOzdPar:CDC2010} is in fact equal to the sum of the  potential component and the nonstrategic component of our decomposition.  
\citet{CanOzdPar:CDC2010} explain how to extend their techniques to ordinal potential games. 
Since a $\nsgameg$-potential game is ordinal, it is therefore possible to apply these techniques to compute the closest $\nsgameg$-potential game and then to optimize on $\nsgameg$.

Then, it is possible to relate the Nash-equilibira of the closest $\nsgameg$-potential game to the approximate Nash-equilibria of the original game. 
Formally, let $\norm{\argdot}_{\muprof,\gammameas}$ be the norm induced by  the inner product $\inner{\argdot}{\argdot}_{\muprof,\nsgameg}$ introduced in \cref{eq:g1g2mu}. 
The  distance between a game $\game$ and its closest $\nsgameg$-potential game $\game_{\Pot}$ is then given by
\begin{align*}
\norm*{\game - \game_{\Pot}}^{2}_{\muprof,\nsgameg}=\inner{\game-\game_{\Pot}}{\game-\game_{\Pot}}_{\muprof,\nsgameg}&=\sum\limits_{\play \in \players} \mumeas^{\play}(\actions^{\play}) \inner{ \gammameas^{\play}\game^{\play}-\gammameas^{\play}\game^{\play}_{\Pot}}{\gammameas^{\play}\game^{\play}-\gammameas^{\play}\game^{\play}_{\Pot}}_{0}\\
&=\sum\limits_{\play \in \players}\mumeas^{\play}(\actions^{\play}) \sum_{\actprof\in\actions} \parens*{\gammameas^{\play}(\actprof^{-\play})}^{2} \parens*{\game^{\play}(\actprof) - \game^{\play}_{\Pot}(\actprof)}^{2}.
\end{align*}

\begin{proposition} 
\label{pr:closest-potential}
Let $\game$ be a game, and $\game_{\Pot}$ be its closest potential game. Define  $\dggp\coloneqq \norm*{\game - \game_{\Pot}}_{\muprof,\nsgameg}$.
Then, every equilibrium of $\game_{\Pot}$ is an $\varepsilon$-equilibrium of $\game$ for some 
\begin{equation}\label{eq:epsilon}
\varepsilon \leq 2 \max\limits_{\playalt\in\players}\max\limits_{\actprof^{-\playalt} \in \actions^{-\playalt}}\frac{\dggp}{\gammameas^{\playalt}(\actprof^{-\playalt})\sqrt{\mumeas^{\playalt}(\actions^{\playalt})}}.
\end{equation}
\end{proposition}

In principle, in \cref{eq:epsilon} it is possible to optimize over the parameters $\muprof$ and $\nsgameg$.
The optimization over $\nsgameg$ is similar to the one described in \citet{CanOzdPar:CDC2010}. 
The properties of an optimization over $\muprof$ are not clear.

%----------------------- SECTION ---------------------------------

\section{Compatibility between decompositions and game transformations}
\label{se:compati}

We are interested in how decompositions and game transformations interact. 
As shown by the examples in the Introduction, applying a $(\muprof,\nsgameg)$-decomposition and then transforming each component  may lead to a different result than first transforming the game and then applying the $(\muprof,\nsgameg)$-decomposition to each component.
That is, in general these operations do  not commute. 
Therefore, we will look at a notion where, instead of focusing on two specific maps, we consider a family of decompositions and a family of game transformations.
The following definitions will be needed for the rest of the analysis.

\begin{definition}
\label{de:game transformation}
Let $\actions$ and $\actionsalt$ be two sets of strategy profiles and let $
\subsgames\subset \games_{\actions}$.  A map $\mapT \colon \subsgames \to\games_{\actionsalt}$ is called a  \acfi{GT}\acused{GT}.
\end{definition}

\begin{definition}
\label{de:3-dec-map}
A \acfi{3DM}\acused{3DM} $\threedecmap_{\actions}$ is a map from $\games_{\actions}$ to $\games_{\actions}^{3}$  such that for every $\game\in \games_{\actions},$
\begin{equation}\label{eq:3-decomposition}
\game=(\threedecmap_{\actions}(\game))_{1}+(\threedecmap_{\actions}(\game))_{2}+(\threedecmap_{\actions}(\game))_{3}
\end{equation}
\end{definition}

\begin{definition}
Let $\threedecmaps$ be a family of \acp{3DM} and let $\maps$ be  a family of \acp{GT}. 
We say that the two families \emph{commute} if, for every \ac{3DM} $\threedecmap \in \threedecmaps$ defined on $\games_{\actions}$ and every \ac{GT} $\mapT\in\maps$ from $\games_{\actions}$ to $\games_{\actionsalt}$, there exists a decomposition $\threedecmapalt \in \threedecmaps$ defined on $\games_{\actionsalt}$ and a triple $\parens{\mapT_{1},\mapT_{2},\mapT_{3}}\in\maps^{3}$ of \acp{GT} from $\games_{\actions}$ to $\games_{\actionsalt}$ such that
\begin{equation}\label{eq:T-commute}
\parens{\mapT_{1},\mapT_{2},\mapT_{3}}\circ\threedecmap = \threedecmapalt \circ \mapT,
\end{equation}
where $\parens{\mapT_{1},\mapT_{2},\mapT_{3}}\circ\threedecmap (\game) \coloneqq (\mapT_{1}(\threedecmap(\game)_{1}),\mapT_{2}(\threedecmap(\game)_{2}),\mapT_{3}(\threedecmap(\game)_{3}))$. 

Equivalently, $\threedecmaps$  and  $\maps$ commute if the diagram in \cref{fi:commutation-2} is commutative.
\begin{center}
\def\ech{1}
\begin{figure}[h]
%\centeringering 
\begin{tikzpicture}[scale=\ech,>=stealth,shorten >=1pt,auto,node distance=4cm,thick,main
 node/.style={draw,font=\Large\bfseries}]

\node [text width=0.5cm,text centered,scale=\ech] (g) at (0,0) {$\games_{\actions}$};
\node [text width=2.5cm,text centered,scale=\ech] (dg) at (6,0) {$\games_{\actions} \times \games_{\actions} \times \games_{\actions} $};
\node [text width=0.5cm,text centered,scale=\ech] (gr) at (0,-2) {$\games_{\actionsalt}$};
\node [text width=2.5cm,text centered,scale=\ech] (dgr) at (6,-2) {$\games_{\actionsalt} \times \games_{\actionsalt} \times \games_{\actionsalt}$};

\draw[->,>=latex,scale=\ech] (g) to node[midway,scale=\ech] {$\threedecmap$}(dg);
%\draw[transform canvas={xshift=4},->,scale=\ech] (g) to node[midway,left,scale=\ech] {$$} (gr);
\draw[->,scale=\ech] (g) to node[midway,right,scale=\ech] {$\mapT$} (gr);
%\draw[->,>=latex,scale=\ech] (g) to node[midway,left,scale=\ech] {$Nash_{\actions}$}(gr);
\draw[->,>=latex,scale=\ech] (gr) to node[midway,scale=\ech] {$\threedecmapalt$}(dgr);
\draw[transform canvas={xshift=25},->,scale=\ech] (dg) to node[midway,right,scale=\ech] {$\mapT_{3}$} (dgr);
\draw[transform canvas={xshift=0},->,scale=\ech] (dg) to node[midway,right,scale=\ech] {$\mapT_{2}$} (dgr);
\draw[transform canvas={xshift=-25},->,scale=\ech] (dg) to node[midway,right,scale=\ech] {$\mapT_{1}$} (dgr);
\end{tikzpicture}
\caption{}\label{fi:commutation-2}
\end{figure}
\end{center}
\end{definition}

\begin{remark}
\label{re:zeta}
As highlighted in the Introduction, restricting to $\threedecmapalt=\threedecmap$ would prevent us from comparing games with different spaces of strategies, like in the duplication framework. 
Moreover, even when $\actions=\actionsalt$, it would be impossible to obtain a commutative diagram for the constant scaling. 
So, even if some meaningful results could be obtained also under the restriction $\mapT_{1}=\mapT_{2}=\mapT_{3}$, to obtain the full spectrum of results in the sequel, we need to deal with the general case where the three \acp{GT} may differ. 
This is why we chose to introduce this  notion of commutation between two families. 
\end{remark}

%----------------------- SUBSECTION ---------------------------------

\subsection{Internal game transformations}
\label{suse:preserving}
In this section, we are interested in comparing games which share the same strategy space. Since $\actions=\actionsalt$, the corresponding space of games will be denoted just by $\games$. 
We will investigate three types of \acp{GT}: permuting the strategies of one player, adding a nonstrategic game and scaling the payoff by a co-measure. 
We will see that these families of transformations commute with the family of $(\muprof,\nsgameg)$-decomposition maps.

%----------------------- SUBSUBSECTION ---------------------------------

\subsubsection{Permutations}
\label{sususe:permutations}

First, we look at the relabeling of the sets of strategies. 
Without loss of generality we focus on  relabeling  the strategies of player $\play$. 
The result can then be extended to  relabeling  the strategies of several players simultaneously by composing the results players by players.

\begin{definition}
Given  a permutation $\permut$ of the set $\actions^{\play}$ of strategies of player $\play$, we define the $\permut$-relabeling transformation $\mapT \colon \games_{\actions} \to \games_{\actions}$ as,
\begin{equation}
\label{eq:permut-relabel}
\forall \actprof\in\actions,\ \forall \playalt\in\players,\ (\mapT(\game))^{\playalt}(\act^{\play},\actprof^{-\play}) = \payoff^{\playalt}(\permut(\act^{\play}),\actprof^{-\play}).
\end{equation}
The image of the game $\game$ by $\permut$-relabeling transformation is  denoted by $\game_{\permut}$.
\end{definition}

\noindent Given a measure $\mumeas^{\play}$ on $\actions^{\play}$, the permuted measure $\mumeas^{\play}_{\permut}$ is defined as
\begin{equation}\label{eq:permuted-meas}
\forall \act^{\play} \in \actions^{\play},\  \mumeas^{\play}_{\permut}(\act^{\play})= \mumeas^{\play}(\permut(\act^{\play})).
\end{equation}

We obtain the following proposition.

\begin{proposition}
\label{pr:decomposition-permutation}
Let $\permut$ be a permutation of the strategies of player $\play$. Then, for every game $\game$, we have
\begin{enumerate}[\upshape(a)]
\item $(\game_{\permut})_{(\muprof_{\permut},\nsgameg_{\permut})\NSd}=(\game_{(\muprof,\nsgameg)\NSd})_{\permut}$,

\item $(\game_{\permut})_{(\muprof_{\permut},\nsgameg_{\permut})\Potd}=(\game_{(\muprof,\nsgameg)\Potd})_{\permut}$,

\item $(\game_{\permut})_{(\muprof_{\permut},\nsgameg_{\permut})\Hard}=(\game_{(\muprof,\nsgameg)\Hard})_{\permut}$, 
\end{enumerate}
i.e., permutations of strategies and decompositions commute.
\end{proposition}

\begin{remark}
In the decomposition of \citet{CanMenOzdPar:MOR2011}, the weights on every profile of strategies are all equal. In particular $\muprof_{\permut}=\muprof$ and $\nsgameg_{\permut}=\nsgameg$. 
It follows that the decomposition of  \citet{CanMenOzdPar:MOR2011} commutes, as a singleton, with permutations.
\end{remark}

%----------------------- SUBSUBSECTION ---------------------------------

\subsubsection{Pseudo-translations}
\label{sususe:pseudo-translations}

We now look at translations of the payoffs by a function that depends on the strategies of the other players. 

We show that, similar to \citet{CanMenOzdPar:MOR2011}, there is a natural relation between the $(\muprof,\nsgameg)$-decomposition of a game $\game$ and the $(\muprof,\nsgameg)$-decomposition of the translated game. 

\begin{definition}
A \ac{GT} $\mapT$ is called a \emph{pseudo-translation} if there exists an affine mapping $\func \colon \games\to\NSG$, such that
\begin{equation}\label{eq:pseudo-translation}
\forall \game \in \games,\ \forall \play\in\players,\ \forall \actprof\in\actions,\ \quad (\mapT(\game))^{\play}(\actprof)=\payoff^{\play}(\actprof)+\func(\game)(\actprof^{-\play}).
\end{equation}
Reciprocally, $\game+\func(\game)$ denotes the image of the game $\game$ by the pseudo-translation of parameter $\func$.
\end{definition}

The following results  already appears in \citet{CanOzdPar:CDC2010}.

\begin{proposition}
\label{pr:decomposition-translation}
Let $\func$ be a game transformation from $\games$ to $\NSG$. Then, for every game $\game$, the $(\muprof,\nsgameg)$-decompositions of $\game$ and $\game+\func(\game)$ are related as follows:
\begin{enumerate}[\upshape(a)]
\item $(\game+\func(\game))_{(\muprof,\nsgameg)\NSd}=\game_{(\muprof,\nsgameg)\NSd}+\func(\game)$,

\item $(\game+\func(\game))_{(\muprof,\nsgameg)\Potd}=\game_{(\muprof,\nsgameg)\Potd}$,

\item $(\game+\func(\game))_{(\muprof,\nsgameg)\Hard}=\game_{(\muprof,\nsgameg)\Hard}$,
\end{enumerate}
i.e., pseudo-translations and decompositions commute.
\end{proposition}

The result of \cref{pr:decomposition-translation} can be represented by the diagram in \cref{fi:commutation-pseudo-translation}. 
Notice that, as mentioned in \cref{re:zeta}, the three maps $\mapT_{1},\mapT_{2},\mapT_{3}$ may be different.
In this case $\mapT_{1}$ is a translation and the remaining two are the identity.

\begin{center}
\def\ech{1}
\begin{figure}[h]
%\centeringering 
\begin{tikzpicture}[scale=\ech,>=stealth,shorten >=1pt,auto,node distance=4cm,thick,main
 node/.style={draw,font=\Large\bfseries}]

\node [text width=0.5cm,text centered,scale=\ech] (g) at (0,0) {$\games$};
\node [text width=7.5cm,text centered,scale=\ech] (dg) at (10,0) {$\NSG \times (\nsgameg\PGd \cap \muprof\NoGd) \times ((\muprof,\nsgameg)\HGd\cap \muprof\NoGd)$};
\node [text width=0.5cm,text centered,scale=\ech] (gr) at (0,-2) {$\games$};
\node [text width=7.5cm,text centered,scale=\ech] (dgr) at (10,-2) {$\NSG \times (\nsgameg\PGd \cap \muprof\NoGd) \times ((\muprof,\nsgameg)\HGd\cap \muprof\NoGd)$};

\draw[->,>=latex,scale=\ech] (g) to node[midway,scale=\ech] {$(\muprof,\nsgameg)$-decomposition}(dg);
%\draw[transform canvas={xshift=4},->,scale=\ech] (g) to node[midway,left,scale=\ech] {$$} (gr);
\draw[->,scale=\ech] (g) to node[midway,right,scale=\ech] {$*+\func(*)$} (gr);
%\draw[->,>=latex,scale=\ech] (g) to node[midway,left,scale=\ech] {$Nash_{\actions}$}(gr);
\draw[->,>=latex,scale=\ech] (gr) to node[midway,below,scale=\ech] {$(\muprof,\nsgameg)$-decomposition}(dgr);
\draw[transform canvas={xshift=-90},->,scale=\ech] (dg) to node[midway,right,scale=\ech] {$*+\func(*)$} (dgr);
\draw[transform canvas={xshift=-40},->,scale=\ech] (dg) to node[midway,right,scale=\ech] {$\Id$} (dgr);
\draw[transform canvas={xshift=60},->,scale=\ech] (dg) to node[midway,right,scale=\ech] {$\Id$} (dgr);
\end{tikzpicture}
\caption{}\label{fi:commutation-pseudo-translation}
\end{figure}
\end{center}

%----------------------- SUBSUBSECTION ---------------------------------

\subsubsection{Scaling}
\label{sususe:scaling}

We now turn to scaling by a positive nonstrategic game, as defined in \cref{de:scaling}. 
Given two co-measures $\nsgameg$ and $\nsgameb$, we define the co-measure $\nsgameg/\nsgameb=\parens*{\parens*{\gammameas/\betameas}^{\play}}_{\play\in\players}$ as follows:
\begin{equation}
\label{eq:gamma/beta}
\forall \play\in\players,\ \forall \actprof^{-\play}\in \actions^{-\play},
\ \parens*{\frac{\gammameas}{\varscal}}^{\play}\parens*{\actprof^{-\play}}=\frac{\gammameas^{\play}\parens*{\actprof^{-\play}}}{\varscal^{\play}\parens*{\actprof^{-\play}}}.
\end{equation}

\begin{theorem} \label{th:scaling}
Let $\nsgameb$ be a co-measure. 
Then, for every $\game\in\games_{\actions}$, we have  
\begin{align*}
(\nsgameb\cdot\game_{(\muprof,\nsgameg)\NSd})&= (\nsgameb\cdot\game)_{(\muprof,\nsgameg/\nsgameb)\NSd}\\
(\nsgameb\cdot\game_{(\muprof,\nsgameg)\Potd})&=(\nsgameb\cdot\game)_{(\muprof,\nsgameg/\nsgameb)\Potd} \\
(\nsgameb\cdot\game_{(\muprof,\nsgameg)\Hard})&=(\nsgameb\cdot\game)_{(\muprof,\nsgameg/\nsgameb)\Hard},
\end{align*}
i.e.,  scalings and decompositions  commute.
\end{theorem}

The result of \cref{th:scaling} is represented by the diagram in \cref{fi:commutation_pseudo_product_scaling}.

\begin{center}
\def\ech{1}
\begin{figure}[h]
%\centeringering 
\begin{tikzpicture}[scale=\ech,>=stealth,shorten >=1pt,auto,node distance=4cm,thick,main
 node/.style={draw,font=\Large\bfseries}]

\node [text width=0.5cm,text centered,scale=\ech] (g) at (0,0) {$\games$};
\node [text width=8.2cm,text centered,scale=\ech] (dg) at (9,0) {$\NSG\ \ \times \ \ (\nsgameg\PGd \cap \muprof\NoGd) \ \times \ ((\muprof,\nsgameg)\HGd\cap \muprof\NoGd)$};
\node [text width=0.5cm,text centered,scale=\ech] (gr) at (0,-2) {$\games$};
\node [text width=8.2cm,text centered,scale=\ech] (dgr) at (9,-2) {$\NSG \ \times \ (\nsgameg/\nsgameb\PGd \cap \muprof\NoGd) \times ((\muprof,\nsgameg/\nsgameb)\HGd\cap \muprof\NoGd)$};

\draw[->,>=latex,scale=\ech] (g) to node[midway,scale=\ech] {$(\muprof,\nsgameg)$-decomposition}(dg);
%\draw[transform canvas={xshift=4},->,scale=\ech] (g) to node[midway,left,scale=\ech] {$$} (gr);
\draw[->,scale=\ech] (g) to node[midway,left,scale=\ech] {$(\nsgameb \cdot *)$} (gr);
%\draw[->,>=latex,scale=\ech] (g) to node[midway,left,scale=\ech] {$Nash_{\actions}$}(gr);
\draw[->,>=latex,scale=\ech] (gr) to node[midway,below,scale=\ech] {$(\muprof,\nsgameg/\nsgameb)$-decomposition}(dgr);
\draw[transform canvas={xshift=-100},->,scale=\ech] (dg) to node[midway,right,scale=\ech] {$(\nsgameb \cdot *)$} (dgr);
\draw[transform canvas={xshift=-40},->,scale=\ech] (dg) to node[midway,right,scale=\ech] {$(\nsgameb \cdot *)$} (dgr);
\draw[transform canvas={xshift=50},->,scale=\ech] (dg) to node[midway,right,scale=\ech] {$(\nsgameb \cdot *)$} (dgr);
\end{tikzpicture}
\caption{}\label{fi:commutation_pseudo_product_scaling}
\end{figure}
\end{center}

Two  cases  are of particular interest. 
First, if $\nsgameg$ and $\nsgameb$ are both  product co-measures, then $\nsgameg/\nsgameb$ is also a product co-measures. 
Hence, one can apply \cref{cor:mu-eta-mixed-equilibrium-product} to the harmonic component obtained in the decomposition of the scaled game $(\nsgameb \cdot \game)$. 
The second case concerns the von Neumann-Morgenstern representation of preferences. 
Preference over lotteries that satisfy the four rationality axioms (completeness, transitivity, independence of irrelevant alternatives, and continuity) are represented by expected utility. 
This representation is unique, modulo affine transformations. 
That is, if one utility function is an affine transformation of another, then they both represent the same preferences.
Given that game theory is based on von Neumann-Morgenstern axioms and payoffs are expressed in utils, two games based on different representations of the same preferences should be strategically equivalent. 
Translated in the language used in this paper,  a family of decompositions is compatible with expected utility theory if it commutes with pseudo-translation by a constant nonstrategic game and with positive constant scaling. 

It follows that the smallest family of decompositions containing the decomposition of \citet{CanMenOzdPar:MOR2011} and compatible with the expected utility theory is then described by parameters such that both $\mumeas^{\play}$ and $\gammameas^{\play}$ are uniform. 
If we restrict to these sets of parameters, the class of $\nsgameg$-potential games coincides with the class of weighted potential games. 
Moreover, the class of $(\muprof,\nsgameg)$-harmonic games coincides with the class of weighted harmonic games introduced in \citet{CheLiuZhaQi:IEEE2016}. 
Hence, our decomposition result can be seen as a generalization of the decomposition result into weighted potential game and weighted harmonic game proven in \citet{CheLiuZhaQi:IEEE2016}. 
Finally, as  previously noticed, constant nonstrategic games are product nonstrategic games; hence they are a special case of the games considered in \cref{cor:mu-eta-mixed-equilibrium-product}.

We now provide one example of product-scaling that depends on the strategies of the other player and we show how mixed Nash equilibria vary in the  transformed game. 

\begin{example}
\label{ex:depend}
Let $\mumeas^{1}=\mumeas^{2}=\gammameas^{1}=\gammameas^{2}=(1,1)$ and let $\nsgameb$ be generated by $\bmeas$ in a way that $\betameas^{1} \equiv \bmeas^{2}=(2,1)$ and $\betameas^{2} \equiv \bmeas^{1}=(1,3)$. Consider the $\muprof$-normalized game $\game$ (on the left) and its $\nsgameb$-scaling $(\nsgameb\cdot\game)$ (on the right):
\begin{center}
\medskip
{\tabulinesep=1.2mm
\begin{tabu}{ r|cc|cc| }
\multicolumn{1}{r}{}
 & \multicolumn{2}{c}{$\act^{2}$}
 & \multicolumn{2}{c}{$\actalt^{2}$} \\
\cline{2-5}
$\act^{1}$ & $4$ & $-1$ & $-3$ & $1$ \\
\cline{2-5}
$\actalt^{1}$ & $-4$ & $0$ & $3$ & $0$ \\
\cline{2-5}
\multicolumn{5}{c}{$\game$}
\end{tabu}
\qquad
\begin{tabu}{ r|cc|cc| }
\multicolumn{1}{r}{}
 &  \multicolumn{2}{c}{$\act^{2}$}
 & \multicolumn{2}{c}{$\actalt^{2}$} \\
\cline{2-5}
$\act^{1}$ & $8$ & $-1$ & $-3$ & $1$ \\
\cline{2-5}
$\actalt^{1}$ & $-8$ & $0$ & $3$ & $0$ \\
\cline{2-5}
\multicolumn{5}{c}{$(\nsgameb\cdot\game)$}
\end{tabu}} .
\medskip
\end{center}
The $(\muprof,\nsgameg)$-decomposition of $\game$ is given by
\begin{center}
\medskip 
{\tabulinesep=1.2mm
\begin{tabu}{ r|cc|cc| }
\multicolumn{1}{r}{}
 &  \multicolumn{2}{c}{$\act^{2}$}
 & \multicolumn{2}{c}{$\actalt^{2}$} \\
\cline{2-5}
$\act^{1}$ & $2$ & $1$ & $-1$ & $-1$ \\
\cline{2-5}
$\actalt^{1}$ & $-2$ & $-2$ & $1$ & $2$ \\
\cline{2-5}
\multicolumn{5}{c}{$\game_{(\muprof,\nsgameg)\Potd}$}
\end{tabu}
\qquad 
\begin{tabu}{ r|cc|cc| }
\multicolumn{1}{r}{}
 &  \multicolumn{2}{c}{$\act^{2}$}
 & \multicolumn{2}{c}{$\actalt^{2}$} \\
\cline{2-5}
$\act^{1}$ & $2$ & $-2$ & $-2$ & $2$ \\
\cline{2-5}
$\actalt^{1}$ & $-2$ & $2$ & $2$ & $-2$ \\
\cline{2-5}
\multicolumn{5}{c}{$\game_{(\muprof,\nsgameg)\Hard}$}
\end{tabu}} 
\medskip
\end{center}
In view of the $(\muprof,\nsgameg)$-decomposition of $\game$,  the $\nsgameb$-scaling of $\game_{(\muprof,\nsgameg)\Potd}$ and $\game_{(\muprof,\nsgameg)\Hard}$ are 
\begin{center}
\medskip
{\tabulinesep=1.2mm
\begin{tabu}{ r|cc|cc| }
\multicolumn{1}{r}{}
 &  \multicolumn{2}{c}{$\act^{2}$}
 & \multicolumn{2}{c}{$\actalt^{2}$} \\
\cline{2-5}
$\act^{1}$ & $4$ & $1$ & $-1$ & $-1$ \\
\cline{2-5}
$\actalt^{1}$ & $-4$ & $-6$ & $1$ & $6$ \\
\cline{2-5}
\multicolumn{5}{c}{$(\nsgameb\cdot\game_{(\muprof,\nsgameg)\Potd})$}
\end{tabu} 
\qquad
\begin{tabu}{ r|cc|cc| }
\multicolumn{1}{r}{}
 &  \multicolumn{2}{c}{$\act^{2}$}
 & \multicolumn{2}{c}{$\actalt^{2}$} \\
\cline{2-5}
$\act^{1}$ & $4$ & $-2$ & $-2$ & $2$ \\
\cline{2-5}
$\actalt^{1}$ & $-4$ & $6$ & $2$ & $-6$ \\
\cline{2-5}
\multicolumn{5}{c}{$(\nsgameb\cdot\game_{(\muprof,\nsgameg)\Hard})$}
\end{tabu}} .
\medskip
\end{center} 
Let $\widetilde{\gammameas}^{1}=\parens{1,1/3}$ and $\widetilde{\gammameas}^{2}=\parens{1/2,1}$. 
It is easy to check that $(\nsgameb\cdot\game_{(\muprof,\nsgameg)\Potd})$ is a $\widetilde{\nsgameg}$-potential game and that $(\nsgameb\cdot\game_{(\muprof,\nsgameg)\Hard})$ is a $(\muprof,\widetilde{\nsgameg})$-harmonic game. 
Notice that $\game_{(\muprof,\nsgameg)\Potd}$ and $(\nsgameb\cdot\game_{(\muprof,\nsgameg)\Potd})$ share the same pure Nash equilibria, however the mixed Nash equilibrium of $\game_{(\muprof,\nsgameg)\Potd}$ has been transformed, according to \cref{pr:equilibrim-preservation}, into the mixed strategy profile $(\mixed^{1},\mixed^{2})$, where $\mixed^{1}=(6/7,1/7)$ and $\mixed^{2}=(1/5,4/5)$, which is a mixed Nash equilibrium in $(\nsgameb\cdot\game_{(\muprof,\nsgameg)\Potd})$. 
Likewise, the mixed Nash equilibrium of $\game_{(\muprof,\nsgameg)\Hard}$ has been transformed into the mixed strategy profile $(\mixed^{1},\mixed^{2})$, where $\mixed^{1}=(3/4,1/4)$ and $\mixed^{2}=(2/3,1/3)$, which is the unique mixed Nash equilibrium in $(\nsgameb\cdot\game_{(\muprof,\nsgameg)\Hard})$.
\end{example}

In the following example, we show that, if we restrict ourselves to uniform $\gammameas ^{\play}$, we cannot deal with scaling transformations by using only the measure vector $\muprof$.

\begin{example}
\label{ex:incompatible-2}
Consider the following $3$-player game where player $1$ chooses the row, player $2$ chooses the column, and player $3$ chooses the matrix. 
For $\play\in\braces{1,2,3}$, player $\play$ has two strategies $\act^{\play}$ and $\actalt^{\play}$. 
This game follows from three matching-pennies games played simultaneously where each player wants to match with one among the others and mismatch with the other one. Precisely, player $1$ wants to match with player $2$ but not with player $3$, player $2$ wants to match with player $3$ but not with player $1$ and player $3$ wants to match with player $1$ but not with player $2$. 
\begin{center}
\medskip
{\tabulinesep=1.2mm
\begin{tabu}{ r|ccc|ccc| }
\multicolumn{1}{r}{}
 &  \multicolumn{3}{c}{$\act^{2}$}
 & \multicolumn{3}{c}{$\actalt^{2}$} \\
\cline{2-7}
$\act^{1}$ & $0$ & $0$ & $0$ & $-2$ & $0$ & $2$ \\
\cline{2-7}
$\actalt^{1}$ & $0$ & $2$ & $-2$ & $2$ & $-2$ & $0$ \\
\cline{2-7}
\multicolumn{7}{c}{$\act^{3}$}
\end{tabu}
\qquad
\begin{tabu}{ r|ccc|ccc| }
\multicolumn{1}{r}{}
 &  \multicolumn{3}{c}{$\act^{2}$}
 & \multicolumn{3}{c}{$\actalt^{2}$} \\
\cline{2-7}
$\act^{1}$ & $2$ & $-2$ & $0$ & $0$ & $2$ & $-2$ \\
\cline{2-7}
$\actalt^{1}$ & $-2$ & $0$ & $2$ & $0$ & $0$ & $0$ \\
\cline{2-7}
\multicolumn{7}{c}{$\actalt^{3}$}
\end{tabu}}
\medskip
\end{center}
This game is $(\muprof,\nsgameg)$-harmonic for uniform $\mumeas^{\play}$ and $\gammameas^{\play}$. 
Consider now the following scaled version of this game where we have multiplied the payoffs of player 1 by $2$ when the other players choose $(\actalt^{2},\act^{3})$.
\begin{center}
\medskip
{\tabulinesep=1.2mm
\begin{tabu}{ r|ccc|ccc| }
\multicolumn{1}{r}{}
 &  \multicolumn{3}{c}{$\act^{2}$}
 & \multicolumn{3}{c}{$\actalt^{2}$} \\
\cline{2-7}
$\act^{1}$ & $0$ & $0$ & $0$ & $-4$ & $0$ & $2$ \\
\cline{2-7}
$\actalt^{1}$ & $0$ & $2$ & $-2$ & $4$ & $-2$ & $0$ \\
\cline{2-7}
\multicolumn{7}{c}{$\act^{3}$}
\end{tabu}
\qquad
\begin{tabu}{ r|ccc|ccc| }
\multicolumn{1}{r}{}
 &  \multicolumn{3}{c}{$\act^{2}$}
 & \multicolumn{3}{c}{$\actalt^{2}$} \\
\cline{2-7}
$\act^{1}$ & $2$ & $-2$ & $0$ & $0$ & $2$ & $-2$ \\
\cline{2-7}
$\actalt^{1}$ & $-2$ & $0$ & $2$ & $0$ & $0$ & $0$ \\
\cline{2-7}
\multicolumn{7}{c}{$\actalt^{3}$}
\end{tabu}}
\medskip
\end{center}
Assume that the scaled game is $(\widetilde{\muprof},\nsgameg)$ harmonic where all $\gammameas^{\play}$ are uniform. In doing so we obtain the following conditions on $\widetilde{\muprof}$:
\[
\widetilde\mumeas^{2}(\act^{2})=\widetilde\mumeas^{3}(\act^{3}), \quad \widetilde\mumeas^{1}(\act^{1})=\widetilde\mumeas^{3}(\act^{3}),\quad \mumeas^{2}(\act^{2})=\frac{\mumeas^{1}(\act^{1})}{2}.
\]
It follows that $\widetilde\mumeas^{1}(\act^{1})=0$ which is a contradiction since $\muprof$ is strictly positive.

\end{example}

%----------------------- SUBSECTION ---------------------------------

\subsection{Games with duplicate strategies} 
\label{suse:duplicate}

In this section, we investigate whether duplications and $(\muprof,\nsgameg)$-decompositions commute. After establishing the result concerning the duplications, we will investigate the reciprocal result where the original game has a duplicate strategy that we wish to eliminate. Finally, we will see that the result can be partially extended to the notion of redundant strategy introduced by \citet{GovWil:E2009}.

Without loss of generality, in the sequel we will duplicate the strategy $\act_{1}^{\play}\in \actions^{\play}$ of player $\play$ into  two strategies $\act_{0}^{\play}\notin \actions^{\play}$ and $\act_{1}^{\play}$.  Given a game $\game$, one can either first decompose it and then duplicate it or vice versa. Each procedure yields a different decomposition. Nevertheless, our notion of $(\muprof,\nsgameg)$-decomposition allows us to obtain a relation between the two approaches, which is described by the commutativity diagram in \cref{fi:depend}.

\begin{definition}
\label{de:extended}
The \emph{extended strategy set} of player $\play$ is denoted by $\widecheck{\actions}^{\play} = \actions^{\play} \cup \{\act_{0}^{\play}\}$ and
\begin{equation*}
\widecheck{\actions}=\widecheck{\actions}^{\play} \times \parens*{\bigtimes_{\playalt\neq\play}\actions^{\playalt}}
\end{equation*} 
stands for the \emph{extended set of strategy profiles}. 
The \emph{extended game} is denoted by $\widecheck{\game}=(\widecheck{\payoff}^{\playalt})_{\playalt\in\players}$, where $\widecheck{\payoff}^{\playalt}:\widecheck{\actions} \to \R$ is such that 
\[
\forall \actprof \in \actions,\ \widecheck{\payoff}^{\playalt}(\actprof)=\payoff^{\playalt}(\actprof),
\]
and 
\[
\forall \actprof^{-\play} \in \actions^{-\play},\ \widecheck{\payoff}^{\playalt}(\act_{0}^{\play},\actprof^{-\play})=\payoff^{\playalt}(\act_{1}^{\play},\actprof^{-\play}).
\]
The \emph{extending map} will be denoted by $\widecheck{\mapT}$.
  
We say that $\widecheck{\muprof}$ is an \emph{extension of $\muprof$} if it satisfies the following equations:
\begin{enumerate}[(i)]
\item 
for every player $\playalt\neq\play$, we have $\widecheck{\mumeas}^{\playalt}(\act^{\playalt})=\mumeas^{\playalt}(\act^{\playalt})$,
\item 
for player $\play$, we have $\widecheck{\mumeas}^{\play}(\act_{0}^{\play})+\widecheck{\mumeas}^{\play}(\act_{1}^{\play})=\mumeas^{\play}(\act_{1}^{\play})$, and $\widecheck{\mumeas}^{\play}(\act^{\play})=\mumeas^{\play}(\act^{\play})$ for any $\act^{\play} \in \actions^{\play} \setminus \{\act_{0}^{\play},\act_{1}^{\play}\}$.
\end{enumerate}

We say that $\widecheck{\nsgameg}$ is an \emph{extension of $\nsgameg$} if it satisfies the following equations:
\begin{enumerate}[(i)]
\item 
for  player $\play$, we have $\widecheck{\gammameas}^{\play}(\actprof^{-\play})=\gammameas^{-\play}(\actprof^{-\play})$,
\item 
for any player $\playalt\neq\play$, we have 
$\widecheck{\gammameas}^{\playalt}(\act_{0}^{\play},\actprof^{-(\play,\playalt)})+\widecheck{\gammameas}^{\play}(\act_{1}^{\play},\actprof^{-(\play,\playalt)})=\gammameas^{\playalt}(\act_{1}^{\play},\actprof^{-(\play,\playalt)})$, 
and $\widecheck{\gammameas}^{\playalt}(\actprof^{-\playalt})=\gammameas^{\playalt}(\actprof^{-\playalt})$ for any $\actprof^{-\playalt} \in \actions^{-\playalt}$ such that $\act^{\play}\not\in\braces{\act_{0}^{\play},\act_{1}^{\play}}$.
\end{enumerate}
Notice that there exist infinitely many extensions of $\muprof$ and of $\nsgameg$.
\end{definition}

We then obtain the following relation between a game $\game$ in $\games_{\actions}$ and the extended game $\widecheck{\game}$ in $\games_{\widecheck{\actions}}$.
\begin{center}
\def\ech{1}
\begin{figure}[h]
%\centeringering 
\begin{tikzpicture}[scale=\ech,>=stealth,shorten >=1pt,auto,node distance=4cm,thick,main
 node/.style={draw,font=\Large\bfseries}]

\node [text width=0.5cm,text centered,scale=\ech] (g) at (0,0) {$\games_{\actions}$};
\node [text width=7.5cm,text centered,scale=\ech] (dg) at (9,0) {$\NSG  \times (\nsgameg\PGd\cap \muprof\NoGd)  \times  ((\muprof,\nsgameg)\HGd\cap \muprof\NoGd) $};
\node [text width=0.5cm,text centered,scale=\ech] (gr) at (0,-2) {$\games_{\widecheck{\actions}}$};
\node [text width=7.5cm,text centered,scale=\ech] (dgr) at (9,-2) {$\NSG \times (\widecheck{\nsgameg}\PGd \cap \widecheck{\muprof}\NoGd) \times (\widecheck{\muprof},\widecheck{\nsgameg})\HGd\cap \widecheck{\muprof}\NoGd)$};

\draw[->,>=latex,scale=\ech] (g) to node[midway,scale=\ech] {$(\muprof,\nsgameg)$-decomposition}(dg);
%\draw[transform canvas={xshift=4},->,scale=\ech] (g) to node[midway,left,scale=\ech] {$$} (gr);
\draw[->,scale=\ech] (g) to node[midway,right,scale=\ech] {$\widecheck{\mapT}$} (gr);
%\draw[->,>=latex,scale=\ech] (g) to node[midway,left,scale=\ech] {$Nash_{\actions}$}(gr);
\draw[->,>=latex,scale=\ech] (gr) to node[midway,below,scale=\ech] {$(\widecheck{\muprof},\widecheck{\nsgameg})$-decomposition}(dgr);
\draw[transform canvas={xshift=-90},->,scale=\ech] (dg) to node[midway,right,scale=\ech] {$\widecheck{\mapT}$} (dgr);
\draw[transform canvas={xshift=-35},->,scale=\ech] (dg) to node[midway,right,scale=\ech] {$\widecheck{\mapT}$} (dgr);
\draw[transform canvas={xshift=50},->,scale=\ech] (dg) to node[midway,right,scale=\ech] {$\widecheck{\mapT}$} (dgr);
\end{tikzpicture}
%\caption{}\label{fi:commutation_reduction}
\caption{}\label{fi:depend}
\end{figure}
\end{center}

\begin{lemma}
\label{le:replica}
Let $\game\in \games_{\actions}$, let $\widecheck{\muprof}$ be an extension of $\muprof$, and let $\widecheck{\nsgameg}$ be the extension of $\nsgameg$. 
Then
\begin{enumerate}[\upshape(i)]
\item\label{it:le:replica-1}
if $\game$ is $\nsgameg$-potential, then $\widecheck{\game}$ is $\widecheck{\nsgameg}$-potential,
\item\label{it:le:replica-2}
if $\game$ is nonstrategic, then $\widecheck{\game}$ is nonstrategic,
\item\label{it:le:replica-3}
if $\game$ is $\muprof$-normalized, then $\widecheck{\game}$ is $\widecheck{\muprof}$-normalized,
\item\label{it:le:replica-4}
if $\game$ is $(\muprof,\nsgameg)$-harmonic, then $\widecheck{\game}$ is $(\widecheck{\muprof},\widecheck{\nsgameg})$-harmonic.
\end{enumerate}
\end{lemma}

The following theorem is an immediate consequence of the previous lemma by uniqueness of the $(\widecheck{\muprof},\widecheck{\nsgameg})$-decomposition.

\begin{theorem} 
\label{th:reduce-replicate}
Let $\game\in \games_{\actions}$, let $\widecheck{\muprof}$ be an extended product measure of $\muprof$ and let $\widecheck{\nsgameg}$ be the extended nonstrategic game of $\nsgameg$. Then
\begin{align*}
\widecheck{{\game_{(\muprof,\nsgameg)\NSd}}}&= (\widecheck{\game})_{(\widecheck{\muprof},\widecheck{\nsgameg})\NSd},\\
\widecheck{{\game_{(\muprof,\nsgameg)\Potd}}}&= (\widecheck{\game})_{(\widecheck{\muprof},\widecheck{\nsgameg})\Potd},\\
\widecheck{{\game_{(\muprof,\nsgameg)\Hard}}}&= (\widecheck{\game})_{(\widecheck{\muprof},\widecheck{\nsgameg})\Hard},
\end{align*}
i.e., extensions and decompositions  commute.
\end{theorem}

We now turn to the reduction of games with duplicate strategies. 
Notice that in \cref{th:reduce-replicate}, the nonstrategic game $\nsgameg$ is changed into $\widecheck{\nsgameg}$ by also duplicating the strategy $\act_{1}^{\play}$. 
Hence, we will restrict $\nsgameg$ to have the same duplicate strategy as $\game$.
To study how our decomposition behaves in games with duplicate strategies, we first provide an example of a $(\muprof,\nsgameg)$-harmonic game with some duplicate strategies, whose reduction is not $(\muprof,\nsgameg)$-harmonic. 
We will then focus on the elimination of one duplicate strategy. 
There are games, such as the one in \cref{ex:multi}, where several players have duplicate strategies---possibly more than one. 
In these cases, it is possible to eliminate duplicate strategies one by one. 
At each iteration of the procedure, one duplicate strategy is eliminated and the underlying measure is updated. 
The order of elimination does not influence the measure and the game obtained at the end of the iterated procedure.

Without loss of generality, in the sequel we only consider games where player $\play$ has two duplicate strategies $\act_{0}^{\play}$ and $\act_{1}^{\play}$ and we will study the elimination of the strategy $\act_{0}^{\play}$.  
Duplicate strategies remain duplicate in the components of the decomposition result. 
That is, given a game with some duplicate strategy, the components of the $(\muprof,\nsgameg)$-decomposition contain  the same duplicate strategy. 
It follows that, when considering a game with some duplicate strategies, one can either first decompose it and then reduce it or vice versa. Each procedure yields a different decomposition. 
Nevertheless, our notion of $(\muprof,\nsgameg)$-decomposition allows us to obtain a relation between the two approaches, which is described by the commutativity diagram in \cref{fi:decreduc}.

\begin{center}
\def\ech{1}
\begin{figure}[h]
%\centeringering 
\begin{tikzpicture}[scale=\ech,>=stealth,shorten >=1pt,auto,node distance=4cm,thick,main
 node/.style={draw,font=\Large\bfseries}]

\node [text width=0.5cm,text centered,scale=\ech] (g) at (0,0) {$\games_{\actions}^{\dup}$};
\node [text width=7.5cm,text centered,scale=\ech] (dg) at (9,0) {$\NSG  \times (\nsgameg\PGd\cap \muprof\NoGd)  \times  ((\muprof,\nsgameg)\HGd\cap \muprof\NoGd) $};
\node [text width=0.5cm,text centered,scale=\ech] (gr) at (0,-2) {$\games_{\reduc{\actions}}$};
\node [text width=7.5cm,text centered,scale=\ech] (dgr) at (9,-2) {$\NSG \times (\reduc{\nsgameg}\PGd \cap \reduc{\muprof}\NoGd) \times (\reduc{\muprof},\reduc{\nsgameg})\HGd\cap \reduc{\muprof}\NoGd)$};

\draw[->,>=latex,scale=\ech] (g) to node[midway,scale=\ech] {$(\muprof,\nsgameg)$-decomposition}(dg);
%\draw[transform canvas={xshift=4},->,scale=\ech] (g) to node[midway,left,scale=\ech] {$$} (gr);
\draw[->,scale=\ech] (g) to node[midway,right,scale=\ech] {$\reduc{\mapT}$} (gr);
%\draw[->,>=latex,scale=\ech] (g) to node[midway,left,scale=\ech] {$Nash_{\actions}$}(gr);
\draw[->,>=latex,scale=\ech] (gr) to node[midway,below,scale=\ech] {$(\reduc{\muprof},\reduc{\nsgameg})$-decomposition}(dgr);
\draw[transform canvas={xshift=-90},->,scale=\ech] (dg) to node[midway,right,scale=\ech] {$\reduc{\mapT}$} (dgr);
\draw[transform canvas={xshift=-35},->,scale=\ech] (dg) to node[midway,right,scale=\ech] {$\reduc{\mapT}$} (dgr);
\draw[transform canvas={xshift=50},->,scale=\ech] (dg) to node[midway,right,scale=\ech] {$\reduc{\mapT}$} (dgr);
\end{tikzpicture}
%\caption{}\label{fi:commutation_reduction}
\caption{}\label{fi:decreduc}.
\end{figure}
\end{center}

\begin{definition}
\label{de:dupli}
The strategy $\act_{0}^{\play}$ is \emph{a duplicate strategy} of $\act_{1}^{\play}$  in the game $\game \in \games$ if for every $\playalt\in\players$ and every $\actprof^{-\play} \in \actions^{-\play}$, we have: $\payoff^{\playalt}(\act_{0}^{\play},\actprof^{-\play})=\payoff^{\playalt}(\act_{1}^{\play},\actprof^{-\play})$.
The set of games such that $\act_{0}^{\play}$ is a duplicate strategy of $\act_{1}^{\play}$ is denoted by $\games_{\actions}^{\dup}$.

A  co-measure vector $\nsgameg$ is called \emph{coherent} with $\games_{\actions}^{\dup}$ if for every $\playalt\in\players\setminus\braces{\play}$ and every $\actprof^{-\play} \in \actions^{-\play}$, we have: $\gammameas^{\playalt}(\act_{0}^{\play},\actprof^{-(\play,\playalt)})=\gammameas^{\playalt}(\act_{1}^{\play},\actprof^{-(\play,\playalt)})$.
\end{definition}

\begin{definition}
\label{de:reduced}
The \emph{reduced strategy set} of player $\play$ is denoted by $\reduc{\actions}^{\play} = \actions^{\play} \setminus \{\act_{0}^{\play}\}$ and
\begin{equation*}
\reduc{\actions}=\reduc{\actions}^{\play} \times \parens*{\bigtimes_{\playalt\neq\play}\actions^{\playalt}}
\end{equation*} 
stands for the \emph{reduced strategy profile set}. 
The \emph{reduced game} is denoted by $\reduc{\game}=(\reduc{\payoff}^{\playalt})_{\playalt\in\players}$, where $\reduc{\payoff}^{\playalt}:\reduc{\actions} \to \R$ is such that $\reduc{\payoff}^{\playalt}(\actprof)=\payoff^{\playalt}(\actprof)$. 
The \emph{reducing map} will be denoted by $\reduc{\mapT}$.  If $\muprof$ is a  measure vector on $\actions$, we also define the \emph{reduced  measure vector}  $\reduc{\muprof}$  as follows:
\begin{enumerate}[(i)]
\item 
for every player $\playalt\neq\play$, we have $\reduc{\mumeas}^{\playalt}(\act^{\playalt})=\mumeas^{\playalt}(\act^{\playalt})$,
\item 
for player $\play$, we have $\reduc{\mumeas}^{\play}(\act_{1}^{\play})
=\mumeas^{\play}(\act_{0}^{\play})+\mumeas^{\play}(\act_{1}^{\play})$, and $\reduc{\mumeas}^{\play}(\act^{\play})=\mumeas^{\play}(\act^{\play})$ for any $\act^{\play} \neq \act_{1}^{\play}$.
\end{enumerate}
\end{definition}

\begin{example}
\label{ex:multi}
Let $\mumeas^{1}=\mumeas^{2}=(1,1,1)$ and let $\nsgameg\equiv 1$. 
Let $\game$ be the following game:
\begin{center}
\medskip
{\tabulinesep=1.2mm
\begin{tabu}{ r|cc|cc|cc| }
\multicolumn{1}{r}{}
 &  \multicolumn{2}{c}{$\act^{2}$}
 & \multicolumn{2}{c}{$\actalt_{0}^{2}$} 
  & \multicolumn{2}{c}{$\actalt_{1}^{2}$}\\
\cline{2-7}
$\act_{0}^{1}$ & $2$ & $-2$ & $-1$ & $1$ & $-1$ & $1$ \\
\cline{2-7}
$\act_{1}^{1}$ & $2$ & $-2$ & $-1$ & $1$ & $-1$ & $1$\\
\cline{2-7}
$\actalt^{1}$ & $-4$ & $4$ & $2$ & $-2$ & $2$ & $-2$\\
\cline{2-7}
\multicolumn{7}{c}{$\game$}
\end{tabu}} 
\medskip
\end{center}

The game $\game$ is $(\muprof,\nsgameg)$-harmonic and admits the uniform profile as equilibrium. By eliminating the duplicate strategy of the row-player, we obtain the reduced game $\reduc{\game}$:

\begin{center}
\medskip
{\tabulinesep=1.2mm
\begin{tabu}{ r|cc|cc|cc| }
\multicolumn{1}{r}{}
 &  \multicolumn{2}{c}{$\act^{2}$}
 & \multicolumn{2}{c}{$\actalt_{0}^{2}$} 
  & \multicolumn{2}{c}{$\actalt_{1}^{2}$}\\
\cline{2-7}
$\act^{1}$ & $2$ & $-2$ & $-1$ & $1$ & $-1$ & $1$\\
\cline{2-7}
$\actalt^{1}$ & $-4$ & $4$ & $2$ & $-2$ & $2$ & $-2$\\
\cline{2-7}
\multicolumn{7}{c}{$\reduc{\game}$}
\end{tabu}} .
\medskip
\end{center}
The reduced game $\reduc{\game}$ admits the profile $\parens*{(2/3,1/3),(1/3,1/3,1/3)}$ as unique equilibrium. 
Let $\reduc{\mumeas}^{1}=(2,1)$, $\reduc{\mumeas}^{2}=\mumeas^{2}=(1,1,1)$, $\reduc{\gammameas}^{1}=(1,1)$ and $\reduc{\gammameas}^{2}=(1,1,1)$. 
Then, $\reduc{\game}$ is a $(\reduc{\muprof}, \reduc{\nsgameg})$-harmonic game.
 
It is possible to further eliminate the duplicate strategy of the column-player. We then obtain the reduced game $\dreduc{\game}$:
\begin{center}
\medskip
{\tabulinesep=1.2mm
\begin{tabu}{r|cc|cc|}
\multicolumn{1}{r}{}
&\multicolumn{2}{c}{$\act^{2}$}
&\multicolumn{2}{c}{$\actalt^{2}$}\\
\cline{2-5}
$\act^{1}$ & $2$ & $-2$ & $-1$ & $1$\\
\cline{2-5}
$\actalt^{1}$ & $-4$ & $4$ & $2$ & $-2$\\
\cline{2-5}
\multicolumn{5}{c}{$\dreduc{\game}$}
\end{tabu}} .
\medskip
\end{center}
Let $\dreduc{\mumeas}^{1}=(2,1)$, $\dreduc{\mumeas}^{2}=(1,2)$, $\dreduc{\gammameas}^{1}=(1,1)$ and $\dreduc{\gammameas}^{2}=(1,1)$. It follows that $\dreduc{\game}$ is a $(\dreduc{\mumeas},\dreduc{\nsgameg})$-harmonic game.
\end{example}

Notice that the mapping $\reduc{\mapT}$ is the right-inverse of $\widecheck{\mapT}$. 
Therefore,  the following theorem is an immediate consequence of \cref{th:reduce-replicate}.

\begin{theorem} 
\label{th:reduce-duplicate}
Let $\game\in \games_{\actions}^{\dup}$  and let $\nsgameg$ be coherent with $\games_{\actions}^{\dup}$. 
Then 
\begin{align*}
\reduc{{\game_{(\muprof,\nsgameg)\NSd}}}&= (\reduc{\game})_{(\reduc{\muprof},\reduc{\nsgameg})\NSd},\\
\reduc{{\game_{(\muprof,\nsgameg)\Potd}}}&= (\reduc{\game})_{(\reduc{\muprof},\reduc{\nsgameg})\Potd},\\
\reduc{{\game_{(\muprof,\nsgameg)\Hard}}}&= (\reduc{\game})_{(\reduc{\muprof},\reduc{\nsgameg})\Hard},
\end{align*}
i.e., reductions with respect to duplication and decompositions commute.
\end{theorem}

Finally, \cref{th:reduce-duplicate} can be extended to the notion of redundant strategy introduced by \citet{GovWil:E2009} when $\nsgameg$ is uniform.

\begin{definition}
\label{de:redund}
Given $\ared\in \simplex(\actions^{\play}\setminus \{\act^{\play}_0\})$, the strategy $\act_{0}^{\play}$ is \emph{$\ared$-redundant}  in the game $\game \in \games$  if for every $\playalt\in\players$ and every $\actprof^{-\play} \in \actions^{-\play}$, we have: 
\[\payoff^{\playalt}(\act_{0}^{\play},\actprof^{-\play})=\sum_{\act^{\play} \in \actions^{\play}\setminus \{\act^{\play}_0\}}\ared(\act^{\play}) \payoff^{\playalt}(\act^{\play},\actprof^{-\play}).
\]
The set of games such that $\act_{0}^{\play}$ is $\ared$-redundant is denoted by $\games_{\actions}^{\aredred}$.
\end{definition}

The reduced strategy set, the reduced game and the reducing map are defined as before. 

\begin{definition}
\label{de:reduced_measure_red}
If, for $\play\in\players$, $\mumeas^{\play}\in\measures_{+}(\actions^{\play})$, we define the \emph{$\ared$-reduced measure vector}  $\muprof_{\ared}^{*}$  as follows:
\begin{enumerate}[(i)]
\item 
for every player $\playalt\neq\play$, we have ${\mumeas_{\ared}^{*}}^{\playalt}(\act^{\playalt})=\mumeas^{\playalt}(\act^{\playalt})$,
\item 
for player $\play$ and for every $\act^{\play}\in \actions^{\play}\setminus \{\act^{\play}_0\}$, we have ${\mumeas^*_\ared}^{\play}(\act^{\play})
= \mumeas^{\play}(\act^{\play})+\ared(\act^{\play})\mumeas^{\play}(\act^{\play}_0)$.
\end{enumerate}
\end{definition}

\begin{theorem} 
\label{th:reduce-redundant}
Let $\game \in \games_{\actions}^{\aredred}$ and let $\nsgameg$ be a  uniform co-measure vector.
Then 
\begin{align*}
\reduc{{\game_{(\muprof,\nsgameg)\NSd}}}&= (\reduc{\game})_{(\mumeas^*_\ared,\reduc{\nsgameg})\NSd},\\
\reduc{{\game_{(\muprof,\nsgameg)\Potd}}}&= (\reduc{\game})_{(\mumeas^*_\ared,\reduc{\nsgameg})\Potd},\\
\reduc{{\game_{(\muprof,\nsgameg)\Hard}}}&= (\reduc{\game})_{(\mumeas^*_\ared,\reduc{\nsgameg})\Hard},
\end{align*}
i.e., reductions with respect to redundancy and decompositions commute.
\end{theorem}

The proof of theorem \cref{th:reduce-redundant} is omitted.

In the sequel, we provide two illustrating examples combining scaling and duplication. First, the given game is a $(\muprof,\nsgameg)$-harmonic game where the row-player has a redundant strategy and the goal is to investigate how the two parameters are transformed after eliminating the redundant strategy and applying a $\nsgameb$ scaling in the $\nsgameb$-scaling game $(\nsgameb \cdot \game)$.

\begin{example}\label{ex:redscale}
Let $\mumeas^{1}=(1-\theta,1,\theta)$, with $\theta \in (0,1)$, $\mumeas^{2}\equiv1$ and let  $\nsgameg$ be such that ${\gammameas}^{1}\equiv 1$ and $\gammameas^{2}\equiv1/2$. 
Let  $\game \in \games_{\actions}^{\aredred}$  be the game where the strategy $\actr^{1}$ of the row-player is a $\ared$-redundant, with  $\ared=(\theta,1-\theta)$.
\begin{center}
\medskip
{\tabulinesep=1.2mm
\begin{tabu}{ r|cc|cc| }
\multicolumn{1}{r}{}
 &  \multicolumn{2}{c}{$\act^{2}$}
 & \multicolumn{2}{c}{$\actalt^{2}$} \\
\cline{2-5}
$\act^{1}$ & $1$ & $-2$ & $-1$ & $2$ \\
\cline{2-5}
$\actr^{1}$ & $2\theta - 1 $ & $2(1-2\theta)$ & $1 -2\theta$ & $2(2\theta - 1)$ \\
\cline{2-5}
$\actalt^{1}$ & $-1$ & $2$ & $1$ & $-2$ \\
\cline{2-5}
\multicolumn{5}{c}{$\game$}
\end{tabu}} .
\medskip
\end{center}
The game $\game$ is $(\muprof,\nsgameg)$-harmonic and admits an infinite set of equilibria given by:  
\begin{equation}\label{eq:infinit-mixed}
\braces*{\parens*{\mixed,\parens*{\frac{1}{2}-\mixed}\frac{1}{\theta},1-\mixed-\parens*{\frac{1}{2}-\mixed}\frac{1}{\theta}}, \parens*{\frac{1}{2},\frac{1}{2}} \colon \mixed \in \parens*{0,1/2}, \, \theta \in \parens*{\frac{1-2\mixed}{2(1-\mixed)},1-\mixed}}
\end{equation} 
By eliminating the redundant strategy of the row-player, we obtain the reduced game $\reduc{\game}$:
\begin{center}
\medskip
{\tabulinesep=1.2mm
\begin{tabu}{r|cc|cc|}
\multicolumn{1}{r}{}
&\multicolumn{2}{c}{$\act^{2}$}
&\multicolumn{2}{c}{$\actalt^{2}$}\\
\cline{2-5}
$\act^{1}$ & $1$ & $-2$ & $-1$ & $2$\\
\cline{2-5}
$\actalt^{1}$ & $-1$ & $2$ & $1$ & $-2$\\
\cline{2-5}
\multicolumn{5}{c}{$\reduc{\game}$}
\end{tabu}} .
\medskip
\end{center}
Let $\reduc{\muprof}=\parens*{\reduc{\mumeas}^{1},\reduc{\mumeas}^{2}}$ be the reduced product measure over $\reduc{\actions}$ defined by
\begin{equation}
\label{eq:reduced-pm}
\reduc{\mumeas}^{1}
=\parens*{\mumeas^{1}(\act^{1})/\ared(\actalt^{1}),\mumeas^{1}(\actalt^{1})/\ared(\act^{1})}
=(1,1).
\end{equation} 
and $\reduc{\mumeas}^{2}=\mumeas^{2}$. 
Then, $\reduc{\game}$ is a $(\reduc{\muprof},\reduc{\nsgameg})$-harmonic game, where $\reduc{\nsgameg}$ is the restriction of $\nsgameg$ over the reduced strategy space $\reduc{\actions}$. 
In the sequel, let $\nsgameb$ be such that $\betameas^{1} \equiv 1$ and $\betameas^{2}\equiv  1/2$. 
Then, consider the $\nsgameb$-scaling game $(\nsgameb \cdot \reduc{\game})$:
\begin{center}
\medskip
{\tabulinesep=1.2mm
\begin{tabu}{r|cc|cc|}
\multicolumn{1}{r}{}
&\multicolumn{2}{c}{$\act^{2}$}
&\multicolumn{2}{c}{$\actalt^{2}$}\\
\cline{2-5}
$\act^{1}$ & $1$ & $-1$ & $-1$ & $1$\\
\cline{2-5}
$\actalt^{1}$ & $-1$ & $1$ & $1$ & $-1$\\
\cline{2-5}
\multicolumn{5}{c}{$(\nsgameb \cdot \reduc{\game})$}
\end{tabu}} .
\medskip
\end{center}
Notice that $(\nsgameb \cdot \reduc{\game})$ is the matching-pennies game, which is $(1,1)$-harmonic game and admits the uniform strategy profile as Nash equilibrium. 
This is due to the fact that $(\nsgameb \cdot \reduc{\game})$ is a $(\reduc{\muprof},\reduc{\nsgameg}/\nsgameb)$-harmonic game where $(\reduc{\muprof},\reduc{\nsgameg}/\nsgameb)=(1,1)$.
\end{example}

We now provide a final example combining scaling and duplication. Here, we deal with a duplicate strategy instead of a redundant. 
We investigate how the decomposition of the duplicated and scaled game $({\nsgameb} \cdot \widecheck{\game})$ is related to the decomposition of the reduced game $\game$.

\begin{example}
\label{ex:two-game transformations}
Let $\mumeas^{1}=\mumeas^{2}=\gammameas^{1}=\gammameas^{2}=(1,1)$, $\betameas^{1}=(2,1)$ and $\betameas^{2}=(1,3)$. 
Consider the following $\muprof$-normalized game $\game$ (on the left) and let $({\nsgameb} \cdot \widecheck{\game})$ (on the right) be its $\nsgameb$-scaling where, moreover, the strategy $\act^{1}$ is duplicate:

\begin{center}
\medskip
{\tabulinesep=1.2mm
\begin{tabu}{ r|cc|cc| }
\multicolumn{1}{r}{}
 &  \multicolumn{2}{c}{$\act^{2}$}
 & \multicolumn{2}{c}{$\actalt^{2}$} \\
\cline{2-5}
$\act^{1}$ & $4$ & $-1$ & $-3$ & $1$ \\
\cline{2-5}
$\actalt^{1}$ & $-4$ & $0$ & $3$ & $0$ \\
\cline{2-5}
\multicolumn{5}{c}{$\game$}
\end{tabu}
\qquad
\begin{tabu}{ r|cc|cc| }
\multicolumn{1}{r}{}
 &  \multicolumn{2}{c}{$\act^{2}$}
 & \multicolumn{2}{c}{$\actalt^{2}$} \\
\cline{2-5}
$\act_{0}^{1}$ & $8$ & $-1$ & $-3$ & $1$ \\
\cline{2-5}
$\act_{1}^{1}$ & $8$ & $-1$ & $-3$ & $1$ \\
\cline{2-5}
$\actalt^{1}$ & $-8$ & $0$ & $3$ & $0$ \\
\cline{2-5}
\multicolumn{5}{c}{$({\nsgameb} \cdot \widecheck{\game})$}
\end{tabu}} .
\medskip
\end{center}
In view of the $(\muprof,\nsgameg)$-decomposition of $\game$ (see \cref{ex:depend}), if we apply the same pair of game transformations to  $\game_{(\muprof,\nsgameg)\Potd}$ and $\game_{(\muprof,\nsgameg)\Hard}$, we get

\begin{center}
\medskip
{\tabulinesep=1.2mm
\begin{tabu}{ r|cc|cc| }
\multicolumn{1}{r}{}
 &  \multicolumn{2}{c}{$\act^{2}$}
 & \multicolumn{2}{c}{$\actalt^{2}$} \\
\cline{2-5}
$\act_{0}^{1}$ & $4$ & $1$ & $-1$ & $-1$ \\
\cline{2-5}
$\act_{1}^{1}$ & $4$ & $1$ & $-1$ & $-1$ \\
\cline{2-5}
$\actalt^{1}$ & $-4$ & $-6$ & $1$ & $6$ \\
\cline{2-5}
\multicolumn{5}{c}{$({\nsgameb} \cdot \widecheck{\game}_{(\muprof,\nsgameg)\Potd})$}
\end{tabu} 
\qquad
\begin{tabu}{ r|cc|cc| }
\multicolumn{1}{r}{}
 &  \multicolumn{2}{c}{$\act^{2}$}
 & \multicolumn{2}{c}{$\actalt^{2}$} \\
\cline{2-5}
$\act_{0}^{1}$ & $4$ & $-2$ & $-2$ & $2$ \\
\cline{2-5}
$\act_{1}^{1}$ & $4$ & $-2$ & $-2$ & $2$ \\
\cline{2-5}
$\actalt^{1}$ & $-4$ & $6$ & $2$ & $-6$ \\
\cline{2-5}
\multicolumn{5}{c}{$({\nsgameb} \cdot \widecheck{\game}_{(\muprof,\nsgameg)\Hard})$}
\end{tabu}} .
\medskip
\end{center} 
Let $\widecheck{\gammameas}_{\nsgameb}^{1}=(1,1,1/3)$, $\widecheck{\gammameas}_{\nsgameb}^{2}=(1/2,1)$, $\widecheck{\mumeas}_{\nsgameb}^{1}=(\ared,1-\ared,1)$, $\ared \in (0,1)$ and $\widecheck{\mumeas}_{\nsgameb}^{2}=(1,1)$. 
It is easy to check that $({\nsgameb} \cdot \widecheck{\game}_{(\muprof,\nsgameg)\Potd})$ is a $\widecheck{\nsgameg}_{\nsgameb}$-potential game and $({\nsgameb} \cdot \widecheck{\game}_{(\muprof,\nsgameg)\Hard})$ is a $(\widecheck{\mumeas}_{\nsgameb},\widecheck{\nsgameg}_{\nsgameb})$-harmonic game. 
Notice that the mixed Nash equilibrium of $\game_{(\muprof,\nsgameg)\Potd}$ has been transformed, according to \cref{pr:equilibrim-preservation}, into the mixed strategy profile $(\mixed^{1},\mixed^{2})$, where $\mixed^{1}=(3/7,3/7,1/7)$ and $\mixed^{2}=(1/5,4/5)$, which is a mixed Nash equilibrium in $({\nsgameb} \cdot \widecheck{\game}_{(\muprof,\nsgameg)\Potd})$. 
Likewise, the mixed Nash equilibrium of $\game_{(\muprof,\nsgameg)\Hard}$ has been transformed into the mixed strategy profile $(\mixed^{1},\mixed^{2})$, where $\mixed^{1}=(3/8,3/8,1/4)$ and $\mixed^{2}=(2/3,1/3)$, which is a mixed Nash equilibrium in $({\nsgameb} \cdot \widecheck{\game}_{(\muprof,\nsgameg)\Hard})$.
\end{example}

\begin{remark}
\label{re:noncontinuous}
Another possibility to take into consideration the existence of duplicate strategies would be to associate to a game a canonical reduced form, decompose the game and, then expand the new game to come back to the original strategy space.  
The question of coherence of decompositions is not solved by such an approach. 
Indeed, if we use on every space of games the decomposition of \citet{CanMenOzdPar:MOR2011}, then this new decomposition is not continuous as shown by the following example:

\begin{center}
\begin{multicols}{2}
\medskip
{\tabulinesep=1.2mm
\begin{tabu}{ r|cc|cc| }
\multicolumn{1}{r}{}
 &  \multicolumn{2}{c}{$\act^{2}$}
 & \multicolumn{2}{c}{$\actalt^{2}$} \\
\cline{2-5}
$\act^{1}$ & $1$ & $-1$ & $-1$ & $1$ \\
\cline{2-5}
$\actalt_{0}^{1}$ & $-1$ & $1$ & $1$ & $-1$ \\
\cline{2-5}
$\actalt_{1}^{1}$ & $-1$ & $1$ & $1$ & $-1$ \\
\cline{2-5}
\multicolumn{5}{c}{$\game_{1}$}
\end{tabu}} .
\medskip

\columnbreak
\medskip
{\tabulinesep=1.2mm
\begin{tabu}{ r|cc|cc| }
\multicolumn{1}{r}{}
 &  \multicolumn{2}{c}{$\act^{2}$}
 & \multicolumn{2}{c}{$\actalt^{2}$} \\
\cline{2-5}
$\act^{1}$ & $1$ & $-1$ & $-1$ & $1$ \\
\cline{2-5}
$\actalt_{0}^{1}$ & $-1$ & $1$ & $1$ & $-1$ \\
\cline{2-5}
$\actalt_{1}^{1}$ & $-1+\varepsilon$ & $1-\varepsilon$ & $1-\varepsilon$ & $-1+\varepsilon$ \\
\cline{2-5}
\multicolumn{5}{c}{$\game_{2}(\varepsilon)$}
\end{tabu}} .
\medskip
\end{multicols}\end{center}
The game $\game_{1}$ is a game with duplicate strategies $\actalt_{0}^{1}$ and $\actalt_{1}^{1}$. If we remove the duplicate strategy, we are led to the well-known matching-pennies game which is a harmonic normalized game and therefore in its decomposition the potential and nonstrategic component would be zero. 
The game $\game_{2}(\varepsilon)$ is ``almost'' a game with duplicate strategies, i.e., $\lim_{\varepsilon \rightarrow 0} \game_{2}(\varepsilon)=\game_{1}$, however it cannot be reduced. Applying the  decomposition result in \citet{CanMenOzdPar:MOR2011} to $\game_{2}(\varepsilon)$, which is neither normalized nor a game with duplicate strategies, leads to the following normalized and nonstrategic components:
%\begin{center}
\begin{multicols}{2}
\medskip
{\tabulinesep=1.2mm
\begin{tabu}{ r|cc|cc| }
\multicolumn{1}{r}{}
 &  \multicolumn{2}{c}{$\act^{2}$}
 & \multicolumn{2}{c}{$\actalt^{2}$} \\
\cline{2-5}
$\act^{1}$ & $\frac{4-\varepsilon}{3}$ & $-1$ & $-\frac{4-\varepsilon}{3}$ & $1$ \\
\cline{2-5}
$\actalt_{0}^{1}$ & $-\frac{2+\varepsilon}{3}$ & $1$ & $\frac{2+\varepsilon}{3}$ & $-1$ \\
\cline{2-5}
$\actalt_{1}^{1}$ & $-\frac{2(1-\varepsilon)}{3}$ & $1-\varepsilon$ & $\frac{2(1-\varepsilon)}{3}$ & $-1+\varepsilon$ \\
\cline{2-5}
\multicolumn{5}{c}{$\game_{2,\No}(\varepsilon)$}
\end{tabu}} 
\medskip

\columnbreak
\medskip
{\tabulinesep=1.2mm
\begin{tabu}{ r|cc|cc| }
\multicolumn{1}{r}{}
 &  \multicolumn{2}{c}{$\act^{2}$}
 & \multicolumn{2}{c}{$\actalt^{2}$} \\
\cline{2-5}
$\act^{1}$ & $-\frac{1-\varepsilon}{3}$ & $0$ & $\frac{1-\varepsilon}{3}$ & $0$ \\
\cline{2-5}
$\actalt_{0}^{1}$ & $-\frac{1-\varepsilon}{3}$ & $0$ & $\frac{1-\varepsilon}{3}$ & $0$ \\
\cline{2-5}
$\actalt_{1}^{1}$ & $-\frac{1-\varepsilon}{3}$ & $0$ & $\frac{1-\varepsilon}{3}$ & $0$ \\
\cline{2-5}
\multicolumn{5}{c}{$\game_{2,{\NS}}(\varepsilon)$}
\end{tabu}} .
\medskip

\end{multicols}
%\end{center}
Hence, the decomposition map is not continuous since at the limit we obtain two nonzero components instead of one: 

\begin{center}
\begin{multicols}{2}
\medskip
{\tabulinesep=1.2mm
\begin{tabu}{ r|cc|cc| }
\multicolumn{1}{r}{}
 &  \multicolumn{2}{c}{$\act^{2}$}
 & \multicolumn{2}{c}{$\actalt^{2}$} \\
\cline{2-5}
$\act^{1}$ & $\frac{4}{3}$ & $-1$ & $-\frac{4}{3}$ & $1$ \\
\cline{2-5}
$\actalt_{0}^{1}$ & $-\frac{2}{3}$ & $1$ & $\frac{2}{3}$ & $-1$ \\
\cline{2-5}
$\actalt_{1}^{1}$ & $-\frac{2}{3}$ & $1$ & $\frac{2}{3}$ & $-1$ \\
\cline{2-5}
%\multicolumn{5}{c}{$\game_{2,\mathcal{NO}}(\varepsilon)$}
\end{tabu}} 
\medskip

\columnbreak
\medskip
{\tabulinesep=1.2mm
\begin{tabu}{ r|cc|cc| }
\multicolumn{1}{r}{}
 &  \multicolumn{2}{c}{$\act^{2}$}
 & \multicolumn{2}{c}{$\actalt^{2}$} \\
\cline{2-5}
$\act^{1}$ & $-\frac{1}{3}$ & $0$ & $\frac{1}{3}$ & $0$ \\
\cline{2-5}
$\actalt_{0}^{1}$ & $-\frac{1}{3}$ & $0$ & $\frac{1}{3}$ & $0$ \\
\cline{2-5}
$\actalt_{1}^{1}$ & $-\frac{1}{3}$ & $0$ & $\frac{1}{3}$ & $0$ \\
\cline{2-5}
%\multicolumn{5}{c}{$\game_{2,{\mathcal{NS}}}(\varepsilon)$}
\end{tabu}} .
\medskip

\end{multicols}
\end{center}

\end{remark}

%----------------------- SECTION ------------------------------

\appendix

\section{Proofs} \label{app:proofs}

To prove the equilibrium characterization in $(\muprof,\nsgameg)$-harmonic game (\cref{th:mu-eta-mixed-equilibrium}), we need some notation that is naturally introduced in the proof of the decomposition result. 
Hence, the proofs of \cref{suse:potential-harmonic} are postponed after the proofs of \cref{suse:decomp1,suse:mu-eta-decomposition}.
We remind the reader that $\mumeas$ indicates the product measure of $\parens*{\mumeas^{\play}}_{\play\in\players}$.

%----------------------- SUBSECTION ---------------------------------

\subsection*{Proofs of \cref{suse:decomp1}}

We prove \cref{pr:direct-sum} by describing $\NoG$ and $\NSG$ through orthogonal projection operators. 
The notation $\Id_{\setA}$ stands for the identity operator over a set $\setA$.
First, we define for any $\play\in\players$, the linear operators $\opLam^{\play} : \Czero\to \Czero$ and $\opPi^{\play} \colon \Czero\to\Czero$ as follows:
\begin{subequations}
\label{eq:Lambda-Pi}
\begin{align}
\label{eq:Lambda}
\opLam^{\play}(\payoff^{\play})(\act^{\play},\actprof^{-\play})
&=\sum_{\actalt^{\play} \in \actions^{\play}}\norml{\mumeas}^{\play}(\actalt^{\play})\payoff^{\play}(\actalt^{\play},\actprof^{-\play}),\\
\label{eq:Pi}
\opPi^{\play} 
&= \Id_{\Czero} -\opLam^{\play}.
\end{align}
\end{subequations}
We call $\NSG^{\play} \subset \Czero$ the set of functions $\func$ for which there exists $\nonstratf \colon \actions^{-\play} \to \R$ such that
\begin{align}\label{eq:ns}
\forall \act^{\play}\in\actions^{\play},\ \forall \actprof^{-\play}\in \actions^{-\play},\ \func(\act^{\play},\actprof^{-\play})=\nonstratf(\actprof^{-\play}).
\end{align}
We call  $\NoG^{\play} \subset \Czero$ the set of functions $\func$ such that
\begin{equation}
\label{eq:no}
\forall \actprof^{-\play}\in \actions^{-\play},\ \sum_{\actprof^{\play} \in \actions^{\play}} \mumeas^{\play}(\act^{\play})\func(\act^{\play},\actprof^{-\play})=0.
\end{equation}

\begin{lemma}\label{le:chara-comp}
For $\opLam$ and $\opPi$ defined as in \cref{eq:Lambda-Pi}, the following properties hold:
\begin{enumerate}[\upshape(a)]
\item \label{it:le:chara-comp-1}
$\opPi^{\play}$ and $\opLam^{\play}$ are projections onto $\Czero$.
\item\label{it:le:chara-comp-2} For every $\func_{\Ker} \in \Ker(\opLam^{\play})$ and for every $\func_{\Ima} \in \Ima(\opLam^{\play})$,
\[
\inner*{\gammameas^{\play} \func_{\Ker}}{\gammameas^{\play} \func_{\Ima}}_{0}=0.
\]
\item \label{it:le:chara-comp-3}
We have $\NoG^{\play}=\Ker(\opLam^{\play})$ and $\NSG^{\play}=\Ima(\opLam^{\play})$.
\end{enumerate}
\end{lemma}

\begin{proof}
\ref{it:le:chara-comp-1} 
In view of \cref{eq:Lambda-Pi}, we have 
\begin{equation}\label{eqKerIma}
\Ker(\opPi^{\play})=\Ima(\opLam^{\play})\quad\text{and}\quad \Ker(\opLam^{\play})=\Ima(\opPi^{\play}).
\end{equation}
Thus, we only need to prove that $\opLam^{\play} \circ \opLam^{\play}=\opLam^{\play}$. Indeed,
\begin{align*}
\opLam^{\play} \circ \opLam^{\play}(\func^{\play})(\act^{\play},\actprof^{-\play})
&= \sum_{\actalt^{\play} \in \actions^{\play}}\norml{\mumeas}^{\play}(\actalt^{\play})\parens*\opLam^{\play}(\func^{\play})(\actalt^{\play},\actprof^{-\play})\\
&=\sum_{\actalt^{\play} \in \actions^{\play}}\norml{\mumeas}^{\play}(\actalt^{\play})\sum_{\actr^{\play} \in \actions^{\play}} \norml{\mumeas}^{\play}(\actr^{\play})\func^{\play}(\actr^{\play},\actprof^{-\play})\\
&=\sum_{r^{\play} \in \actions^{\play}}\norml{\mumeas}^{\play}(\actr^{\play})\parens*{\sum_{\actalt^{\play} \in \actions^{\play}} \norml{\mumeas}^{\play}(\actalt^{\play})}\func^{\play}(\actr^{\play},\actprof^{-\play})\\
&=\sum_{\actr^{\play} \in \actions^{\play}} \norml{\mumeas}^{\play}(\actr^{\play})\func^{\play}(\actr^{\play},\actprof^{-\play}),
\end{align*}
where the last equality stems from the fact that $\norml{\mumeas}^{\play}$ is probability distribution. 
Hence, $\opLam^{\play}$ is a projection operator.

\noindent
\ref{it:le:chara-comp-2} Let $\func_{\Ker} \in \Ker(\opLam^{\play})$ and $\func_{\Ima} \in \Ima(\opLam^{\play})$. 
By definition of $\opLam^{\play}$, any function in $\Ima(\opLam^{\play})$ does not depend on $\act^{\play}$, i.e., $\Ima(\opLam^{\play})\subset \NSG^{\play}$. 
Hence, there exists $\nonstratf \colon \actions^{-\play} \to \R$ such that $\nonstratf(\actprof^{-\play})=\func_{\Ima}(\act^{\play},\actprof^{-\play})$. 
Therefore,
\begin{align*}
\inner*{\gammameas^{\play}\func_{\Ker}}{\gammameas^{\play}\func_{\Ima}}_{0} 
&= \sum_{\actprof\in\actions}\mumeas(\actprof)(\gammameas^{\play}(\actprof^{-\play}))^{2}\func_{\Ker}(\actprof)\func_{\Ima}(\actprof)\\
&=\mumeas^{\play}(\actions^{\play})\sum_{\actprof^{-\play} \in \actions^{-\play}}\mumeas^{-\play}(\actprof^{-\play})(\gammameas^{\play}(\actprof^{-\play}))^{2}\nonstratf(\actprof^{-\play}) \parens*{\sum_{\act^{\play}\in\actions^{\play}}\norml{\mumeas}^{\play}(\act^{\play})\func_{\Ker}(\act^{\play},\actprof^{-\play})}\\
&=0,
\end{align*}
since $\func_{\Ker} \in \Ker(\opLam^{\play})$. 

\ref{it:le:chara-comp-3} In view of \cref{eq:no,eq:Lambda}, we have
\[
\sum_{\act^{\play} \in \actions^{\play}} \mumeas^{\play}(\act^{\play})\func(\act^{\play},\actprof^{-\play})
=\mumeas^{\play}(\actions^{\play}) \parens*{\sum_{\act^{\play} \in \actions^{\play}} \norml{\mumeas}^{\play}(\act^{\play})\func(\act^{\play},\actprof^{-\play})}
=\mumeas^{\play}(\actions^{\play})\opLam^{\play}(\func)(\actprof).
\]
Hence, we have $\Ker(\opLam^{\play})=\NoG^{\play}$. 
We saw in the proof of \ref{it:le:chara-comp-2} that $\Ima(\opLam^{\play}) \subset \NSG^{\play}$. 
Reciprocally, a function $\func\in \NSG^{\play}$ does not depend on $\act^{\play}$ and an immediate computation yields $\func=\opLam^{\play} (\func)$. 
Therefore $\func \in \Ima(\opLam^{\play})$ and we obtain $\Ima(\opLam^{\play})= \NSG^{\play}$.
\end{proof}

We now define $\opLam:\Czero^{\nplayers} \to \Czero^{\nplayers}$ and $\opPi :\Czero^{\nplayers} \to \Czero^{\nplayers}$ as 
\begin{align}
\opLam(\funcvec) &=\parens{\opLam^{1}(\func^{1}),\dots,\opLam^{\nplayers}(\func^{\nplayers})},\label{eq:Lproduct} \\
\opPi(\funcvec)&= \parens{\opPi^{1}(\func^{1}),\dots,\opPi^{\nplayers}(\func^{\nplayers})},
\end{align}
with $\funcvec=(\func^{\play})_{\play\in\players}$.

\begin{lemma} \label{le:chara-NO-NS}
\begin{enumerate}[\upshape(a)]
\item\label{it:le:chara-NO-NS-1} We have
\begin{align}
\NoG&=\braces*{ \game\in \games : \forall \play\in\players, \payoff^{\play} \in \NoG^{\play}}=\Ker(\opLam),\\
\NSG&=\braces*{\game \in \games : \forall \play\in\players, \payoff^{\play} \in \NSG^{\play}}=\Ima(\opLam).
\end{align}
\item\label{it:le:chara-NO-NS-2}
$\opPi$ and $\opLam$ are orthogonal projections onto $\games$ for $\inner{\argdot}{\argdot}_{\muprof,\nsgameg}$.
\end{enumerate}
\end{lemma}

\begin{proof}
\ref{it:le:chara-NO-NS-1} 
This follows immediately from \cref{le:chara-comp}\ref{it:le:chara-comp-3}.

\noindent
\ref{it:le:chara-NO-NS-2} We have $\Ker(\opLam)=\NoG$ and $\Ima(\opLam)=\NSG$, where $\opLam$ is defined as in \cref{eq:Lproduct}. 
Let $\game_{\NoG} \in \NoG$ and $\game_{\NSG} \in \NSG$. 
Then, using the definition of the scalar product of \cref{eq:g1g2mu} and \cref{le:chara-comp}\ref{it:le:chara-comp-2}, we obtain
\begin{align*}
\inner{\game_{\NoG}}{\game_{\NSG}}_{\muprof,\nsgameg} 
= \sum_{\play\in\players} \mumeas^{\play}(\actions^{\play}) \inner*{\gammameas^{\play} \payoff_{\NoG}^{\play}}{\gammameas^{\play}\payoff_{\NSG}^{\play}}_{0}= 0,
\end{align*}
since $\payoff^{\play}_{\NoG} \in \Ker(\opLam^{\play})$ and $\payoff^{\play}_{\NSG} \in \Ima(\opLam^{\play})$ for every $\play\in\players$.
\end{proof}

\begin{proof}[Proof of \cref{pr:direct-sum}]
This is an immediate corollary of \cref{le:chara-NO-NS}. We showed that $\NoG$ and $\NSG$ are orthogonal subspaces and further that any game $\game$ can be decomposed into $\opLam(\game) \in \NSG$ and $(\Id_{\games}-\opLam )(\game) \in \NoG$.
\end{proof}

%----------------------- SUBSECTION ---------------------------------

\subsection*{Proofs of \cref{suse:mu-eta-decomposition}}

We first give the definition of flow associated to a game and then, in \cref{pr:operator}, we characterize the different classes of games introduced in \cref{de:classes-of-games}. 
Then, from these characterizations and from the properties of the Moore-Penrose pseudo-inverse,  we  deduce  the proof of \cref{th:ortho-sum}.

Let \begin{equation}
\label{eq:C-one}
\Cone \coloneqq \braces{\flow \colon \actions\times \actions \to \R \text{ such that } \flow(\actprof,\actaltprof)=-\flow(\actaltprof,\actprof), \forall  \actprof,\actaltprof \in \actions}
\end{equation} 
be the set of flows. 
We endow $\Cone$ with the following inner product:
\begin{equation} 
\label{eq:inn-prod-1}
\forall \flow,\flowalt\in\Cone, \quad \inner*{\flow}{\flowalt}_{1}
=\frac{1}{2}\sum_{\actprof,\actaltprof\in\actions}\mumeas(\actprof)\mumeas(\actaltprof)\flow(\actprof,\actaltprof)\flowalt(\actprof,\actaltprof).
\end{equation}

To any game $\game$, we associate an undirected graph as follows. Given a pair of strategy profiles $(\actprof,\actaltprof)\in\actions\times\actions$, if there exists a unique $\play\in\players$, such that $\act^{\play} \neq \actalt^{\play}$ then $(\actprof,\actaltprof)$ will be referred to as an $\play$-\emph{comparable profile pair}. 
The set of $\play$-comparable profile pairs will be denoted by  $\compar^{\play} \subset \actions \times \actions$.  
For any two different players $\play$ and $\playalt$ we have: $\compar^{\play} \cap \compar^{\playalt}=\emptyset$. 
We call $\compar$ the set of \emph{comparable profile pairs}, i.e.,  $\compar=\cup_{\play\in\players}\compar^{\play}$. 
Following \citet{CanMenOzdPar:MOR2011}, for any given game $\game$, we associate to player $\play$ an undirected graph defined as $\graph^{\play}\coloneqq(\actions,\compar^{\play})$ 
and we further associate to the game $\game$ the disjoint union of the graphs $\graph^{\play}$, denoted by $\graph\coloneqq(\actions,\compar)$. 

For any $\play\in\players$, let $\symmfunc^{\play} \colon \actions\times\actions \to \R$ be the nonnegative symmetric function defined as
\begin{equation}
\label{eq:symm-func}
{\symmfunc^{\play}}(\actprof,\actaltprof)=
\begin{dcases}
\frac{1}{\sqrt{\mumeas^{-\play}(\actprof^{-\play})}}&\text{if }(\actprof,\actaltprof) \in \compar^{\play},\\
0&\text{otherwise}.
\end{dcases} 
\end{equation}
Recalling that any pair of strategy profiles cannot be comparable for  more than one player, we have:
\begin{equation}
\symmfunc^{\play}(\actprof,\actaltprof)\symmfunc^{\playalt}(\actprof,\actaltprof) = 0,\quad\text{for all $\playalt\neq\play$ and $\actprof,\actaltprof\in\actions$}. 
\end{equation}

To any player $\play$, we associate the \emph{partial gradient operator} $\gradop^{\play}:\Czero \rightarrow \Cone$, defined for any $\func\in\Czero$ as follows:
\begin{equation}
\gradop^{\play}(\func)({\actprof,\actaltprof})
=\symmfunc^{\play}(\actprof,\actaltprof) \parens*{\func(\actaltprof)-\func(\actprof)}.\label{eq:deltaW}
\end{equation}
The gradient operator $\gradop$ on $\graph$ is defined as $\gradop=\sum_{\play\in\players}\gradop^{\play}$.

We now introduce the adjoint operators. 
Recall that in \cref{eq:inn-prod-0} we have considered the following inner product on $\Czero$:
\[
\forall \ \func,\funcalt \in \Czero, \ \inner{\func}{\funcalt}_{0}=\sum_{\actprof\in\actions} \mumeas(\actprof)\func(\actprof)\funcalt(\actprof).
\]
The adjoint of $\gradop^{\play}$, denoted by $\gradop^{\play*} \colon \Cone\to\Czero$ is the unique linear operator satisfying: 
\begin{equation} 
\label{eq:deltak-fk}
\inner{\gradop^{\play}\func}{\flow}_{1}=\inner{\func}{\gradop^{\play*}\flow}_{0},
\end{equation}
for any $\func\in\Czero$, $\flow\in\Cone$. 
By linearity of dual operations, we obtain that the dual of $\gradop$ satisfies $\gradop^{*}=\sum_{\play\in\players} \gradop^{\play*}$. 
Moreover, we have the following explicit expression for $\gradop^{\play*}$.

\begin{proposition} \label{pr:adjoint}  
The adjoint $\gradop^{\play*}$ of the gradient operator $\gradop^{\play}$ satisfies for any $\flow\in\Cone$:
\begin{align}
\label{eq:adjoint}
\forall \actprof\in\actions, \quad \gradop^{\play*}\flow(\actprof)
&=-\sum_{\actaltprof\in\actions}\mumeas(\actaltprof)\symmfunc^{\play}(\actprof,\actaltprof)\flow(\actprof,\actaltprof)\\
&=-\sum_{\actaltprof:(\actprof,\actaltprof)\in \compar^{\play}}\mumeas^{\play}(\actalt^{\play})\sqrt{\mumeas^{-\play}(\actprof^{-\play})}\flow(\actprof,\actaltprof).
\end{align}
\end{proposition}

\begin{proof}
We introduce the  basis $\parens{\basis_{\actrprof}}_{\actrprof\in\actions}$ of $\Czero$,  defined as
\begin{equation} \label{eq:basis1}
\basis_{\actrprof}(\actprof)=
\begin{dcases}\frac{1}{\sqrt{\mumeas(\actrprof)}}&\text{if }\actprof=\actrprof,\\
0 &\text{otherwise}.
\end{dcases}
\end{equation}
This basis is orthonormal with respect to the inner product in  \cref{eq:inn-prod-0}.
For any $\flow\in\Cone$, we have
\begin{equation*}
\gradop^{\play*}\flow = \sum_{\actrprof \in S}\inner*{\basis_{\actrprof}}{\gradop^{\play*}\flow}_{0} \basis_{\actrprof}
\end{equation*}
and thus, 
\begin{equation*}
\gradop^{\play*}\flow(\actprof_{0})=\frac{1}{\sqrt{\mumeas(\actprof_{0})}}\inner{ \basis_{\actprof_{0}}}{\gradop^{\play*}\flow}_{0}.
\end{equation*} 
By using the relation between $\gradop^{\play*}$ and $\gradop^{\play}$ in \cref{eq:deltak-fk} and then the definition of $\gradop$ in \cref{eq:deltaW}, we get
\begin{align*}
\gradop^{\play*}\flow(\actprof_{0})&=\frac{1}{\sqrt{\mumeas(\actprof_{0})}}\inner*{\gradop\basis_{\actprof_{0}}}{\flow}_{1}\\
& =\frac{1}{2\sqrt{\mumeas(\actprof_{0})}}\sum_{\actprof,\actaltprof\in\actions}\mumeas(\actprof)\mumeas(\actaltprof)\parens*{\gradop\basis_{\actprof_{0}}}(\actprof,\actaltprof)\flow(\actprof,\actaltprof)\\
& = \frac{1}{2\sqrt{\mumeas(\actprof_{0})}}\left(\sum_{\actprof,\actaltprof\in\actions}\mumeas(\actprof)\mumeas(\actaltprof)\symmfunc^{\play}(\actprof,\actaltprof)\basis_{\actprof_{0}}(\actaltprof)\flow(\actprof,\actaltprof)\right.\\
&\quad\left.\quad-\sum_{\actprof,\actaltprof\in\actions}\mumeas(\actprof)\mumeas(\actaltprof)\symmfunc^{\play}(\actprof,\actaltprof)\basis_{\actprof_{0}}(\actprof)\flow(\actprof,\actaltprof)\right).
\end{align*}
Using the definition of $\basis_{\actprof_{0}}$ in \cref{eq:basis1}, we obtain
\begin{align*}
\gradop^{\play*}\flow(\actprof_{0})
&=\frac{1}{2\sqrt{\mumeas(\actprof_{0})}}\left(\sum_{\actprof\in\actions}\mumeas(\actprof)\mumeas(\actprof_{0})\frac{1}{\sqrt{\mumeas(\actprof_{0})}}\symmfunc^{\play}(\actprof,\actprof_{0})\flow(\actprof,\actprof_{0}) \right.\\
& \quad \left.-\sum_{\actaltprof\in\actions}\mumeas(\actaltprof)\mumeas(\actprof_{0})\frac{1}{\sqrt{\mumeas(\actprof_{0})}}\symmfunc^{\play}(\actprof_{0},\actaltprof)\flow(\actprof_{0},\actaltprof)\right)\\
 &=\frac{1}{2}\parens*{\sum_{\actprof\in\actions}\mumeas(\actprof)\symmfunc^{\play}(\actprof,\actprof_{0})\flow(\actprof,\actprof_{0})-\sum_{\actaltprof\in\actions}\mumeas(\actaltprof)\symmfunc^{\play}(\actprof_{0},\actaltprof)\flow(\actprof_{0},\actaltprof)}\\
&=-\frac{1}{2}\sum_{\actprof\in S}\mumeas(\actprof)\symmfunc^{\play}(\actprof_{0},\actprof)\flow(\actprof_{0},\actprof)-\frac{1}{2}\sum_{\actaltprof \in S}\mumeas(\actaltprof)\symmfunc^{\play}(\actprof_{0},\actaltprof)\flow(\actprof_{0},\actaltprof)\\
&=-\sum_{\actaltprof \in S}\mumeas(\actaltprof)\symmfunc^{\play}(\actprof_{0},\actaltprof)\flow(\actprof_{0},\actaltprof),
\end{align*}
where the third equality is due to the skew-symmetric structure of $\flow$ and the last equality is simply obtained by merging the two sums.
\end{proof}

To $\graph$ we associate the \emph{joint embedding operator} $\embop \colon \Czero^{\nplayers} \to \Cone$. 
The operator $\embop$ maps a game $\game$ into a flow $\embop(\game)$. 
It is defined for any $\game \in \Czero^{\nplayers}$ as
\begin{align}\label{eq:D}
\embop(\game)=\sum_{\play\in\players}\gradop^{\play} (\gammameas^{\play} \payoff^{\play}).
\end{align}

\begin{lemma}
\label{le:deltas}
The following properties hold:
\begin{enumerate}[\upshape(a)]
\item
\label{it:le:delta-1}
For all $\play\neq\playalt$, we have
\begin{equation}\label{eq:delta-1}
\gradop^{\play*} \gradop^{\playalt} = 0.
\end{equation}

\item
\label{it:le:delta-2}
For all $\game\in \Czero^{\nplayers}$, we have
\begin{equation}\label{eq:delta-2}
\gradop^{*}\embop(\game)
=\sum_{\play\in\players}\gradop^{\play*}\gradop^{\play}(\gammameas^{\play} \payoff^{\play}).
\end{equation}

\item
\label{it:le:delta-3}
For all $\func\in\Czero$, we have
\begin{equation}\label{eq:delta-3}
\gradop^{\play*}\gradop^{\play}(\gammameas^{\play}\func)
=\gammameas^{\play}\gradop^{\play*}\gradop^{\play}(\func).
\end{equation}
\end{enumerate}
\end{lemma}

\begin{proof}
\noindent
\ref{it:le:delta-1}
Let $\func\in\Czero$.
By \cref{eq:adjoint}, for all $\actprof\in\actions$, we have
\begin{align*}
\gradop^{\play*}(\gradop^{\playalt}\func)(\actprof)&=-\sum_{\actaltprof\in\actions}\mumeas(\actaltprof)\symmfunc^{\play}(\actprof,\actaltprof)(\gradop^{\playalt}\func)(\actprof,\actaltprof)\\
&=-\sum_{\actaltprof\in\actions}\mumeas(\actaltprof)\symmfunc^{\play}(\actprof,\actaltprof)\symmfunc^{\playalt}(\actprof,\actaltprof)\parens*{\func(\actaltprof)-\func(\actprof)}\\
&=0,
\end{align*}
since  for $\play\neq\playalt$, we have $\symmfunc^{\play}(\actprof,\actaltprof)\symmfunc^{\playalt}(\actprof,\actaltprof)=0$. 

\noindent
\ref{it:le:delta-2}
Let $\game \in \games$. We then get
\begin{equation*}
\gradop^{*}\embop(\game)=\sum_{\play\in\players}\gradop^{\play*} \parens*{\sum_{\playalt\in\players}\gradop^{\playalt}(\gammameas^{\playalt} \payoff^{\playalt})}
=\sum_{\play\in\players}\gradop^{\play*}\gradop^{\play}(\gammameas^{\play} \payoff^{\play}),
\end{equation*}
since, by \ref{it:le:delta-1}, all cross-products are equal to $0$.

\noindent
\ref{it:le:delta-3}
Let $\func\in\Czero$. We have, for all $\actprof\in\actions$,
\begin{align*}
\gradop^{\play*}(\gradop^{\play} \gammameas^{\play}\func)(\actprof)
&=-\sum_{\actaltprof\in\actions}\mumeas(\actaltprof)\symmfunc^{\play}(\actprof,\actaltprof)(\gradop^{\play} \gammameas^{\play}\func)(\actprof,\actaltprof)\\
& \hspace{-10mm} =-\sum_{\actalt^{\play}\in\actions^{\play}} \mumeas(\actalt^{\play},\actprof^{-\play})\symmfunc^{\play}\parens*{\actprof,(\actalt^{\play},\actprof^{-\play})}\symmfunc^{\play}\parens*{\actprof,(\actalt^{\play},\actprof^{-\play})}\gammameas^{\play}(\actprof^{-\play})\parens*{\func(\act^{\play},\actprof^{-\play})-\func(\actalt^{\play},\actprof^{-\play})}\\
&\hspace{-10mm} =\gammameas^{\play}(\actprof^{-\play})\parens*{-\sum_{\actalt^{\play}\in\actions^{\play}} \mumeas(\actalt^{\play},\actprof^{-\play})\symmfunc^{\play}\parens*{\actprof,(\actalt^{\play},\actprof^{-\play})}\symmfunc^{\play}\parens*{\actprof,(\actalt^{\play},\actprof^{-\play})}\parens*{\func(\act^{\play},\actprof^{-\play})-\func(\actalt^{\play},\actprof^{-\play})}}\\
& \hspace{-10mm}= \gammameas^{\play}(\actprof^{-\play})\gradop^{\play*}(\gradop^{\play}\func)(\actprof). \qedhere
\end{align*}
\end{proof}

At this point, we can relate the classes of games introduced in \cref{de:classes-of-games} to the previously defined operators. 

\begin{proposition}\label{pr:operator} 
We have:
\begin{enumerate}[\upshape(a)]
\item
\label{it:pr:operator-a}
 $\NSG=\braces{ \game\in \games \colon \embop(\game)=0 }$.
\item
\label{it:pr:operator-b}
 $\nsgameg\PGd=\braces{ \game\in \games \colon \embop(\game) \in \Ima \gradop}=(\embop)^{-1}(\Ima(\gradop))$.
\item
\label{it:pr:operator-c}
 $(\muprof,\nsgameg)\HGd=\braces{\game\in \games \colon  \embop(\game)\in \Ker\gradop^{*}}$.
\end{enumerate}
\end{proposition} 

\begin{proof}
\ref{it:pr:operator-a} 
Let $\game \in \games$.  \cref{eq:D} implies that $\embop(\game) \in 0$ if and only if, for all $\play\in\players$ and for all $\actprof,\actaltprof \in S$, we have
\begin{equation}\label{eq:Wi}
\symmfunc^{\play}(\actprof,\actaltprof)\parens*{\gammameas^{\play}(\actprof^{-\play})\payoff^{\play}(\actprof)-\gammameas^{\play}(\actprof^{-\play})\payoff^{\play}(\actaltprof)}=0.
\end{equation}
Since $\symmfunc^{\play}(\actprof,\actaltprof)$ is strictly positive on $\compar^{\play}$ and $\gammameas^{\play}$ is strictly positive, \cref{eq:Wi} holds if and only if, for all $\play\in\players$, for all $\actprof^{-\play}\in\actions^{-\play}$, and for all $\act^{\play},\actalt^{\play}\in\actions^{\play}$, we have
\[ 
\payoff^{\play}(\act^{\play},\actprof^{-\play})-\payoff^{\play}(\actalt^{\play},\actprof^{-\play})=0.
\]
Hence $\payoff^{\play}$ does not depend on player~$\play$'s strategy, which, by \cref{de:classes-of-games}\ref{it:classes-nonstrategic}, means that the game is nonstrategic.

\noindent
\ref{it:pr:operator-b} Let $\game\in\games$. Then, by \cref{eq:D}, $\embop(\game) \in \Ima(\gradop)$ if and only if there exists $\funcphi\in\Czero$ such that for every $\play\in\players$, and $\actprof,\actaltprof\in\actions$, we have
\[
\symmfunc^{\play}(\actprof,\actaltprof)\parens*{\gammameas^{\play}(\actprof^{-\play})\payoff^{\play}(\actprof)-\gammameas^{\play}(\actprof^{-\play})\payoff^{\play}(\actaltprof)}
=\symmfunc^{\play}(\actprof,\actaltprof)\parens*{\funcphi(\actprof)-\funcphi(\actaltprof)}.
\]
We equivalently have, for any $\act^{\play},\actalt^{\play}\in\actions^{\play}$ and any $\actprof^{-\play}\in\actions^{-\play}$,
\begin{align*}
\funcphi(\act^{\play},\actprof^{-\play})-\funcphi(\actalt^{\play},\actprof^{-\play})& =\gammameas^{\play}(\actprof^{-\play})\parens*{\payoff^{\play}(\act^{\play},\actprof^{-\play})-\payoff^{\play}(\actalt^{\play},\actprof^{-\play})}.
\end{align*}
Hence, the result follows from \cref{de:classes-of-games}\ref{it:classes-eta-potential}.

\noindent
\ref{it:pr:operator-c} 
Let $\game$ be a $(\muprof,\nsgameg)$-harmonic game. 
In view of  \ref{it:le:delta-2}, we can write the condition in terms of gradient and adjoint operator. 
By replacing them with their explicit expression, we obtain that, for any $\actprof\in\actions$, 
\begin{align*}
\gradop^{*}(\embop(\game))(\actprof)
&=\sum_{\play\in\players}\gradop^{\play*}\gradop^{\play} (\gammameas^{\play}\payoff^{\play})(\actprof)\\
&=\sum_{\play\in\players} \parens*{-\sum_{\actaltprof \in \actions}\mumeas(\actaltprof)\symmfunc^{\play}(\actprof,\actaltprof)\symmfunc^{\play}(\actprof,\actaltprof)\parens*{(\gammameas^{\play}\payoff^{\play})(\actprof)-(\gammameas^{\play}\payoff^{\play})(\actaltprof)}}\\
&=-\sum_{\play\in\players} \sum_{\actalt^{\play}\in \actions^{\play}} \mumeas(\actalt^{\play},\actprof^{-\play})\frac{1}{\mumeas^{-\play}(\actprof^{-\play})}\gammameas^{\play}(\actprof^{-\play})\parens*{\payoff^{\play}(\act^{\play},\actprof^{-\play})-\payoff^{\play}(\actalt^{\play},\actprof^{-\play})}\\
&=-\sum_{\play\in\players} \sum_{\actalt^{\play}\in\actions^{\play}} \mumeas^{\play}(\actalt^{\play})\gammameas^{\play}(\actprof^{-\play})\parens*{\payoff^{\play}(\act^{\play},\actprof^{-\play})-\payoff^{\play}(\actalt^{\play},\actprof^{-\play})},
\end{align*}
and the result follows from  \cref{de:classes-of-games}\ref{it:classes-mu-eta-harmonic}.
\end{proof}

The aim of the next results is to prove that the space of games is the orthogonal sum of $\nsgameg$-potential $\muprof$-normalized games, $(\muprof,\nsgameg)$-harmonic $\muprof$-normalized games, and nonstrategic games. 
We  first introduce the Moore-Penrose pseudo-inverse. 
Then, we prove several results leading to the above-mentioned orthogonality. 
Finally, we prove that any game can be decomposed in such  a triple by providing an explicit formula that uses the pseudo-inverse. 
These  results yield \cref{th:ortho-sum}.

Let $\gradop^{\play\dag}:\Cone \to \Czero$ be defined as follows:
\begin{itemize}
\item On $\Ima(\gradop^{\play})$, $\gradop^{\play\dag}$ is  the inverse of the restriction of $\gradop^{\play}$ on $\NoG^{\play}$, i.e.,
\begin{align*}
\gradop^{\play\dag}{\restriction_{\Ima(\gradop^{\play})}}=\bracks*{\gradop^{\play}{\restriction_{\NoG^{\play}}}}^{-1}.
\end{align*}
\item on $\Ima(\gradop^{\play})^{\perp}$, $\gradop^{\play\dag}=0$.
\item on $\Cone$, $\gradop^{\play\dag}$ is defined linearly, i.e., 
\begin{equation*}
\gradop^{\play\dag}(\flow)=\gradop^{\play\dag}(\flow_{\Ima} + \flow_{\Ima}^{\perp})=\gradop^{\play\dag}(\flow_{\Ima}),
\end{equation*}
where $\flow_{\Ima} \in \Ima(\gradop^{\play})$.
\end{itemize}
The operator $\gradop^{\play\dag}$ is in fact the Moore-Penrose pseudo-inverse of $\gradop^{\play}$. In particular, $\gradop^{\play\dag}\gradop^{\play}$ is the orthogonal projection onto $\NoG^{\play}$ and thus, 
\begin{equation}\label{eq:deltaidagdeltai}
\gradop^{\play\dag}\gradop^{\play}=\opPi^{\play}.
\end{equation} 
Moreover, $\gradop^{\play} \gradop^{\play\dag}\gradop^{\play} = \gradop^{\play}$.
Furthermore, $(\Id_{\Czero} - \gradop^{\play\dag}\gradop^{\play})$ is the orthogonal projection onto $\NSG^{\play}$ and thus, $(\Id_{\Czero} - \gradop^{\play\dag}\gradop^{\play})=\opLam^{\play}$. 
Moreover, $\gradop^{\play} \gradop^{\play\dag}$ is the orthogonal projection onto $\Ima(\gradop^{\play})$ and $(\Id_{\Cone}-\gradop^{\play} \gradop^{\play\dag})$ is the orthogonal projection onto $\Ker(\gradop^{\play*})$.

We further define $\gradop^{\dag}:\Cone \to \Czero$ as the pseudo-inverse of $\gradop$. 

\begin{lemma}
\label{le:varphi}
Let $\funcphi \in \NoG^{\play}$ and $\widetilde{\funcphi} \in \Czero$ and assume that $\gradop^{\play} \funcphi= \gradop^{\play}\widetilde{\funcphi}$. Then, for all $\funcpsi \in \NoG^{\play}$ we have: $\inner*{\funcphi}{\funcpsi}_{0}=\inner*{\widetilde{\funcphi}}{\funcpsi }_{0}$.
\end{lemma}

\begin{proof}
Since $ \funcphi \in \NoG^{\play}$ and $\gradop^{\play\dag}\gradop^{\play}$ is the identity operator on $\NoG^{\play}$, we have:
\[
\inner*{\funcphi}{\funcpsi}_{0}
=\inner*{\gradop^{\play\dag} \gradop^{\play} \funcphi}{\funcpsi}_{0}
=\inner*{\gradop^{\play\dag}  \gradop^{\play} \widetilde{\funcphi}}{\funcpsi}_{0},
\]
where the last equality follows from the initial assumption. Using the fact that $\gradop^{\play\dag} \gradop^{\play} $ is self-adjoint, that the operator $\gradop^{\play\dag} \gradop^{\play}$ is the identity on $\NoG^{\play}$, and that $\funcpsi \in \NoG^{\play}$, we get
\begin{equation*}
\inner*{\funcphi}{\funcpsi}_{0}
=\inner*{\widetilde{\funcphi}}{\gradop^{\play\dag}\gradop^{\play}  \funcpsi}_{0}
=\inner*{\widetilde{\funcphi}}{\funcpsi}_{0}. \qedhere
\end{equation*}
\end{proof}

\begin{lemma}\label{le:orthcompo}
For any $\play\in\players$, we have
\begin{align}
\gradop^{\play*} \circ \gradop^{\play}=\mumeas^{\play}(\actions^{\play})\opPi^{\play}.
\end{align}
\end{lemma}

\begin{proof}
Using \cref{eq:adjoint} and replacing $\symmfunc^{\play}$ with its value, for all $\func\in \Czero$ and all $\actprof\in\actions$, we have
\begin{align*}
\parens*{\gradop^{\play*} \parens*{\gradop^{\play}\func}}(\actprof)
&=\sum_{\actaltprof\in\actions} \mumeas(\actaltprof) \symmfunc^{\play}(\actprof,\actaltprof)\symmfunc^{\play}(\actprof,\actaltprof) \parens*{\func(\actprof) - \func(\actaltprof)}\\
& = \sum_{\actalt^{\play}\in\actions^{\play}}\mumeas^{\play}(\actalt^{\play})\parens*{\func(\act^{\play},\actprof^{-\play}) - \func(\actalt^{\play},\actprof^{-\play})}.
\end{align*}
By introducing  the probability distribution $\norml{\mumeas}$ associated to $\muprof$, we obtain
\begin{align*}
\parens*{\gradop^{\play*} \parens*{\gradop^{\play}\func}}(\actprof)&= \mumeas^{\play}(\actions^{\play})\sum_{\actalt^{\play} \in \actions^{\play}}\norml{\mumeas}^{\play}(\actalt^{\play})\parens*{\func(\act^{\play},\actprof^{-\play}) - \func(\actalt^{\play},\actprof^{-\play})}\\
&=\mumeas^{\play}(\actions^{\play})\parens*{\opPi^{\play}(\func)}(\actprof),
\end{align*}
which concludes the proof.
\end{proof}

\begin{proposition} \label{pr:normalized-harmonic}
A game $\game$ is a $\muprof$-normalized $(\muprof,\nsgameg)$-harmonic game if and only if 
\[
\sum_{\play\in\players}\mumeas^{\play}(\actions^{\play})\gammameas^{\play}\payoff^{\play}=0 \quad\text{and}\quad \opPi(\game)=\game.
\]
\end{proposition}

\begin{proof}
Let $\game \in \NoG$. By \cref{le:chara-comp}\ref{it:le:chara-comp-3}, this holds if and only if $\opLam(\game)=0$, which is equivalent to $\opPi(\game)=\game$. 
Therefore, for all $\actprof\in\actions$, we get
\begin{align*}
\sum_{\play\in\players} \mumeas^{\play}(\actions^{\play})\gammameas^{\play}(\actprof^{-\play})\payoff^{\play}(\actprof)  &=\sum_{\play\in\players} \mumeas^{\play}(\actions^{\play})\gammameas^{\play}(\actprof^{-\play})\opPi^{\play}(\payoff^{\play})(\actprof)\\
&=\sum_{\play\in\players}\gammameas^{\play}(\actprof^{-\play})\gradop^{\play*}\gradop^{\play} \payoff^{\play}(\actprof)\\
&=\sum_{\play\in\players}\gradop^{\play*}\gradop^{\play} (\gammameas^{\play}\payoff^{\play})(\actprof)\\
&=\gradop^{*}\parens*{\embop(\game)}(\actprof),
\end{align*}
where the second equality follows from \cref{le:orthcompo}, the third from \cref{le:deltas}\ref{it:le:delta-3}, and the last from \cref{le:deltas}\ref{it:le:delta-2}.
This, together with \cref{pr:operator}\ref{it:pr:operator-c}, completes the proof.
\end{proof}

\begin{proposition}
\label{pr:ortho-jeux}
The sets of games $\NoG\cap(\muprof,\nsgameg)\HGd$, $\NoG\cap\nsgameg\PGd$, and $\NSG$ are orthogonal with respect to the inner product in \cref{eq:g1g2mu}.
\end{proposition}

\begin{proof}
We know by \cref{pr:direct-sum} that $\NoG$ and $\NSG$ are orthogonal spaces. 
It is sufficient to prove that $\NoG\cap(\muprof,\nsgameg)\HGd$ and $\NoG\cap\nsgameg\PGd$ are also orthogonal spaces.

To prove orthogonality between $\muprof$-normalized potential games and $\muprof$-normalized $(\muprof,\nsgameg)$-harmonic games, let $\game_{\Pot} \in \NoG \cap\nsgameg\PGd$ and $\game_{\Har} \in \NoG \cap(\muprof,\nsgameg)\HGd$. 

In view of \cref{eq:g1g2mu}, we get
\begin{align*}
\inner*{\game_{\Pot}}{\game_{\Har}}_{\muprof,\nsgameg} 
&=\sum_{\play\in\players} \mumeas^{\play}(\actions^{\play}) \inner*{\gammameas^{\play}\payoff_{\Pot}^{\play}}{\gammameas^{\play} \payoff_{\Har}^{\play}}_{0}\\
&=\sum_{\play\in\players} \inner*{\gammameas^{\play}\payoff_{\Pot}^{\play}}{\mumeas^{\play}(\actions^{\play})\gammameas^{\play}\payoff_{\Har}^{\play}}_{0}.
\end{align*}
Since $\game_{\Pot} \in \NoG \cap \nsgameg\PGd$, there exists $\funcphi$ such that for any $\play\in\players$, $\gradop^{\play} (\gammameas^{\play}\payoff_{\Pot}^{\play})=\gradop^{\play} \funcphi$. 
Hence using first \cref{le:varphi} and then \cref{pr:normalized-harmonic}, we obtain
\begin{align*}
\inner*{\game_{\Pot}}{\game_{\Har}}_{\muprof,\nsgameg} 
&=\sum_{\play\in\players}\inner*{\funcphi}{\mumeas^{\play}(\actions^{\play}) \gammameas^{\play}\payoff_{\Har}^{\play}}_{0} \\
&=\inner*{\funcphi}{\sum_{\play\in\players}\mumeas^{\play}(\actions^{\play})\gammameas^{\play} \payoff_{\Har}^{\play}}_{0}\\ 
&= \inner*{\funcphi}{0}_{0}\\
&=0. \qedhere
\end{align*}
\end{proof}

The following lemma  states that the flow induced by a game $\game$ and the flow induced by its projection $\opPi(\game)$ on normalized games  are equal.

\begin{lemma}\label{lem:flow_normalized}
Let $\game\in \games$, then $\embop(\opPi(\game))=\embop(\game)$.
\end{lemma}

\begin{proof}
The proof relies on two properties of the pseudo-inverse. Let $\game \in \games$. By \cref{eq:deltaidagdeltai}, we have  $\gradop^{\play\dag} \gradop^{\play}= \opPi^{\play}$. Hence,
\begin{align*}
\embop(\opPi(\game))
&=\sum_{\play\in\players}\gradop^{\play} \parens*{\gammameas^{\play} \opPi^{\play}(\payoff^{\play})}\\
&=\sum_{\play\in\players}\gradop^{\play} \gammameas^{\play} \parens*{\gradop^{\play\dag}\gradop^{\play} \payoff^{\play}}\\
&=\sum_{\play\in\players} \gammameas^{\play} \gradop^{\play} \parens*{\gradop^{\play\dag}\gradop^{\play} \payoff^{\play}}.
\end{align*}
By definition of the pseudo-inverse, $\gradop^{\play} \gradop^{\play\dag} \gradop^{\play}=\gradop^{\play}$, hence we can simplify the right-hand side to obtain
\begin{align*}
\embop(\opPi(\game))&=\sum_{\play\in\players} \gammameas^{\play} \gradop^{\play} \payoff^{\play} =\sum_{\play\in\players} \gradop^{\play} (\gammameas^{\play} \payoff^{\play})=\embop(\game). 
\qedhere
\end{align*}
\end{proof}

\begin{lemma}\label{le:projections}
Given a game $\game$, let $\game_{\Pot},\game_{\Har},\game_{\NSG}$ be the games defined as follows:
\begin{equation}\label{eq:compos}
\game_{\Pot}=\opPi(\funcvec),\quad 
\game_{\Har}=\opPi \parens*{\game -\funcvec},\quad 
\game_{\NSG}=\opLam(\game),
\end{equation}
where 
\begin{equation}\label{eq:bolf}
\funcvec=\parens*{(1/\gammameas^{1})\gradop^{\dag}\embop(\game),\dots,(1/\gammameas^{\nplayers})\gradop^{\dag}\embop(\game)}.
\end{equation} 
Then, 
\begin{equation}\label{eq:sum-og-gs}
\game = \game_{\Pot} + \game_{\Har} + \game_{\NSG},
\end{equation}
where $\game_{\Pot}$ is $\nsgameg$-potential $\muprof$-normalized, $\game_{\Har}$ is $(\muprof,\nsgameg)$-harmonic $\muprof$-normalized, and $\game_{\NSG}$ is nonstrategic. 
Hence, $\game_{\Pot},\game_{\Har}$, and $\game_{\NSG}$ are the components of the $(\muprof,\nsgameg)$-decomposition of $\game$.
\end{lemma}

\begin{proof} 
It is clear that $\game_{\Pot}+\game_{\Har}+\game_{\NSG}=\game$. 
By  \cref{le:chara-NO-NS}, we know that $\game_{\NSG}$ is nonstrategic whereas $\game_{\Pot}$ and $\game_{\Har}$ are $\muprof$-normalized. 
Then, we need to verify that $\game_{\Pot}$ and $\game_{\Har}$ are potential and harmonic, respectively.

Let $\funcphi \colon \actions\to \R$ be such that $\funcphi=\gradop^{\dag}\embop(\game)$. 
We start with the $\nsgameg$-potential component. 
By using \cref{lem:flow_normalized}, we obtain
\begin{align*}
\embop(\game_{\Pot})=\embop(\opPi(\funcvec))=\embop(\funcvec).
\end{align*}
It follows  from \cref{eq:D} that
\begin{align*}
\embop(\game_{\Pot})
&=\sum_{\play\in\players}\gradop^{\play} \parens*{\gammameas^{\play}\frac{1}{\gammameas^{\play}}\funcphi}\\
& =\sum_{\play\in\players} \gradop^{\play} (\funcphi)\\
& =\gradop (\funcphi).
\end{align*}
We get for the $(\muprof,\nsgameg)$-harmonic component the same simplification of $\opPi$:
\begin{align*}
\gradop^{*}\parens*{\embop(\game_{\Har})}
&=\gradop^{*}\parens*{\embop(\opPi\parens*{\game - \funcvec})}\\
&=\gradop^{*}\parens*{\embop(\game) - \embop(\funcvec)}\\
&=\gradop^{*}\parens*{\embop(\game) - \sum\limits_{\play\in\players}\gradop^{\play}\frac{\gammameas^{\play}}{\gammameas^{\play}}\gradop^{\dag}\embop(\game)}\\
&=\gradop^{*}(\Id_{\Cone}-\gradop\gradop^{\dag})\embop(\game)\\
&=0,
\end{align*}
where the third equality is obtained by replacing $\funcvec$ with its definition, as in \cref{eq:bolf}, and the last one is due to the fact that $(\Id_{\Cone}-\gradop\gradop^{\dag})$ is the orthogonal projection onto $\Ker(\gradop^{*})$.
\end{proof}

\begin{proof}[Proof of \cref{th:ortho-sum}]
This is a consequence of \cref{pr:ortho-jeux,le:projections}.
\end{proof}

We now characterize the set of measures which yield the same $(\muprof,\nsgameg)$-decomposition.

\begin{proof}[Proof of \cref{pro:equivalence-class}]
Let $\scala,\scalb>0$. For every $\play\in\players$, let $\mumeas_{\scala}^{\play}= \scala \mumeas^{\play}$ and $\gammameas_{\scalb}^{\play}=\scalb \gammameas^{\play}$. Clearly,  $(\mumeas_{\scala},\nsgameg_{\scalb})$ yields the same decomposition as $(\muprof,\nsgameg)$. 
We now check that there is no other pair of measure and nonstrategic game which induces the same decomposition.

Let $\muprof,\mathring{\muprof}$ be two product measures and let $\nsgameg,\mathring{\nsgameg}$ be two positive nonstrategic games such that the $(\muprof,\nsgameg)$-decomposition and the $(\mathring{\muprof},\mathring{\nsgameg})$-decomposition coincide. Equivalently, the particular classes of games that appear as components of each decomposition are the same: any $\nsgameg$-potential game is $\mathring{\nsgameg}$-potential, any $\muprof$-normalized game is also $\mathring{\muprof}$-normalized, and any $(\muprof,\nsgameg)$-harmonic game is also $(\mathring{\muprof},\mathring{\nsgameg})$-harmonic.

We now consider suitable games  to obtain the relation between $\muprof,\mathring{\muprof},\nsgameg$, and $\mathring{\nsgameg}$. 
We first focus on $\muprof$-normalized games. 
Fix $\play\in\players$ and $\actalt^{\play},\actr^{\play}\in\actions^{\play}$ and consider the game $\game \in \games$ such that, for all $\actprof^{-\play}\in\actions^{-\play}$,
\begin{align*}
\payoff^{\play}(\actalt^{\play},\actprof^{-\play})=\frac{1}{\mumeas^{\play}(\actalt^{\play})}\quad\text{and}\quad \payoff^{\play}(\actr^{\play},\actprof^{-\play})=-\frac{1}{\mumeas^{\play}(\actr^{\play})}.
\end{align*}
Let the payoffs for every other player and the payoffs of player $\play$ in any other profile be equal to $0$. The game $\game$ is $\muprof$-normalized and therefore, it is also $\mathring{\muprof}$-normalized:
\[
\frac{\mathring{\mumeas}^{\play}(\actalt^{\play})}{\mumeas^{\play}(\actalt^{\play})}-\frac{\mathring{\mumeas}^{\play}(\actr^{\play})}{\mumeas^{\play}(\actr^{\play})}=0
\]
and thus,
\[
\frac{\mathring{\mumeas}^{\play}(\actalt^{\play})}{\mumeas^{\play}(\actalt^{\play})}=\frac{\mathring{\mumeas}^{\play}(\actr^{\play})}{\mumeas^{\play}(\actr^{\play})}.
\]
By changing the game, it follows that all the quotients are equal to some positive real number and thus, for every $\play\in\players$, there exists $\scala_i >0$ such that $\mathring{\mumeas}^{\play}=\scala_{\play} \mumeas^{\play}$.

We now consider $\nsgameg$-potential games. We construct a game where we focus only on two players and two strategies for each player. Let $\play,\playalt\in\players$, $\actalt^{\play} \in\actions^{\play}$ and $\actalt^{\playalt}\in\actions^{\playalt}$. Define the potential function $\potent$ on $\actions$ as follows:
\[
\forall \actprof\in\actions,\ \potent(\actprof)=
\begin{cases} 
1 & \text{ if $\act^{\play}=\actalt^{\play}$ and $\act^{\playalt}=\actalt^{\playalt}$},\\
0 & \text{ otherwise},
\end{cases}
\]
Fix $\actprof^{-(\play,\playalt)} \in \actions^{-(\play,\playalt)}$. 
Let $\actr^{\play} \neq \actalt^{\play}$ and $\actr^{\playalt} \neq \actalt^{\playalt}$. 
We focus on the profiles $(\actr^{\play},\actr^{\playalt},\argdot)$, $(\actr^{\play},\actalt^{\playalt},\argdot)$, $(\actalt^{\play},\actalt^{\playalt},\argdot)$ and $(\actalt^{\play},\actr^{\playalt},\argdot)$ where $\argdot$ is a short notation for $\actprof^{-(\play,\playalt)}$. 
The following matrices represent the potential function $\potent$ and a $\nsgameg$-potential game $\game$ associated to $\potent$: 
\begin{center}
\medskip
{\tabulinesep=1.2mm
\begin{tabu} to 45mm {X[ 1 , r ] | X[ 1 , c ] | X[ 1 , c ] | }
\multicolumn{1}{r}{}
 &  \multicolumn{1}{c}{$\actalt^{\playalt}$}
 & \multicolumn{1}{c}{$\actr^{\playalt}$} \\
\cline{2-3}
$\actalt^{\play}$ & $1$ & $0$ \\
\cline{2-3}
$\actr^{\play}$ & $0$ & $0$ \\
\cline{2-3}
\multicolumn{1}{r}{}
&\multicolumn{2}{c}{$\potent$}
\end{tabu}
\qquad
\begin{tabu} to 70mm {X[ 1 , r ] | X[ 4 , c ] X[ 1 , c ] | X[ 1 , c ] X[ 4 , c ] | }
\multicolumn{1}{r}{}
 &  \multicolumn{2}{c}{$\actalt^{\playalt}$}
 & \multicolumn{2}{c}{$\actr^{\playalt}$} \\
\cline{2-5}
$\actalt^{\play}$ & $0$ & $0$ & $0$ & $-1/\gammameas^{\playalt}(\act^{\play},\argdot)$ \\
\cline{2-5}
$\actr^{\play}$ &  $-1/\gammameas^{\play}(\actalt^{\playalt},\argdot)$ & $0$ & $0$ & $0$ \\
\cline{2-5}
\multicolumn{5}{c}{$\game$}
\end{tabu}} .
\medskip
\end{center}
Player $\play$ chooses the row and her payoff is the first coordinate whereas player $\playalt$ chooses the column and her payoff is the second coordinate.

By assumption, $\game$ is an $\mathring{\nsgameg}$-potential game and so we have
\begin{multline*}
\mathring{\gammameas}^{\play}(\actalt^{\playalt},\argdot) (\payoff^{\play}(\actr^{\play},\actalt^{\playalt},\argdot)-\payoff^{\play}(\actalt^{\play},\actalt^{\playalt},\argdot)) 
+\mathring{\gammameas}^{\playalt}(\actr^{\play},\argdot) (\payoff^{\playalt}(\actr^{\play},\actr^{\playalt},\argdot)-\payoff^{\playalt}(\actr^{\play},\actalt^{\playalt},\argdot))\\
+\mathring{\gammameas}^{\play}(\actr^{\playalt},\argdot) (\payoff^{\play}(\actalt^{\play},\actr^{\playalt},\argdot)-\payoff^{\play}(\actr^{\play},\actr^{\playalt},\argdot))+\mathring{\gammameas}^{\playalt}(\actalt^{\play},\argdot) (\payoff^{\playalt}(\actalt^{\play},\actalt^{\playalt},\argdot)-\payoff^{\playalt}(\actalt^{\play},\actr^{\playalt},\argdot))=0.
\end{multline*}
Replacing the corresponding payoff in $\game$, we get:
\[
\mathring{\gammameas}^{\play}(\actalt^{\playalt},\argdot)\parens*{-\frac{1}{\gammameas^{\play}(\actalt^{\playalt},\argdot)}-0}+\mathring{\gammameas}^{\playalt}(\actr^{\play},\argdot) (0-0)+
\mathring{\gammameas}^{\play}(\actr^{\playalt},\argdot) (0-0)+
\mathring{\gammameas}^{\playalt}(\actalt^{\play},\argdot)\parens*{0+ \frac{1}{\gammameas^{\playalt}(\actalt^{\play},\argdot)}}=0
\]
Hence, for every strategy $\actalt^{\play} \in \actions^{\play}$ and every strategy $\actalt^{\playalt} \in \actions^{\playalt}$, we obtain
\[
\frac{\mathring{\gammameas}^{\play}(\actalt^{\playalt},\argdot)}{\gammameas^{\play}(\actalt^{\playalt},\argdot)}
=\frac{\mathring{\gammameas}^{\playalt}(\actalt^{\play},\argdot)}{\gammameas^{\playalt}(\actalt^{\play},\argdot)}
\]
By changing the game, we obtain that the equality is true for every pair of players and for every pair of strategies. In particular, it is also true for two strategies of a player since they have to be equal to the same quotient for another player. It follows that there exists $\scalb >0$, such that for every $\play\in\players$, $\mathring{\gammameas}^{\play}=\scalb \gammameas^{\play}$.

Finally, we consider $(\muprof,\nsgameg)$-harmonic games. Let $\play,\playalt\in\players$, $\actalt^{\play} \in \actions^{\play}$ and $\actalt^{\playalt} \in \actions^{\playalt}$. 
We construct a particular $(\muprof,\nsgameg)$-harmonic game $\game$. The payoff of player $\play\in\players$ is defined  as follows:
\[
\payoff^{\play}(\actprof)=
\begin{dcases} 
0 &\text{ if }\act^{\play}=\actalt^{\play},\\
\frac{1-\mumeas^{\playalt}(\actalt^{\playalt})}{\gammameas^{\play}(\actprof^{-\play})} & \text{ if } \act^{\play} \neq \actalt^{\play} \text{ and }\act^{\playalt}=\actalt^{\playalt},\\
\frac{-\mumeas^{\playalt}(\actalt^{\playalt})}{\gammameas^{\play}(\actprof^{-\play})} & \text{ if } \act^{\play} \neq \actalt^{\play}\text{ and }\act^{\playalt} \neq \actalt^{\playalt}.
\end{dcases}
\]
Likewise, for player $\playalt\in\players$, we have
\[
\payoff^{\playalt}(\actprof)=
\begin{dcases} 
0 & \text{ if }\act^{\playalt}=\actalt^{\playalt},\\
\frac{1-\mumeas^{\play}(\actalt^{\play})}{\gammameas^{\playalt}(\actprof^{-\playalt})} & \text{ if } \act^{\playalt} \neq \actalt^{\playalt} \text{ and } \act^{\play}=\actalt^{\play},\\
\frac{\mumeas^{\play}(\actalt^{\play})}{\gammameas^{\playalt}(\actprof^{-\playalt})} & \text{ if } \act^{\playalt} \neq \actalt^{\playalt}  \text{ and }\act^{\play} \neq \actalt^{\play}.
\end{dcases}
\]
The payoff of all the other players is assumed to be $0$.

This game could be reduced to a game with two strategies for each player, where a player chooses either her strategy labeled by $\actalt$ or any other strategy. This yields the following representation:
\begin{center}
\medskip
{\tabulinesep=1.2mm
\begin{tabu}{ r|cc|cc|}
\multicolumn{1}{r}{}
 &  \multicolumn{2}{c}{$\actalt^{\playalt}$}
 & \multicolumn{2}{c}{$\act^{\playalt} \neq \actalt^{\playalt}$}\\
\cline{2-5}
$\actalt^{\play}$ & $0$ & $0$ & $0$ & $-(1-\mumeas^{\play}(\actalt^{\play}))$ \\
\cline{2-5}
$\act^{\play} \neq \actalt^{\play}$ &  $(1-\mumeas^{\playalt}(\actalt^{\playalt}))$ & $0$ & $-\mumeas^{\playalt}(\actalt^{\playalt})$ & $\mumeas^{\play}(\actalt^{\play})$ \\
\cline{2-5}
\multicolumn{5}{c}{\hspace{5mm} $\gammameas^{\play} \payoff^{\play}$ (left) and $\gammameas^{\playalt} \payoff^{\playalt}$ (right)}
\end{tabu}}
\medskip
\end{center}
One can check that $\game$ is indeed $(\muprof,\nsgameg)$-harmonic. 
Therefore, $\game$ is also $(\mathring{\muprof},\mathring{\nsgameg})$-harmonic. 
We have seen previously that $\mathring{\nsgameg}$ is a multiple of $\nsgameg$, hence $\game$ is also $(\mathring{\muprof},\nsgameg)$-harmonic. 
Let $\actr^{\play} \neq \actalt^{\play}$, $\actr^{\playalt} \neq \actalt^{\playalt}$ and $\actprof^{-(\play,\playalt)}\in \actions^{-(\play,\playalt)}$. 
Since all strategies different from $\actalt^{\play}$ (resp. $\actalt^{\playalt}$) are duplicate, the $(\mathring{\muprof},\nsgameg)$-harmonicity of $\game$ at $(\actr^{\play},\actr^{\playalt},\argdot)$ yields, by \cref{de:classes-of-games}\ref{it:classes-mu-eta-harmonic},
\begin{align*}
\mathring{\mumeas}^{\playalt}(\actalt^{\playalt})\gammameas^{\playalt}(\actr^{\play},\argdot)& \parens*{\payoff^{\playalt}(\actr^{\play},\actalt^{\playalt},\argdot) - \payoff^{\playalt}(\actr^{\play},\actr^{\playalt},\argdot)} +\\
& \mathring{\mumeas}^{\play}(\actalt^{\play})\gammameas^{\play}(\actr^{\playalt},\argdot)\parens*{\payoff^{\play}(\actalt^{\play},\actr^{\playalt},\argdot)-\payoff^{\play}(\actr^{\play},\actr^{\playalt},\argdot)}=0,
\end{align*}
where $\argdot$ is a short notation for $\actprof^{-(\play,\playalt)}$.
Replacing $\gammameas^{\play} \payoff^{\play}$ and $\gammameas^{\playalt}\payoff^{\playalt}$ with their definitions, we obtain
\[
\mathring{\mumeas}^{\playalt}(\actalt^{\playalt})(0-\mumeas^{\play}(\actalt^{\play}))+\mathring{\mumeas}^{\play}(\actalt^{\play})(0-(-\mumeas^{\playalt}(\actalt^{\playalt})))=0
\]
and, as a consequence,
\[
\frac{\mathring{\mumeas}^{\playalt}(\actalt^{\playalt})}{\mumeas^{\playalt}(\actalt^{\playalt})}=\frac{\mathring{\mumeas}^{\play}(\actalt^{\play})}{\mumeas^{\play}(\actalt^{\play})}.
\]
By changing the game, we obtain that the equality is true for every pair of players and for every pair of strategies. It follows that there exists a unique $\scala>0$, such that for every $\play\in\players$, $\mathring{\mumeas}^{\play}=\scala \mumeas^{\play}$. This concludes the proof.
\end{proof}

%----------------------- SUBSECTION ---------------------------------

\subsection*{Proofs of \cref{suse:potential-harmonic}}

The following  lemma  states that the global flow around a set $\subactions\subset\actions$ only depends on the flow on $\subactions^{c}\coloneqq\actions\setminus\subactions$ since the flows inside $\subactions$ compensate one another.    

\begin{lemma}
\label{le:antis}
Let $\flow\in \Cone$ and  $\subactions\subset\actions$. Then, 
\begin{equation*}
\sum_{\actprof\in\subactions}\sum_{\actaltprof\in\actions}\mumeas(\actprof)\mumeas(\actaltprof)\flow(\actprof,\actaltprof)=\sum_{\actprof\in\subactions}\sum_{\actaltprof \in\subactions^{c}}\mumeas(\actprof)\mumeas(\actaltprof)\flow(\actprof,\actaltprof).
\end{equation*}
\end{lemma}

\begin{proof}
We have:
\[\sum_{\actprof\in\subactions}\sum_{\actaltprof\in\actions}\mumeas(\actprof)\mumeas(\actaltprof)\flow(\actprof,\actaltprof)
=\sum_{\actprof\in\subactions}\sum_{\actaltprof \in\subactions}\mumeas(\actprof)\mumeas(\actaltprof)\flow(\actprof,\actaltprof)+\sum_{\actprof \in\subactions}\sum_{\actaltprof\in\subactions^{c}}\mumeas(\actprof)\mumeas(\actaltprof)\flow(\actprof,\actaltprof)\]
and due to the skew-symmetric structure of $\subactions$ for any $\actprof,\actaltprof\in\actions$, $\flow(\actprof,\actaltprof)+\flow(\actaltprof,\actprof)=0$ and thus the first term of the right-hand side is equal to $0$.
\end{proof}

\begin{proof}[Proof of  \cref{th:mu-eta-mixed-equilibrium}] 
Let $\game$ be a $(\muprof,\nsgameg)$-harmonic game. Then, $\gradop^{*}\embop(\game)(\actprof)=0$  for all $\actprof\in\actions$.  
Let  $\play\in\players$ and   $\actr^{\play} \in \actions^{\play}$. 
Call $\subactions$ the subset of strategy profiles $\braces*{(\actr^{\play},\actprof^{-\play}) \colon \actprof^{-\play} \in \actions^{-\play}}$. 
Then, multiplying by $\mumeas(\actr^{\play},\actprof^{-\play})$ and summing over $\actprof^{-\play} \in \actions^{-\play}$ we get
\begin{align*}
0&=\sum_{\actprof^{-\play} \in \actions^{-\play}}\mumeas(\actr^{\play},\actprof^{-\play})\gradop^{*}\embop(\game)(\actr^{\play},\actprof^{-\play})\\
&=-\sum_{\actrprof\in \subactions}\mumeas(\actrprof)\sum_{\playalt\in\players}\gradop^{\playalt*}\gradop^{\playalt}(\gammameas^{\playalt} \payoff^{\playalt})(\actrprof)\\
&=\sum_{\actrprof\in \subactions}\sum_{\actprof\in\actions}\mumeas(\actrprof)\mumeas(\actprof)\sum_{\playalt\in\players}\symmfunc^{\playalt}(\actrprof,\actprof)\gradop^{\playalt}(\gammameas^{\playalt} \payoff^{\playalt})(\actrprof,\actprof).
\end{align*}
Thus, in view of  \cref{le:antis}, we can eliminate some terms of the summation and then notice that $\actprof\in\subactions$ and $\actaltprof \in \subactions^c$ are not $\playalt$-comparable if $\playalt\neq\play$, hence
\begin{align*}
0&=\sum_{\actrprof \in \subactions}\sum_{\actprof \in \subactions^c}\mumeas(\actrprof)\mumeas(\actprof)\sum_{j \in N}\symmfunc^{\playalt}(\actrprof,\actprof)\gradop^{\playalt}(\gammameas^{\playalt} \payoff^{\playalt})(\actrprof,\actprof)\\
&=\sum_{\actrprof\in \subactions}\sum_{\actprof \in \subactions^c}\mumeas(\actrprof)\mumeas(\actprof) \parens*{\symmfunc^{\play}(\actrprof,\actprof)\gradop^{\play}(\gammameas^{\play} \payoff^{\play})(\actprof,\actaltprof) +\sum_{\playalt\neq\play} \symmfunc^{\playalt}(\actrprof,\actprof)\gradop^{\playalt}(\gammameas^{\playalt} \payoff^{\playalt})(\actprof,\actaltprof)}\\
&=\sum_{\actrprof\in\subactions}\sum_{\actprof \in \subactions^c}\mumeas(\actrprof)\mumeas(\actprof) \symmfunc^{\play}(\actrprof,\actprof)\gradop^{\play}(\gammameas^{\play} \payoff^{\play})(\actrprof,\actprof) + 0.
\end{align*}
We can now replace $\gradop^{\play}$ and $\symmfunc^{\play}$ with their definitions to obtain
\begin{align*}
0 &=\sum_{\actrprof\in\subactions}\sum_{\actprof \in \subactions^c}\mumeas(\actrprof)\mumeas(\actprof) \symmfunc^{\play}(\actrprof,\actprof)^{2}\parens*{(\gammameas^{\play} \payoff^{\play})(\actprof)-(\gammameas^{\play} \payoff^{\play})(\actrprof)}\\
&=\sum_{\actprof^{-\play} \in \actions^{-\play}}\sum_{\act^{\play} \in \actions^{\play}} \mumeas(\actr^{\play},\actprof^{-\play})\mumeas(\act^{\play},\actprof^{-\play})\frac{1}{\mumeas^{-\play}(\actprof^{-\play})}\parens*{(\gammameas^{\play}\payoff^{\play})(\actr^{\play},\actprof^{-\play})-(\gammameas^{\play}\payoff^{\play})(\act^{\play},\actprof^{-\play})}\\
&= \mumeas^{\play}(\actr^{\play})\sum_{\actprof^{-\play} \in \actions^{-\play}}\mumeas^{-\play}(\actprof^{-\play})\parens*{\sum_{\act^{\play} \in \actions^{\play}}\mumeas^{\play}(\act^{\play})\parens*{(\gammameas^{\play}\payoff^{\play})(\actr^{\play},\actprof^{-\play})-(\gammameas^{\play}\payoff^{\play})(\act^{\play},\actprof^{-\play})}}.
\end{align*}
Dividing by $\mumeas^{\play}(\actr^{\play})$, which is strictly positive, we obtain
\begin{align*}
0=\sum_{\actprof^{-\play} \in \actions^{-\play}}\mumeas^{-\play}(\actprof^{-\play})\parens*{\sum_{\act^{\play} \in \actions^{\play}}\mumeas^{\play}(\act^{\play})\parens*{(\gammameas^{\play}\payoff^{\play})(\actr^{\play},\actprof^{-\play})-(\gammameas^{\play}\payoff^{\play})(\act^{\play},\actprof^{-\play})}}.
\end{align*}
It follows that, for all $\actr^{\play} \in \actions^{\play}$,
\begin{multline*}
\sum_{\actprof^{-\play} \in \actions^{-\play}}\mumeas^{-\play}(\actprof^{-\play})\sum_{\act^{\play} \in \actions^{\play}}\mumeas^{\play}(\act^{\play})(\gammameas^{\play}\payoff^{\play})(\actr^{\play},\actprof^{-\play})\\
= \sum_{\actprof^{\play} \in \actions^{-\play}}\mumeas^{-\play}(\actprof^{-\play})\sum_{\act^{\play} \in \actions^{\play}}\mumeas^{\play}(\act^{\play})(\gammameas^{\play}\payoff^{\play})(\act^{\play},\actprof^{-\play}),
\end{multline*}
and thus,  equivalently,
\begin{multline*}
\parens*{\sum_{\act^{\play} \in \actions^{\play}}\mumeas^{\play}(\act^{\play})}\sum_{\actprof^{-\play} \in \actions^{-\play}}\mumeas^{-\play}(\actprof^{-\play})(\gammameas^{\play}\payoff^{\play})(\actr^{\play},\actprof^{-\play})\\
= \sum_{\actprof^{-\play} \in \actions^{-\play}}\mumeas^{-\play}(\actprof^{-\play})\sum_{\act^{\play} \in \actions^{\play}}\mumeas^{\play}(\act^{\play})(\gammameas^{\play}\payoff^{\play})(\act^{\play},\actprof^{-\play}).
\end{multline*}

Then, dividing by $\mumeas^{\play}(\actions^{\play})$, we have
\begin{align*}
\sum_{\actprof^{-\play} \in \actions^{-\play}}\norml{{\mumeas}}^{\,\,-\play}(\actprof^{-\play}) \gammameas^{\play}(\actprof^{-\play})\payoff^{\play}(\actr^{\play},\actprof^{-\play})=\sum_{\actprof^{-\play} \in \actions^{-\play}}\norml{{\mumeas}}^{\,\,-\play}(\actprof^{-\play})\sum_{\act^{\play} \in \actions^{\play}}\norml{\mumeas}^{\play}(\act^{\play})\gammameas^{\play}(\actprof^{-\play})\payoff^{\play}(\act^{\play},\actprof^{-\play}).
\end{align*}

Notice that the right-hand side is independent of $\actr^{\play}$. This concludes the proof since $\actr^{\play}$ is arbitrary.
\end{proof}

%----------------------- SUBSECTION ---------------------------------

\noindent In order to prove \cref{cor:mu-eta-mixed-equilibrium-product}, we first prove the following lemma.
\begin{lemma}
\label{pr:equilibrim-preservation}
Let $\nsgameb$ be a product co-measure vector generated by $\bprof$.
If $\equils_{\actions}(\game)$ is the Nash-equilibrium set of the game $\game$, then the Nash-equilibrium set of $(\nsgameb\cdot\game)$ is given by
\begin{equation}\label{eq:equilibria-prod-scal}
\equils_{\actions}(\nsgameb\cdot\game)
=\braces*{(\norml{\mixedalt}^{\play})_{\play\in\players} : \mixedalt^{\play}(\act^{\play})=\frac{\mixed^{\play}(\act^{\play})}{\bmeas^{\play}(\act^{\play})} \text{ and } (\mixed^{\play})_{\play\in\players} \in \equils_{\actions}(\game)}.
\end{equation} 
\end{lemma}

In particular, when $\nsgameb$ is a product-scaling, if $\mixedprof$ is a pure Nash equilibrium of $\game$ then $\mixedprof$ is also a pure Nash equilibrium of $(\nsgameb\cdot\game)$. 
Knowing the set of equilibria of some game $\game$, it is possible to compute the set of equilibria of $(\nsgameb\cdot\game)$ without knowing $\game$. 
In particular, when, for every $\play\in\players$, $\bmeas^{\play}(\act^{\play})$ does not depend on $\act^{\play}$, we obtain constant scaling, and the Nash equilibrium sets of both games coincide. 
Notice that, a scaling that is not a product-scaling nevertheless preserves the set of pure Nash-equilibria.

\begin{proof}[Proof of \cref{pr:equilibrim-preservation}]
Let $\game \in \games$ and let $\mixedprof=(\mixed^{\play})_{\play\in\players}$ be a Nash equilibrium of $\game$. Let
$\supp(\mixed^{\play})=\braces{\act^{\play} \in \actions^{\play} \colon \mixed^{\play}(\act^{\play})>0}$. 
By definition of $\mixedalt^{\play}$, we have that $\supp(\mixed^{\play})=\supp(\mixedalt^{\play})$. 
Then, for any $\play\in\players$ and any $\act^{\play},\actalt^{\play} \in \supp(\mixed^{\play})$, we have $\mixed^{-\play}(\actprof^{-\play})\payoff^{\play}(\act^{\play},\actprof^{-\play})=\mixed^{-\play}(\actprof^{-\play})\payoff^{\play}(\actalt^{\play},\actprof^{-\play})$. By \cref{de:scaling,eq:prod-nu-i}, we get
\begin{align*}
\mixedalt^{-\play}(\actprof^{-\play})(\nsgameb\cdot\game)^{\play}(\act^{\play},\actprof^{-\play})
&=\dfrac{\mixed^{-\play}(\actprof^{-\play})}{\prod_{\playalt\neq\play} \bmeas^{\playalt}(\act^{\playalt})}(\nsgameb\cdot\game)^{\play}(\act^{\play},\actprof^{-\play})\\
&= \mixed^{-\play}(\actprof^{-\play})\payoff^{\play}(\act^{\play},\actprof^{-\play})\\
&= \mixed^{-\play}(\actprof^{-\play})\payoff^{\play}(\actalt^{\play},\actprof^{-\play})\\
&=\dfrac{\mixed^{-\play}(\actprof^{-\play})}{\prod_{\playalt\neq\play} \bmeas^{\playalt}(\act^{\playalt})}(\nsgameb\cdot\game)^{\play}(\actalt^{\play},\actprof^{-\play})\\
&=\mixedalt^{-\play}(\actprof^{-\play})(\nsgameb\cdot\game)^{\play}(\actalt^{\play},\actprof^{-\play}).
\end{align*}
Let $\act^{\play} \in \supp(\mixed^{\play})$ and $\actalt^{\play} \notin \supp(\mixed^{\play})$. Then, 
\begin{align*}
\mixedalt^{-\play}(\actprof^{-\play})(\nsgameb\cdot\game)^{\play}(\act^{\play},\actprof^{-\play})
&=\dfrac{\mixed^{-\play}(\actprof^{-\play})}{\prod_{\playalt\neq\play} \bmeas^{\playalt}(\act^{\playalt})}(\nsgameb\cdot\game)^{\play}(\act^{\play},\actprof^{-\play})\\
&= \mixed^{-\play}(\actprof^{-\play})\game^{\play}(\act^{\play},\actprof^{-\play})\\
&\geq \mixed^{-\play}(\actprof^{-\play})\game^{\play}(\actalt^{\play},\actprof^{-\play})\\
&= \dfrac{\mixed^{-\play}(\actprof^{-\play})}{\prod_{\playalt\neq\play} \bmeas^{\playalt}(\act^{\playalt})}(\nsgameb\cdot\game)^{\play}(\actalt^{\play},\actprof^{-\play})\\
&=\mixedalt^{-\play}(\actprof^{-\play})(\nsgameb\cdot\game)^{\play}(\actalt^{\play},\actprof^{-\play}).
\end{align*}
The same equality holds for the normalized strategy $\norml{\mixedalt}^{\play}$ and it thus follows that $\norml{\actaltprof}=\parens{\norml{\actalt}^{\play}}_{\play\in\players}$ is a Nash equilibrium.
\end{proof}

\begin{proof}[Proof of \cref{cor:mu-eta-mixed-equilibrium-product}]
By \cref{th:mu-eta-mixed-equilibrium}, $\norml{\mumeas}$ is a Nash equilibrium in $\nsgameg \cdot \game$. We have that $$\frac{1}{\nsgameg} \cdot \parens*{\nsgameg \cdot \game}=\game,$$
where $\frac{1}{\nsgameg}$ is the  co-measure vector defined as follows:
\begin{equation}\label{eq:1-over-gamma}
\forall \play \in \players, \forall \actprof^{-\play} \in \actions^{-\play}, \,\, \parens*{\frac{1}{\nsgameg}}^{\play}\parens*{\actprof^{-\play}}=\frac{1}{\gammameas^{\play}(\actprof^{-\play})}.
\end{equation}
By \cref{pr:equilibrim-preservation}, we get that $\overline{\norml{\mumeas}\nsgameg}=\norml{\mumeas \nsgameg}$ is a Nash equilibrium in $\game$.
\end{proof}

%----------------------- SUBSECTION ---------------------------------

\subsection*{Proofs of \cref{suse:closest-potential}}

\begin{proof}[Proof of \cref{pr:closest-potential}]
According to the definition of the induced norm, for any $\play \in \players$ and for all $\actprof \in \actions$ we have
\begin{equation}
\label{eq:majnorm}
\abs*{\payoff^{\play}(\actprof)-\payoff_{\Pot}^{\play}(\actprof)} 
\leq \frac{1}{\gammameas^{\play}(\actprof^{-\play})\sqrt{\mumeas^{\play}(\actions^{\play})}}\norm*{\game - \game_{\Pot}}_{\muprof,\nsgameg} 
\leq \max\limits_{\playalt\in\players}\max\limits_{\actprof^{-\playalt} \in \actions^{-\playalt}}\frac{\dggp}{\gammameas^{\playalt}(\actprof^{-\playalt})\sqrt{\mumeas^{\playalt}(\actions^{\playalt})}}.
\end{equation}
Let $\actprof \in \actions$ be an equilibrium in $\game_{\Pot}$ and  consider $\actaltprof \in \actions$ such that $\act^{\play} \neq \actalt^{\play}$ for some $\play \in \players$ and $\act^{\playalt}=\actalt^{\playalt}$ for any $\playalt\neq\play$. 
Then we have
\begin{align*}
\payoff^{\play}(\actaltprof)-\payoff^{\play}(\actprof)
&\leq \payoff^{\play}(\actaltprof)-\payoff^{\play}(\actprof)-\parens*{\payoff_{\Pot}^{\play}(\actaltprof)-\payoff_{\Pot}^{\play}(\actprof)}\\
&\leq 2 \max\limits_{\playalt\in\players}\max\limits_{\actprof^{-\playalt} \in \actions^{-\playalt}}\frac{\dggp}{\gammameas^{\playalt}(\actprof^{-\playalt})\sqrt{\mumeas^{\playalt}(\actions^{\playalt})}},
\end{align*}
where the first inequality follows from the fact that $\actprof \in \actions$ is an equilibrium in $\game_{\Pot}$ and the second one is due to \cref{eq:majnorm}.
\end{proof}

%----------------------- SUBSECTION ---------------------------------

\subsection*{Proofs of \cref{suse:preserving}}

\begin{proof}[Proof of \cref{pr:decomposition-translation}]
Given \cref{eq:pseudo-translation}, it is clear that $\game_{(\muprof,\nsgameg)\NSd}+\func(\game)$ is nonstrategic whereas $\game_{(\muprof,\nsgameg)\Potd}$ is $\muprof$-normalized $\nsgameg$-potential, and $\game_{(\muprof,\nsgameg)\Hard}$ is $\muprof$-normalized $(\muprof,\nsgameg)$-potential. 
By uniqueness of the decomposition in \cref{th:ortho-sum}, we obtain the result.
\end{proof}

\begin{proof}[Proof of \cref{th:scaling}]
Let $\muprof$ be a positive measure on $\actions$. Let $\nsgameg$ and $\nsgameb$ be two  co-measure vectors.
By  \cref{th:ortho-sum}, there exist a $\nsgameg$-potential $\muprof$-normalized game $\game_{(\muprof,\nsgameg)\Potd}$, a $(\muprof,\nsgameg)$-harmonic $\muprof$-normalized game $\game_{(\muprof,\nsgameg)\Hard}$, and a nonstrategic game $\game_{(\muprof,\nsgameg)\NSd}$ such that
\begin{equation}
\game=\game_{(\muprof,\nsgameg)\NSd}+\game_{(\muprof,\nsgameg)\Potd}+\game_{(\muprof,\nsgameg)\Hard}.
\end{equation}
For any $\play\in\players$, $\act^{\play},\actalt^{\play} \in \actions^{\play}$ and $\actprof^{-\play} \in \actions^{-\play}$, we have:
\begin{align*}
\potent(\actalt^{\play},\actprof^{-\play})-\potent(\act^{\play},\actprof^{-\play})
=\gammameas^{\play}(\actprof^{-\play}) \payoff_{(\muprof,\nsgameg)\Potd}^{\play}(\actalt^{\play},\actprof^{-\play})-\gammameas^{\play}(\actprof^{-\play}) \payoff_{(\muprof,\nsgameg)\Potd}^{\play}(\act^{\play},\actprof^{-\play}),
\end{align*}
therefore,
\begin{align*}
\potent(\actalt^{\play},\actprof^{-\play})-\potent(\act^{\play},\actprof^{-\play}) &=\frac{\gammameas^{\play}(\actprof^{-\play})}{\betameas^{\play}(\actprof^{-\play})} \parens*{\betameas^{\play}(\actprof^{-\play})\payoff_{(\muprof,\nsgameg)\Potd}^{\play}(\actalt^{\play},\actprof^{-\play})-\betameas^{\play}(\actprof^{-\play})\payoff_{(\muprof,\nsgameg)\Potd}^{\play}(\act^{\play},\actprof^{-\play})}\\
&=\frac{\gammameas^{\play}(\actprof^{-\play})}{\betameas^{\play}(\actprof^{-\play})} \parens*{(\nsgameb\cdot\payoff_{(\muprof,\nsgameg)\Potd})^{\play}(\actalt^{\play},\actprof^{-\play})-(\nsgameb\cdot\payoff_{(\muprof,\nsgameg)\Potd})^{\play}(\act^{\play},\actprof^{-\play})}.
\end{align*}
Hence $(\nsgameb\cdot\game_{(\muprof,\nsgameg)\Potd})$ is $\nsgameg/\nsgameb$-potential. Furthermore, one can check easily that $(\nsgameb\cdot\game_{(\muprof,\nsgameg)\Potd})$ is  $\muprof$-normalized.

Moreover, we  have
\begin{align*}
0& =\sum_{\play\in\players} \sum_{\actalt^{\play}\in \actions^{\play}} \mumeas^{\play}(\actalt^{\play})\gammameas^{\play}(\actprof^{-\play})
\parens*{\payoff_{(\muprof,\nsgameg)\Hard}^{\play}(\act^{\play},\actprof^{-\play})-\payoff_{(\muprof,\nsgameg)\Hard}^{\play}(\actalt^{\play},\actprof^{-\play})}\\
& =\sum_{\play\in\players} \sum_{\actalt^{\play}\in \actions^{\play}} \frac{\mumeas^{\play}(\actalt^{\play})\gammameas^{\play}(\actprof^{-\play})}{\betameas^{\play}(\actprof^{-\play})}\parens*{\betameas^{\play}(\actprof^{-\play}) \payoff_{(\muprof,\nsgameg)\Hard}^{\play}(\act^{\play},\actprof^{-\play})-\betameas^{\play}(\actprof^{-\play}) \payoff_{(\muprof,\nsgameg)\Hard}^{\play}(\actalt^{\play},\actprof^{-\play})}\\
& =\sum_{\play\in\players} \sum_{\actalt^{\play}\in \actions^{\play}} \mumeas^{\play}(\actalt^{\play})\frac{\gammameas^{\play}(\actprof^{-\play})}{\betameas^{\play}(\actprof^{-\play})} 
\parens*{(\nsgameb\cdot\payoff_{(\muprof,\nsgameg)\Hard})^{\play}(\actalt^{\play},\actprof^{-\play})-(\nsgameb\cdot\payoff_{(\muprof,\nsgameg)\Hard})^{\play}(\act^{\play},\actprof^{-\play})}.
\end{align*}
One can also check easily that $(\nsgameb\cdot\game_{(\muprof,\nsgameg)\Hard})$ is  $\muprof$-normalized. 
Since 
$(\nsgameb\cdot\game)=(\nsgameb\cdot\game_{(\muprof,\nsgameg)\Potd})+(\nsgameb\cdot\game_{(\muprof,\nsgameg)\Hard})+(\nsgameb\cdot\game_{(\muprof,\nsgameg)\NSd})$, 
the result follows from uniqueness of the $(\muprof,\nsgameg/\nsgameb)$-decomposition.
\end{proof}

%----------------------- SUBSECTION ---------------------------------

\subsection*{Proofs of \cref{suse:duplicate}}

To prove \cref{th:reduce-replicate}, it is sufficient to establish \cref{le:replica}. Uniqueness of the decomposition then implies the theorem.

\begin{proof}[Proof of \cref{le:replica}]
\ref{it:le:replica-1}
Assume that $\game$ is $\nsgameg$-potential. By \cref{pr:operator}, there exists $\potent \colon \actions \to  \R$ such that $\embop(\game)=\gradop(\potent)$.

Define $\widecheck{\potent}$ the extension of the potential $\potent$ to $\widecheck{\actions}$ by
\[
\forall \actprof\in \actions^{-\play}, \widecheck{\potent}(\act_{0}^{\play},\actprof^{-\play})=\potent(\act_{1}^{\play},\actprof^{-\play}).
\]
Then $\widecheck{\potent}$ is a potential function for $\widecheck{\game}$ and hence $\widecheck{\game}$ is $\widecheck{\nsgameg}$-potential. 

\ref{it:le:replica-2}
Assume now that $\game$ is nonstrategic. 
We know that $\embop(\game)=0$.
Clearly, the flow induced by  $\widecheck{\game}$ is identical to the flow induced by $\game$ over the strategy profile set $\actions$. 
Moreover, the flow in $\widecheck{\game}$ between the strategy profile  $(\act_{0}^{\play},\actprof^{-\play})$ and another profile is equal to the flow between  $(\act_{1}^{\play},\actprof^{-\play})$ and another profile. 
Therefore, $\embop(\widecheck{\game})=0$ and $\widecheck{\game}$ is nonstrategic. 

\ref{it:le:replica-3}
For the rest, assume that $\game$ is $\muprof$-normalized. 
By definition, we know that for all $\playalt\in \players$ and for all $\actprof^{-\playalt}\in \actions^{-\playalt}$,
\[
\sum_{\act^{\playalt} \in \actions^{\playalt}} \mumeas^{\playalt}(\act^{\playalt})\payoff^{\playalt}(\act^{\playalt},\actprof^{-\playalt})=0.
\]
We need to distinguish two cases. If $\playalt \neq \play$, then, for all $\actprof^{-\playalt} \in \reduc{\actions}^{-\playalt}\subset \actions^{-\playalt}$, we have
\[
\sum_{\act^{\playalt} \in \actions^{\playalt}} \widecheck{\mumeas}^{\playalt}(\act^{\playalt})\payoff^{\playalt}(\act^{\playalt},\actprof^{-\playalt})
=\sum_{\act^{\playalt} \in \actions^{\playalt}} \mumeas^{\playalt}(\act^{\playalt})\payoff^{\playalt}(\act^{\playalt},\actprof^{-\playalt})=0.
\]
If $\playalt=\play$, then, for all $\actprof^{-\play}\in \actions^{-\play}$, we have
\begin{align*}
\sum_{\act^{\play} \in \widecheck{\actions^{\play}}} \widecheck{\mumeas}^{\play}(\act^{\play})\payoff^{\playalt}(\act^{\play},\actprof^{-\play})
&= \widecheck{\mumeas}^{\play}(\act_{0}^{\play})\payoff^{\play}(\act_{0}^{\play},\actprof^{-\play})+ \widecheck{\mumeas}^{\play}(\act_{1}^{\play})\payoff^{\play}(\act_{1}^{\play},\actprof^{-\play})\\
&\quad+ \sum_{\act^{\play} \in \actions^{\play}\setminus\braces{\act_{1}^{1}}} \widecheck{\mumeas}^{\play}(\act^{\play})\payoff^{\play}(\act^{\play},\actprof^{-\play})\\
&=\mumeas^{\play}(\act_{1}^{\play})\payoff^{\play}(\act_{1}^{\play},\actprof^{-\play})+\sum_{\act^{\play} \in \actions^{\play}\setminus\braces{\act_{1}^{\play}}} \mumeas^{\play}(\act^{\play})\payoff^{\play}(\act^{\play},\actprof^{-\play})\\
&= \sum_{\act^{\play} \in \actions^{\play}} \mumeas^{\play}(\act^{\play})\payoff^{\play}(\act^{\play},\actprof^{-\play})\\
&=0.
\end{align*}

\noindent Therefore, $\widecheck{\game}$ is $\widecheck{\muprof}$-normalized.

\ref{it:le:replica-4}
Finally, assume that $\game$ is $(\muprof,\nsgameg)$-harmonic. Let $\actprof=(\act^{\playalt})_{\playalt\in\players}$ and for all $\playalt\in\players$ and $\actalt^{\playalt} \in \actions^{\playalt}$, put $\payoff^{\playalt}_{\actalt^{\playalt}}(\actprof)=\payoff^{\playalt}(\act^{\playalt},\actprof^{-\playalt})-\payoff^{\playalt}(\actalt^{\playalt},\actprof^{-\playalt})$. 
Then, since $\game$ is assumed to be $(\muprof,\nsgameg)$-harmonic, for any $\actprof\in\actions$, we have
\begin{align*}
\gradop^{*} \embop(\game) (\actprof)
&= \sum_{\playalt\in\players}\sum_{\actalt^{\playalt} \in \actions^{\playalt}}\mumeas^{\playalt}(\actalt^{\playalt})\gammameas^{\playalt}(\actprof^{-\playalt})\payoff^{\playalt}_{\actalt^{\playalt}}(\actprof)\\
&=\sum_{\actalt^{\play} \in \actions^{\play}}\mumeas^{\play}(\actalt^{\play})\gammameas^{\play}(\actprof^{-\play})\payoff^{\play}_{\actalt^{\play}}(\actprof) + \sum_{\playalt\neq\play}\sum_{\actalt^{\playalt} \in \actions^{\playalt}}\mumeas^{\playalt}(\actalt^{\playalt})\gammameas^{\playalt}(\actprof^{-\playalt})\payoff^{\playalt}_{\actalt^{\playalt}}(\actprof)\\
&=0.
\end{align*} 
In the extended game $\widecheck{\game}=(\widecheck{\payoff}^{\playalt})_{\playalt\in\players}$ at any $\actprof\in\widecheck{\actions}$ we get:
\begin{align*}
\gradop^{*} \embop(\widecheck{\game}) (\actprof)&=\sum_{\actalt^{\play} \in \widecheck{\actions^{\play}}}\widecheck{\mumeas}^{\play}(\actalt^{\play})\widecheck{\gammameas}^{\play}(\actprof^{-\play})\widecheck{\payoff}^{\play}_{\actalt^{\play}}(\actprof) + \sum_{\playalt\neq\play}\sum_{\actalt^{\playalt} \in \actions^{\playalt}}\widecheck{\mumeas}^{\playalt}(\actalt^{\playalt})\widecheck{\gammameas}^{\playalt}(\actprof^{-\playalt})\widecheck{\payoff}^{\playalt}_{\actalt^{\playalt}}(\actprof)\\
&=\widecheck{\mumeas}^{\play}(\act_{0}^{\play})\widecheck{\gammameas}^{\play}(\actprof^{-\play})\widecheck{\payoff}^{\play}_{\act_{0}^{\play}}(\actprof)+\widecheck{\mumeas}^{\play}(\act_{1}^{\play})\widecheck{\gammameas}^{\play}(\actprof^{-\play})\widecheck{\payoff}^{\play}_{\act_{1}^{\play}}(\actprof)+
\sum_{\actalt^{\play} \in \widecheck{\actions^{\play}}\setminus{\{\act_{0}^{\play},\act_{1}^{\play}\}}}\widecheck{\mumeas}^{\play}(\actalt^{\play})\widecheck{\gammameas}^{\play}(\actprof^{-\play})\widecheck{\payoff}^{\play}_{\actalt^{\play}}(\actprof) \\
& + \sum_{\playalt\neq\play}\sum_{\actalt^{\playalt} \in \actions^{\playalt}}\widecheck{\mumeas}^{\playalt}(\actalt^{\playalt})\widecheck{\gammameas}^{\playalt}(\actprof^{-\playalt})\widecheck{\payoff}^{\playalt}_{\actalt^{\playalt}}(\actprof).
\end{align*} 
Notice that $\widecheck{\gammameas}^{\play}\equiv\gammameas^{\play}$ and, for any $\playalt\neq\play$, we  have 
\[
\widecheck{\gammameas}^{\playalt}(\act_{0}^{\play},\actprof^{-\parens{\play,\playalt}})=\widecheck{\gammameas}^{\playalt}(\act_{1}^{\play},\actprof^{-\parens{\play,\playalt}})=\gammameas^{\playalt}(\act_{1}^{\play},\actprof^{-\parens{\play,\playalt}}).
\]
By \cref{de:extended} and since $\widecheck{\payoff}^{\play}_{\act_{0}^{\play}}(\actprof)=\widecheck{\payoff}^{\play}_{\act_{1}^{\play}}(\actprof)$, it follows that
\begin{multline*}
\gradop^{*} \embop(\widecheck{\game}) (\actprof)
=\parens*{\widecheck{\mumeas}^{\play}(\act_{0}^{\play})+\widecheck{\mumeas}^{\play}(\act_{1}^{\play})}\gammameas^{\play}(\actprof^{-\play})\payoff^{\play}_{\act_{1}^{\play}}(\actprof)
+\sum_{\actalt^{\play} \neq \act_{0}^{\play},\act_{1}^{\play}}\gammameas^{\play}(\actprof^{-\play})\mumeas^{\play}(\actalt^{\play})\payoff^{\play}_{\actalt^{\play}}(\actprof)\\+\sum_{\playalt\neq\play}\sum_{\actalt^{\playalt} \in \actions^{\playalt}}\gammameas^{\playalt}(\actprof^{-\playalt})\mumeas^{\playalt}(\actalt^{\playalt})\payoff^{\playalt}_{\actalt^{\playalt}}(\actprof)=0.
\end{multline*} 
Since $\mumeas^{\play}(\act_{1}^{\play})=\widecheck{\mumeas}^{\play}(\act_{0}^{\play})+\widecheck{\mumeas}^{\play}(\act_{1}^{\play})$, it  follows that $\gradop^{*} \embop(\widecheck{\game}) (\actprof)=\gradop^{*} \embop(\game) (\actprof)=0$ and we  obtain that $\widecheck{\game}$ is $(\widecheck{\muprof},\widecheck{\nsgameg})$-harmonic.
\end{proof}

\begin{proof}[Proof of \cref{th:reduce-duplicate}]
Let $\game\in \games_{\actions}^{\dup}$ and let $\nsgameg$ be coherent with $\games_{\actions}^{\dup}$. 
Then, we see that $\game=\widecheck{\reduc{\game}}$, $\nsgameg =\widecheck{\reduc{\nsgameg}}$ and $\muprof$ is an extension of $\reduc{\muprof}$.
By applying \cref{th:reduce-replicate} to the $(\reduc{\muprof},\reduc{\nsgameg})$-decomposition of the game $\reduc{\game}$ and, by uniqueness of the decomposition, we obtain 
\[
\widecheck{({\reduc{\game})_{(\reduc{\muprof},\reduc{\nsgameg})\Potd}}}= \parens*{\widecheck{\reduc{\game}}}_{\parens*{\widecheck{\reduc{\muprof}},\widecheck{\reduc{\nsgameg}}}\Potd}.
\]
After simplification of the notation and applying one more time $\reduc{\mapT}$, we get 
\[
(\reduc{\game})_{(\reduc{\muprof},\reduc{\nsgameg})\Potd}= \reduc{{\game_{(\muprof,\nsgameg)\Potd}}}.
\]
A similar computation yields the result concerning the $(\muprof,\nsgameg)$-harmonic component and the nonstrategic component.
\end{proof}

%----------------------- REFERENCES ---------------------------------

\bibliographystyle{apalike}
\bibliography{../bibtex/bibdecomposition}

\newpage

%----------------------- SECTION ---------------------------------

\section{List of symbols}

\begin{longtable}{p{.10\textwidth} p{.85\textwidth}}

$\cmeas$ & generator of the positive product-strategic game $\nsgameg$\\

$\Czero$ & vector space of functions $\func: \actions \to \R$\\

$\Cone$ & set of flows, defined in \cref{eq:C-one}\\

$\embop$ & joint embedding operator\\

$\equils_{\actions}$ & Nash-equilibrium correspondence\\

$\game$ & game\\

$\reduc{\game}$ & reduced game where duplication of strategies has been eliminated\\

$\widecheck{\game}$ & extended game, defined in \cref{de:extended}\\ 

$\game_{(\muprof,\nsgameg)\NSd}$ & nonstrategic game\\

$\game_{(\muprof,\nsgameg)\Potd}$ & $\nsgameg$-potential $\muprof$-normalized game\\

$\game_{(\muprof,\nsgameg)\Hard}$ & $(\muprof,\nsgameg)$-harmonic $\muprof$-normalized game\\

$\payoff^{\play}$ & payoff function of player $\play$\\

$\games_{\actions}$ & space of games with strategy-profile set $\actions$\\

$\games_{\actions}^{\dup}$ & space of games with one duplicate strategy\\

$\games_{\actions}^{\aredred}$ & set of games with an $\ared$-redundant strategy, defined in \cref{de:redund}\\ 

$\Id$ & identity function\\

$\nonstratf$ & nonstrategic game\\

$\measures(\setA)$ & set of measures over $\setA$\\

$\measures_{+}(\setA)$ & set of positive measures over $\setA$\\

$\players$ & a finite set of players\\

$\NSG$ & class of nonstrategic games\\

$\act^{\play}$ & strategy of player $\play$\\

$\actprof$ & strategy profile\\

$\actions$ & set of strategy profiles\\

$\widecheck{\actions}$ & extended set of strategy profiles, defined in \cref{de:extended}\\ 

$\reduc{\actions}$ & reduced set of strategy profiles, defined in \cref{de:reduced}\\ 

$\actions^{\play}$ & strategy set of player $\play$\\

$\widecheck{\actions}^{\play}$ & extended strategy set of player $\play$, defined in \cref{de:extended}\\ 

$\reduc{\actions}^{\play}$ & reduced strategy set of player $\play$, defined in \cref{de:reduced}\\

$\actions^{-\play}$ & set of strategy subprofiles that exclude player $\play$\\

$\mapT$ & map from $\games_{\actions}$ to $\games_{\actionsalt}$\\

$\widecheck{\mapT}$ & extending map, defined in \cref{de:extended}\\ 

$\reduc{\mapT}$ & reducing map, defined in \cref{de:reduced}\\ 

$\maps$ & set of game transformations\\

$\symmfunc^{\play}$ & symmetric function, defined in \cref{eq:symm-func}\\

$\threedecmaps$ & family of \aclp{3DM}\\

$\varscal^{\play}$ & function from $\actions^{-\play}$ to $\R_{++}$\\

$\nsgameb$ & nonstrategic game\\

$\nsgameg$ & nonstrategic game \\

$\nsgameg\PGd$ & class of $\nsgameg$-potential games\\

$\threedecmap_{\actions}$ & \acl{3DM}, defined in \cref{eq:3-decomposition}\\

$\gradop$ & gradient operator\\

$\gradop^{\play}$ & partial gradient operator\\

$\gradop^{*}$ & $\sum_{\play\in\players} \gradop^{\play*}$\\

$\gradop^{\play*}$ & adjoint of $\gradop^{\play}$\\

$\gammameas^{-\play}\func$ & function in $\Czero$, defined in \cref{eq:eta-minus-i}\\

$\muprof$ & measure on $\actions$\\

$\permut$ & permutation of $\actions$\\

$\reduc{\muprof}$ &  reduced measure vector, defined in \cref{de:reduced}\\ 

$(\mumeas\gammameas)^{\play}$ & measure on $\actions^{\play}$, defined in \cref{eq:mu-eta-i}\\

$\muprof\NoGd$ & class of $\muprof$-normalized games\\

$(\muprof,\nsgameg)\HGd$ & class of $(\muprof,\nsgameg)$-harmonic games\\

$\numeas$ & product measure on $\actions$\\

$\numeas^{\play}$ & measure  on $\actions^{\play}$\\

$\numeas^{-\play}$ & product measure  on $\actions^{-\play}$\\

$\norml{\numeas}^{\play}$ & normalized measure\\

$\potent$ & potential function\\

$\simplex(\setA)$ & set of probability distributions over $\setA$\\

$\inner{\argdot}{\argdot}_{0}$ & inner product of $\Czero$\\

$\inner{\argdot}{\argdot}_{\muprof,\nsgameg}$ & inner product of $\games_{\actions}$\\

$\oplus_{\muprof,\nsgameg}$ & direct orthogonal sum with respect to $\inner{\argdot}{\argdot}_{\muprof,\nsgameg}$.

\end{longtable}

\end{document}